\documentclass[11pt]{article}
\pdfoutput=1

\title{Abelian Hypergroups and Quantum Computation}

 \usepackage{amssymb}
 \usepackage{amsthm}
 \usepackage{authblk}
 \usepackage{mathtools}
 \usepackage{vmargin}

 \usepackage{comment}
 \usepackage{mdwlist}
 \usepackage{enumitem }

\usepackage[hyphens]{url}
\usepackage[colorlinks=true,
            citecolor=magenta,
            linkcolor=blue]
           {hyperref}
           
\usepackage[hyphenbreaks]{breakurl}
 
 \usepackage{sectsty}
 \allsectionsfont{\normalfont\sffamily\bfseries}
 
\usepackage{shadethm}
\usepackage{thmtools}
\usepackage{mdframed}

\definecolor{mybg}{rgb}{0.94,0.94,0.94}

 \theoremstyle{definition}
 
 \newshadetheorem{theorem}{Theorem} 
 \newshadetheorem{lemma}{Lemma}
 \newshadetheorem{comparison}{Comparison}
 \newshadetheorem{conjecture}{Conjecture}
 \newshadetheorem{corollary}{Corollary} 
 \newshadetheorem{claim}{Claim}    
 \newshadetheorem{proposition}{Proposition}
 \newshadetheorem{remark}{Remark}
 \newshadetheorem{definition}{Definition}

 \newmdtheoremenv[backgroundcolor=mybg, %
     innertopmargin =1.5pt ,
     innerbottommargin=1.5pt,
     innerleftmargin=2pt,
     innerrightmargin=2pt,
     leftmargin=-1.5pt,
     rightmargin=-1.5pt,
     linewidth=\topskip, %
     innerrightmargin=0pt, %
     skipabove=\topskip, skipbelow=\topskip,%
     topline=false,bottomline=false,leftline=false,rightline=false]{example}{Example}

 \usepackage[square,numbers,compress]{natbib}

 \usepackage{faktor}

  \newcommand{\abs}[1]{|{#1}|}
  \newcommand{\ket}[1]{|{#1}\rangle}
  \newcommand{\bra}[1]{\langle{#1}|}
  
  \newcommand{\set}[2]{\{{#1}\,|\,{#2}\}}

  \newcommand{\defeq}{\vcentcolon=}

  \newcommand{\Comp}{\mathcal{H}}
  
  \newcommand{{\triv}}{{\rm triv}}
  \newcommand{\Conj}[1]{{\overline{#1}}}
  \newcommand{\Repr}[1]{{\widehat{#1}}}
 
 \newcommand{\Fourier}[1]{\mathcal{F}_{#1}}
 
 \DeclareMathOperator{\Ker}{ker}
  
  \newcommand{\Hchi}{\mathcal{X}}
  \newcommand{\A}{\mathcal{A}}

  \newcommand{\Complex}{\mathbb{C}}
  \newcommand{\Integer}{\mathbb{Z}}

  \def\Z{\mathbb{Z}}
  
  \def\C{\mathbb{C}}

  \newcommand\polylog[1]{\textnormal{polylog$\,#1$}}

  \usepackage[dvipsnames]{xcolor}

 \usepackage{relsize}
 \newcommand{\PX}[1]{ X_{\mathsmaller{#1}} }
 \newcommand{\PZ}[1]{ Z_{\mathsmaller{#1}} }

 \newcommand{\w}[1]{ w_{\mathsmaller{#1}} } 
 \newcommand{\ws}[1]{ w_{\mathsmaller{#1}} } 
 \newcommand{\varw}[1]{ \varpi_{\mathsmaller{#1}} } 
 
 \newcommand{\Integers}{\Z}
 
 \DeclareFontFamily{OT1}{pzc}{}
 \DeclareFontShape{OT1}{pzc}{m}{it}{<-> s * [1.10] pzcmi7t}{}
 \DeclareMathAlphabet{\mathpzc}{OT1}{pzc}{m}{it}

  \newcommand{\myparagraph}[1]{%
    \bigbreak 
    \noindent\textsf{\textbf{#1}}\enspace\ignorespaces}
  \makeatother

 \newcommand{\subhyp}[2]{ \overline{#1}_{#2} }

\title{Abelian Hypergroups and Quantum Computation}

\author[1,3]{Juan Bermejo-Vega\thanks{jbermejovega@gmail.com}}

\author[2]{Kevin C. Zatloukal\thanks{kevinz@mit.edu}}

\affil[1]{\small Max-Planck-Institut f\"ur Quantenoptik, Theory Division, Garching, Germany}

\affil[2]{\small \ Center for Theoretical Physics, Massachusetts Institute of Technology,  Cambridge, MA, USA}

\affil[3]{\small \ Dahlem Center for Complex Quantum Systems, Freie Universit\"{a}t Berlin, 14195 Berlin, Germany}

\begin{document}

\maketitle

\begin{abstract} 
Motivated by a connection, described here for the first time, between the hidden normal subgroup problem (HNSP) and abelian hypergroups (algebraic objects that model collisions of physical particles), we develop a stabilizer formalism using abelian hypergroups and an associated classical simulation theorem (a la Gottesman-Knill). Using these tools, we develop the first provably efficient quantum algorithm for finding hidden subhypergroups of nilpotent abelian hypergroups and, via the aforementioned connection, a new, hypergroup-based algorithm for the HNSP on nilpotent groups. We also give efficient methods for manipulating non-unitary, non-monomial stabilizers and an adaptive Fourier sampling technique of general interest. 

\end{abstract}

\section{Introduction}

\subsubsection*{Background and motivation}

Ever since Shor's groundbreaking discovery of an efficient quantum algorithm
for factoring \cite{Shor}, researchers have strived to understand the source of its quantum speed up and  find new applications for quantum computers. An era of breakthroughs followed, in which researchers  found that factoring and discrete log are instances of the so-called  \emph{Hidden Subgroup Problem} (HSP), a more general problem about \emph{finite groups};  developed efficient quantum algorithms for the abelian group HSP \cite{kitaev_phase_estimation,Kitaev97_QCs:_algorithms_error_correction,Brassard_Hoyer97_Exact_Quantum_Algorithm_Simons_Problem,Hoyer99Conjugated_operators,MoscaEkert98_The_HSP_and_Eigenvalue_Estimation,mosca_phd,cheung_mosca_01_decomp_abelian_groups,Damgard_QIP_note_HSP_algorithm}; and discovered that solving the nonabelian group HSP over symmetric and dihedral groups would lead to a revolutionary algorithm  for  Graph Isomorphism \cite{Ettinger99aquantum} and break lattice-based cryptography \cite{Regev:2004:QCL:976327.987177}.

Motivated by these breakthroughs, there has been a great deal of research work over the last decade aimed at finding efficient quantum algorithms for nonabelian HSPs, leading to many  successes
\cite{Hallgren00NormalSubgroups:HSP,EttingerHoyerKnill2004_Hidden_Subgroup,Kuperberg2005_Dihedral_Hidden_Subgroup,Regev2004_Dihedral_Hidden_Subgroup,Kuperberg2013_Hidden_Subgroup,RoettelerBeth1998_Hidden_Subgroup,IvanyosMagniezSantha2001_Hidden_Subgroup,MooreRockmoreRussellSchulman2004,InuiLeGall2007_Hidden_Subgroup,BaconChildsVDam2005_Hidden_Subgroup,ChiKimLee2006_Hidden_Subgroup,IvanyosSanselmeSantha2007_Hidden_Subgroup,MagnoCosmePortugal2007_Hidden_Subgroup,IvanyosSanselmeSantha2007_Nil2_Groups,FriedlIvanyosMagniezSanthaSen2003_Hidden_Translation,Gavinsky2004_Hidden_Subgroup,ChildsVDam2007_Hidden_Shift,DenneyMooreRussel2010_Conjugate_Stabilizer_Subgroups,Wallach2013_Hidden_Subgroup,lomont_HSP_review,childs_lecture_8,VanDamSasaki12_Q_algorithms_number_theory_REVIEW}, though  efficient quantum algorithms for dihedral and symmetric HSP have still not been found.

Thus far, the foundation of nearly all known  quantum algorithms for nonabelian HSPs has been the seminal work of Hallgren, Russell, and Ta-Shma \cite{Hallgren00NormalSubgroups:HSP}, which showed that hidden \emph{normal} subgroups can be found efficiently  for \emph{any} nonabelian group. For example, the algorithms for (near) Hamiltonian groups \cite{Gavinsky:2004:QSH:2011617.2011625} work because all subgroups of such groups are (nearly) normal. Likewise, the sophisticated algorithm of Ivanyos et al. for 2-nilpotent groups \cite{Ivanyos:2008:EQA:1792918.1792983} cleverly reduces the problem of finding a hidden non-normal subgroup to
two problems of finding hidden normal subgroups. 

Surprisingly, given the importance of the nonabelian HSP program in the  history of quantum computing, the success of the quantum algorithm for the hidden \emph{normal} subgroup problem (HNSP) \cite{Hallgren00NormalSubgroups:HSP} remains poorly explained. The initial motivation for this work was improve our understanding of the quantum algorithm for the HNSP up to the same level as those for abelian HSPs.

Our approach was inspired by a  (fairly unexpected) recently discovered  connection by Bermejo-Vega-Lin-Van den Nest \cite{BermejoLinVdN13_BlackBox_Normalizers}, between Shor's algorithm and two other foundational  results in quantum computation, namely, Gottesman's \emph{Pauli stabilizer formalism} (PSF) \cite{Gottesman_PhD_Thesis},  widely used in quantum error correction, and the Gottesman-Knill theorem \cite{Gottesman_PhD_Thesis,Gottesman99_HeisenbergRepresentation_of_Q_Computers,Gottesman98Fault_Tolerant_QC_HigherDimensions}, which proves the efficient classical simulability of Clifford circuits. In short, the {BVLVdN connection} states that the quantum algorithms for \emph{abelian} HSPs  belong to a common family of highly structured quantum circuits built of \emph{normalizer gates over abelian groups} \cite{VDNest_12_QFTs,BermejoVega_12_GKTheorem,BermejoLinVdN13_Infinite_Normalizers}  (quantum Fourier transforms, group automorphism gates, and quadratic phase gates). This fact, combined with the generalized Group Stabilizer Formalism (GSF) for simulating normalizer circuits  \cite{VDNest_12_QFTs,BermejoVega_12_GKTheorem,BermejoLinVdN13_Infinite_Normalizers}, was used to prove a sharp \emph{no-go theorem} for finding new quantum algorithms with the standard abelian group Fourier sampling techniques.

Given the success of \cite{BermejoLinVdN13_BlackBox_Normalizers} at understanding abelian HSP quantum algorithms  using an abelian group stabilizer formalism, our aim in this work is to gain a deeper understanding of the algorithm for HNSPs on nonabelian groups using a more sophisticated stabilizer formalism. Furthermore, because the PSF (and generalizations)  have  seminal applications in fault tolerance \cite{Gottesman98Fault_Tolerant_QC_HigherDimensions,BravyiKitaev05MagicStateDistillation}, measurement based quantum computation \cite{raussen_briegel_onewayQC}, and condensed matter theory \cite{kitaev_anyons}, we expect a new stabilizer formalism to  find new uses outside of  quantum algorithm analysis.

\subsubsection*{Main results}

While it would be natural to generalize the abelian group stabilizer formalisms into a nonabelian group stabilizer formalism, we find that the proper way to understand the quantum algorithm for the HNSP is not to generalize the ``abelian'' property but rather the ``group'' property. In particular, we will work with \emph{abelian hypergroups}. These are generalizations of groups and can be thought of as collections of particles (and anti-particles) with a ``collision'' operation that creates new particles. A group is a special case of a hypergroup where each collision produces exactly one resulting particle. 

Our first result is a formal connection between the HNSP and abelian hypergroups which will be helpful to understand why quantum computers can solve this problem:
\begin{enumerate}
\item[I.] \textbf{Connecting the HNSP to the abelian HSHP.} We demonstrate (\textbf{section \ref{sect:HNSP}})	that, in many natural cases, the HNSP can be reduced to a problem on abelian hypergroups, called the hidden \emph{subhypergroup} problem (HSHP) \cite{Amini_hiddensub-hypergroup,Amini2011fourier}. This occurs because all of the information about the normal subgroups of a nonabelian group is captured in its hypergroup of conjugacy classes. Even in a nonabelian group, there is a multiplication operation on conjugacy classes that remains abelian. Our results show that, in many natural cases, finding hidden normal subgroups remains a problem about an abelian algebraic structure even when the group is nonabelian.
\end{enumerate}
Our next results show that the tools that proved successful for understanding quantum algorithms for abelian group HSPs (as well as many other problems) can be generalized to the setting of abelian hypergroups:
\begin{enumerate}

\item[II.] \textbf{A hypergroup stabilizer formalism.} We extend the PSF \cite{Gottesman_PhD_Thesis,Gottesman99_HeisenbergRepresentation_of_Q_Computers,Gottesman98Fault_Tolerant_QC_HigherDimensions}, a powerful tool for describing quantum many-body states  as eigenstates of commuting \emph{groups} of Pauli operators, and the GSF, the abelian group extension of \cite{VDNest_12_QFTs,BermejoVega_12_GKTheorem}, to a stabilizer formalism using commuting \emph{hypergroups} of generalized Pauli operators. The latter are no longer unitary nor monomial but still exhibit rich Pauli-like features that let us manipulate them with (new) hypergroup techniques and are normalized by associated Clifford-like gates. We also provide a normal form for  hypergroup stabilizer states (\textbf{theorem \ref{thm:Normal Form CSS States}}) that are CSS-like \cite{CalderbankShor_good_QEC_exist,Steane1996_Multiple_Particle_Interference_QuantumErrorCorrection,Calderbank97_QEC_Orthogonal_Geometry} in our setting.

\item[III.] \textbf{A hypergroup Gottesman-Knill theorem.} We introduce models of  \emph{normalizer circuits over abelian hypergroups}, which contain hypergroup quantum Fourier transforms (QFTs) and other entangling gates.  Our models extend the known families of Clifford circuits \cite{Gottesman_PhD_Thesis,Gottesman99_HeisenbergRepresentation_of_Q_Computers,Gottesman98Fault_Tolerant_QC_HigherDimensions} and (finite) abelian group normalizer circuits  \cite{VDNest_12_QFTs,BermejoVega_12_GKTheorem,BermejoLinVdN13_BlackBox_Normalizers}. We show (\textbf{theorem \ref{thm:Evolution of Stabilizer States}}) that the dynamical evolution of such circuits can be tracked in our hypergroup stabilizer picture and,  furthermore, that for large hypergroup families (including  products $\mathcal{T}^m$ of  constant size hypergroups), many hypergroup  normalizer circuits can be efficiently simulated classically (\textbf{theorem \ref{thm:simulation}}).\footnote{Here, we rely  on computability assumptions (section \ref{sect:Assumptions on Hypergroups})  that are always fulfilled in \cite{VDNest_12_QFTs,BermejoVega_12_GKTheorem}.}

\end{enumerate}
	We complete our analysis of the HNSP, which we reduced to the abelian HSHP (result I.),  showing that our normalizer circuit model encompasses an earlier HSHP quantum algorithm based on a variant of Shor-Kitaev's quantum phase estimation, which was proposed by  not fully analyzed by  Amini-Kalantar-Roozbehani in \cite{Amini_hiddensub-hypergroup,Amini2011fourier}. Using our stabilizer formalism,  we prove the latter to be  inefficient on easy instances, and, thereby, point out the  abelian HSHP as the \emph{first} known commutative hidden substructure problem in quantum computing that \emph{cannot} be solved by standard phase estimation. In spite of this  no-go result, we also show, in our last main contribution, that in the interesting cases from the nonabelian HSP perspective, the abelian HSHP can actually be solved with a novel \emph{adaptive/recursive} quantum Fourier sampling approach:
\begin{enumerate}
\item[IV.] \textbf{New quantum algorithms.} We present the first provably efficient quantum algorithm for finding hidden subhypergroups of \emph{nilpotent}\footnote{These are conjugacy class hypergroups associated to \emph{nilpotent groups} \cite{Humphrey96_Course_GroupTheory}. The latter form a \emph{large} group class that includes abelian groups, Pauli/Heisenberg groups over $\Z_{p^r}$ with prime $p$, dihedral groups $D_{2N}$ with $N=2^n$, groups of prime-power order and their direct products.} abelian hypergroups, provided we have efficient circuits for the required QFTs. This algorithm also leads, via the connection above (result I.), to a new efficient quantum algorithm for the HNSP over nilpotent groups that directly exploits the abelian hypergroup structure and is fundamentally different from the algorithm of Hallgren et al. \cite{Hallgren00NormalSubgroups:HSP}.
\end{enumerate}
Our correctness proofs for these last quantum algorithms  can further be extended to  crucial non-nilpotent groups\footnote{We give another algorithm that works for all groups  under some additional mild assumptions.} (and their associated class hypergroups) such as  the dihedral and symmetric groups.\footnote{For dihedral groups/hypergroups we give a quantum algorithm; for symmetric ones, a \emph{classical} one already does the job  because symmetric groups/hypergroups have few normal subgroups/subhypergroups. } In contrast, no efficient quantum algorithm for the nilpotent, dihedral and symmetric HSPs is known. This provides strong evidence that abelian HSHP  is a much \emph{easier} problem for quantum computers than  nonabelian HSP, and, because of its Shor-like  connection with a stabilizer formalism, perhaps even a more \emph{natural} one.

\subsubsection*{Applications} Though  lesser known than nonabelian groups, abelian hypergroups  have a wide range of applications in  convex   optimization (cf.\ association schemes \cite{KlerkSDP_AssociationSchemes,anjos2011handbook}), classical cryptography, coding theory \cite{corsini2003applications},  particle physics \cite{Dehhan_Hypergrou_Particles},  conformal field theory   \cite{Wildberger1994HypergroupsApplications}. In  topological quantum computation \cite{kitaev_anyons}, fusion-rule hypergroups  \cite{Kitaev2006_Anyons_Exactly_Solved_Model}  are indispensable  in the study of nonabelian anyons  \cite{Kitaev2006_Anyons_Exactly_Solved_Model}. Our stabilizer formalism over the latter hypergroups likely has applications for quantum error correction and for the simulation of protected gates over topological quantum field theories \cite{Beverland14_ProtectedGates_Topological}. 

The stabilizer formalism and classical simulation techniques presented in this paper are unique in that they are the first and only available methods  to manipulate stabilizer operators that neither \emph{unitary}, nor \emph{ monomia}l, nor \emph{sparse} that we are aware of \cite{nest_MMS}. Furthermore, our stabilizer formalism yields the first known families of qubit/qudit stabilizer operators for any arbitrary finite dimension $d$ that are not the standard Weyl-Heisenberg operators \cite{Gottesman98Fault_Tolerant_QC_HigherDimensions}, with associated normalizer gates that are \emph{not} the standard qudit Clifford gates. Additionally, our methods  allow great flexibility to construct new codes because the stabilizer families can be chosen over any hypergroup of interest.

\subsubsection*{Relationship to prior work}

Normalizer circuits associated to   \emph{abelian groups} have been extensively studied in earlier works of  Van den Nest, Bermejo-Vega and Lin \cite{VDNest_12_QFTs,BermejoVega_12_GKTheorem,BermejoLinVdN13_Infinite_Normalizers,BermejoLinVdN13_BlackBox_Normalizers}: in \cite{VDNest_12_QFTs,BermejoVega_12_GKTheorem}, the groups were finite and given in an explicitly \emph{decomposed} form $\Integers_{N_1}\times \cdots \times \Integers_{N_m}$;  \cite{BermejoLinVdN13_Infinite_Normalizers} further considered infinite group factors\footnote{Though we have not considered infinite hypergroups, many of our results should extend to locally compact abelian hypergroups (see \cite{Amini2011fourier} and the discussion in \cite{BermejoLinVdN13_Infinite_Normalizers} about locally compact abelian groups).} (integers, hypertori) and infinite dimensional quantum gates;  \cite{BermejoLinVdN13_BlackBox_Normalizers} added matrix group factors (like $\Integers_N^\times$) and black box groups \cite{BermejoLinVdN13_BlackBox_Normalizers}. The circuits in \cite{VDNest_12_QFTs,BermejoVega_12_GKTheorem,BermejoLinVdN13_Infinite_Normalizers} were proven to be efficiently  simulable by classical computers using abelian group stabilizer formalisms. Those in \cite{BermejoLinVdN13_BlackBox_Normalizers} were shown to be powerful enough to implement Shor's and  abelian HSP quantum algorithms. Moreover, simulating them was shown to be at least as hard as factoring and \emph{exactly} as hard as decomposing finite abelian groups.

Our efficient classical simulation result (\textbf{theorem \ref{thm:simulation}}) is not a full generalization of the Gottesman-Knill theorem \cite{Gottesman_PhD_Thesis,Gottesman99_HeisenbergRepresentation_of_Q_Computers}, like the one developed in \cite{VDNest_12_QFTs,BermejoVega_12_GKTheorem}, but of its CSS-preserving variant \cite{Delfosee14_Wigner_function_Rebits} without intermediate measurements. In our work, we dedicate most effort to cope with the highly nontrivial difficulty that our Pauli operators are \emph{non-monomial} and \emph{non-unitary}, which renders \emph{all} existing stabilizer formalism techniques \cite{Gottesman_PhD_Thesis,Gottesman99_HeisenbergRepresentation_of_Q_Computers,Gottesman98Fault_Tolerant_QC_HigherDimensions,VDNest_12_QFTs,BermejoVega_12_GKTheorem,BermejoLinVdN13_Infinite_Normalizers,BermejoLinVdN13_BlackBox_Normalizers, AaronsonGottesman04_Improved_Simul_stabilizer,dehaene_demoor_coefficients,dehaene_demoor_hostens,VdNest10_Classical_Simulation_GKT_SlightlyBeyond,deBeaudrap12_linearised_stabiliser_formalism,nest_MMS,NiBuerschaperVdNest14_XS_Stabilizers} inapplicable. To tackle this issue, we develop new simulation techniques based on hypergroup methods, up to a fairly mature state, though further improvement remains possible (see\ \textbf{section \ref{sect:Simulation}} for a discussion and a related conjecture). 

The Group Stabilizer Formalism of \cite{VDNest_12_QFTs,BermejoVega_12_GKTheorem,BermejoLinVdN13_Infinite_Normalizers} was applied in \cite{BermejoLinVdN13_BlackBox_Normalizers} to show that the problem of decomposing finite abelian black-box groups, which can be solved by a quantum algorithm due to Cheung and Mosca \cite{mosca_phd,cheung_mosca_01_decomp_abelian_groups},  is  complete   for the complexity class associated to black-box group normalizer circuits; thereby, proving no-go theorem for finding new quantum algorithm better than Cheung-Mosca's. The authors of \cite{BermejoLinVdN13_BlackBox_Normalizers} raised the open question of whether normalizer circuits over \emph{different algebraic structures} could be found and be used to bypass their no-go theorem. Our quantum algorithms for abelian HSHPs answer their question in the affirmative: the circuits we use to solve that problem are instances of normalizer circuits over nonabelian groups / hypergroups (\textbf{section \ref{sect:Quantum Algorithms}}).

The  hidden subhypergroup problem (HSHP) we discuss was first considered by Amini, Kalantar and Roozbehani in \cite{Amini_hiddensub-hypergroup,Amini2011fourier}, yet (to the best of our knowledge) no  provably efficient quantum algorithm for this problem has been given before. We show that an earlier quantum algorithm proposed in \cite{Amini2011fourier} for solving the problem using a variant of Shor-Kitaev's quantum phase estimation \cite{kitaev_phase_estimation} is inefficient on easy instances (section \ref{sect:Quantum Algorithms}). Interestingly, this means that abelian HSHP is the first known \emph{commutative}  hidden substructure problem that cannot be solved by standard phase estimation. Instead, our quantum algorithm is based on  a novel \emph{adaptive/recursive Fourier sampling} quantum approach.

In this work,  classical simulation techniques played a role in the development of the quantum algorithms we present. In this way, our results relate to other projects where classical simulations methods helped to find new quantum algorithms \cite{VdNest_Q_alg_spin_models_simulable_gate_sets,Ni13Commuting_Circuits} or  complexity theoretic hardness results \cite{Aaronson11_Computational_Complexity_Linear_Optics,BremnerJozsaShepherd08,BermejoLinVdN13_BlackBox_Normalizers,MorimaeFujiiFitzsimons14_Hardness_Simulating_DQC1}.

For any non-abelian group $G$, the simulation results we present lead  to efficient  classical algorithms for simulating quantum Fourier transforms over $G$ (specifically,  as employed in weak Fourier sampling routines) acting on coset states $\ket{aH}$, $a\in G$, $H\subset G$ such that $aH$ is invariant under conjugation\footnote{This happens, e.g., if $aH=N$ for normal $N$ or if the subgroup $H$ contains the derived subgroup $[G,G]$.}. In this sense, our work connects with \cite{bermejo2011classical}, where efficient classical algorithms  were given for simulating weak and strong quantum Fourier sampling on arbitrary coset states of  semi-direct products group $\Z_p\ltimes A$, where $p$ is  prime  and $A$ is an  abelian group given in a canonical  form $\Z_{N_1}\times \cdots \times \Z_{N_m}$.

Finally, we mention that the results in \cite{VDNest_12_QFTs,BermejoVega_12_GKTheorem,BermejoLinVdN13_Infinite_Normalizers,BermejoLinVdN13_BlackBox_Normalizers} and the results in the present work are not contradictory.  Briefly, the efficient simulations in \cite{VDNest_12_QFTs,BermejoVega_12_GKTheorem,BermejoLinVdN13_Infinite_Normalizers} are possible because more knowledge about the normalizer circuits  structure is given (namely, cyclic factor decompositions of the associated finite groups), but in  \cite{BermejoLinVdN13_BlackBox_Normalizers} this information is missing and normalizer circuits become useful to identify  hidden structures  (cf. reference for extended discussion). Similarly, in our work,  some hypergroups with stronger computability properties lead to efficiently classically simulable circuits, while other lead to valuable quantum algorithms (see section \ref{sect:Assumptions on Hypergroups} and appendix \ref{app:Discret Log}).

\subsubsection*{Structure of the paper}

We give a non-technical introduction to the theory of hypergroups in \textbf{section \ref{sect:Hypergroups}}. We then re-introduce the hidden normal subgroup problem (HNSP) and prove its connection to the hidden subhypergroup problem (HSHP) in \textbf{section \ref{sect:HNSP}}. We present our models of hypergroup normalizer circuits, our hypergroup stabilizer formalism, and our simulation results in \textbf{sections \ref{sect: Circuit Model}-\ref{sect:Simulation}} and describe these on some examples. Finally, we use these tools to develop new quantum algorithms for abelian HSHP and HNSP in \textbf{section \ref{sect:Quantum Algorithms}}.

While our motivation for developing our hypergroup stabilizer formalism was to understand more about the HNSP, we note that the results of sections \ref{sect: Circuit Model}--\ref{sect:Simulation} are more general, as they apply to arbitrary hypergroups. We expect that these tools will have applications outside of the analysis of quantum algorithms such as to the development of new error correcting codes.

\section{Abelian hypergroups and hypergroup duality}\label{sect:Hypergroups}

This section is an introduction for quantum computer scientists to the beautiful theory of  \emph{finite abelian hypergroups}\footnote{The hypergroups we consider are frequently called ``finite commutative hypergrops'' in mathematics. We call them ``abelian'' because of the focus of this work on abelian and nonabelian HSPs. In some of our references \cite{Roth75_Character_Conjugacy_Hypergroups,McMullen79_Algebraic_Theory_Hypergroups,McMullen_Duality_abelian_Hypergroups}, the hypergroups in this work are called ``\emph{reversible abelian hypergroups}''.}, whose origin dates back to works by Dunkl \cite{Dunkl1973}, Jewett \cite{Jewett19751}, Spector \cite{Spector1978} in the  70s. Our account is based on \cite{Roth75_Character_Conjugacy_Hypergroups,McMullen79_Algebraic_Theory_Hypergroups,McMullen_Duality_abelian_Hypergroups,Wildberger97_Duality_Hypergroups,Wildberger_Lagrange,Wildberger1994HypergroupsApplications,Wildberger2001algebraic,BloomHeyer95_Harmonic_analysis} and borrows most notation and terminology from \cite{Wildberger1994HypergroupsApplications,Wildberger97_Duality_Hypergroups,Wildberger_Lagrange,Wildberger2001algebraic}. Throughout the paper, hypergroups and groups are assumed to be \emph{finite} unless said otherwise.

In brief, abelian hypergroups are algebraic structures that generalize  abelian  groups, although   in a  \emph{different} way than nonabelian groups. Despite being less known than the latter, abelian hypegroups   have a wide number of applications in multiple fields, including  combinatorics, convex optimization \cite{KlerkSDP_AssociationSchemes,anjos2011handbook}; cryptography, classical error correction \cite{corsini2003applications}; classical information theory \cite{Wildberger1994HypergroupsApplications};  particle physics \cite{Dehhan_Hypergrou_Particles} and conformal field theory  \cite{Wildberger1994HypergroupsApplications}, to begin with. In  topological quantum computation \cite{kitaev_anyons}, certain hypergroups known by the names of ``fusion theories or categories'' \cite{Wildberger1994HypergroupsApplications,Kitaev2006_Anyons_Exactly_Solved_Model}  are invaluable in the study of topological order and nonabelian anyons  \cite{Kitaev2006_Anyons_Exactly_Solved_Model}. 

On top of their versatility, abelian hypergroups also admit a simple and intuitive \emph{physical} definition, which we give now before going into the full mathematical details of their theory. In simple terms, a \emph{finite abelian hypergroup} $\mathcal{T}$ is a set of particle types $\{x_0, x_1,
\dots, x_n\}$ that can collide. When $x_i$ collides with $x_j$ a particle $x_k$ is created with probability $n_{ij}^k$. Furthermore, a non reactive \emph{vacuum} particle $x_0$ will be created with non-zero probability by such process iff $x_i$  is the antiparticle of $x_k$ (which always exists).

\subsection{Definition}\label{sect:DefinitionsHypergroup}

\newcommand{\B}{\mathcal{B}}

We now turn the  intuitive definition of hypergroup above into a precise mathematical one.

A finite abelian hypergroup $\mathcal{T} = \{x_0, x_1,
\dots, x_n\}$ is a basis of a commutative complex $C^*$ algebra $\mathcal{A}(\mathcal{T})=\C \mathcal{T}$, called the \emph{hypergroup algebra} of $\mathcal{T}$, with a particular structure. $\mathcal{A}(\mathcal{T})$ is endowed with an associative commutative \emph{hyper}operation 
\begin{equation}
\label{eq:Hyperoperation}
x_i x_j = \sum_{k=0}^n n_{ij}^k x_k \ \ \ \forall x_i, x_j \in \mathcal{T},
\end{equation}
which returns a superposition of outcomes in $\mathcal{T}$ (we write ``$x_k\in x_ix_j$'' when $x_k$ is a possible outcome of $x_ix_j$, with $n_{ij}^k\neq 0$); a multiplicative identity $x_0=1$; and an involution $x_i \rightarrow \overline{x}_i$. Note that commutativity and the presence of the involution imply that  $n_{\overline{a},\overline{b}}^{\overline{c}}=n_{ba}^c=n_{ab}^c $ holds for any $a,b,c\in\mathcal{T}$.

Furthermore, the ``structure constants'' $n_{ij}^k \ge 0$ are \emph{real} numbers with three properties:
\begin{itemize}
\item[(i)] \textbf{Anti-element property.}  For every $x_i$ and any $x_j$, the identity $x_0=1$ can be an outcome of $x_ix_j$ if and only if $x_j = \overline{x}_i$. We call  $\overline{x}_i$  the \emph{anti-element} of $x_i$.
\item[(ii)] \textbf{Normalization property.} For all values of  $k=0,\ldots, n$ we have $\sum_{k=0}^n
  n_{i,j}^k = 1$;  in other words, $n_{ij}^k$ is a probability distribution (of outcomes) over $k$.

\item[(iii)]\textbf{Reversibility.}\footnote{This last property (iii) and (\ref{eq:Reversibility Property}) can both be derived from the previous axioms \cite{McMullen79_Algebraic_Theory_Hypergroups}.} For every $x,y,z\in\mathcal{T}$, it holds that $z\in xy$ if and only if $y\in\overline{x}z$. Moreover, if the \emph{weight of $x$} is defined as $\w{x}:=1/n_{x\overline{x}}^0$, the following identity holds:
\begin{equation}\label{eq:Reversibility Property}
\frac{n_{xy}^z}{\w{z}}=\frac{n_{\overline{x}z}^y}{\w{y}}=\frac{n_{z\overline{y}}^x}{\w{x}}
\end{equation}
\end{itemize} 

As a simple example,  any finite abelian group $G$ is an abelian hypergroup. The elements of $G$ define the basis of the group algebra $\Complex G$ and the
involution is $\overline{x} := x^{-1}$. In the case of a group, though, for any
$i, j \in \Integer_{n+1}$, there is only a single nonzero $n_{i,j}^k$ since
$x_i x_j = x_k$ for some $k$; though hypergroups have a more complicated
multiplication than groups,  they preserve the property that the product of
$x$ and  $\overline{x}$ includes the identity.

\myparagraph{Hypergroups in this work.} Though nonabelian hypergroups exist\footnote{In fact, every nonabelian group $G$ is also a kind of nonabelian hypergroup.}, this paper focuses on  \emph{abelian} ones because they  fulfill certain useful dualities (see below). In sections \ref{sect:HNSP} and \ref{sect:Quantum Algorithms},  we further focus on specific abelian hypergroups that arise from \emph{finite groups} (section \ref{sect:Examples Hypergroups}). 

\subsection{Glossary}\label{sect:Glossary}
\label{sect:Orthogonality and QFT}\label{sect:Subhypergroups, Quotients}

We now give a glossary of hypergroup theoretic concepts for future reference. In all definitions below $\mathcal{T}$ is fixed to be an arbitrary \emph{abelian} finite hypergroup. 

\myparagraph{Weight functions.}Every subset $X\subset\mathcal{T}$ has a \emph{weight} $\varpi_X:=\sum_{x\in X}w_x$, with $w_x$ as in (\ref{eq:Reversibility Property}).

\myparagraph{Subhypergroup.}{A \emph{subhypergroup $\mathcal{N}$}}   is a subset of $\mathcal{T}$ that is also a hypergroup with the same identity, involution, structure constants and weights.

\myparagraph{Quotient hypergroup.}For any  subhypergroup $\mathcal{N}$ the {\emph{quotient hypergroup}} $\mathcal{T/N}$  is an abelian hypergroup whose elements are the cosets $a\mathcal{N}:=\{x\in\mathcal{T}:x\in a b \text{ for some }b\in \mathcal{N}
\}$. Its  hyperoperation is  defined \cite{McMullen79_Algebraic_Theory_Hypergroups,Wildberger_Lagrange} by, first, identifying each $a\mathcal{N}$ with an element of the  $\mathcal{A}(\mathcal{T})$ algebra\footnote{See \cite{Roth75_Character_Conjugacy_Hypergroups} for a set theoretic definition} via $a\mathcal{N} :=\sum_{x\in a\mathcal{N}}\w{x} x /\varpi_{a\mathcal{N}}$. Then, $\mathcal{T/N}$ inherits a hyperoperation with  structure constants $r_{a\mathcal{N},b\mathcal{N}}^{c\mathcal{N}} =\sum_{d\in c\mathcal{N}}n_{ab}^{d}$ and weights $\ws{a\mathcal{N}}=1/(\sum_{b\in\mathcal{N}}n_{x,y}^{b})=\varw{a\mathcal{N}}/\varw{\mathcal{N}}$.

\myparagraph{Morphisms.} A map between two hypergroups $f:\mathcal{T}\rightarrow\mathcal{T'}$ is a \emph{homomorphism} if $f(ab)=f(a)f(b)=\sum_{c} n_{ab}^c f(c)$ and $f(\overline{a})=\overline{f(a)}$. An invertible homomorphism is an \emph{isomorphism}. An isomorphism from $\mathcal{T}$ to $\mathcal{T}$ is an \emph{automorphism}. As with groups, isomorphic hypergroups have identical hypergroup-theoretic properties (weights, subhypergroups, etc.).

\myparagraph{Character hypergroup $\mathcal{T^*}$.} A complex function $\mathcal{X}_\mu : \mathcal{T} \rightarrow \Complex$ is a \emph{character of $\mathcal{T}$} if it is not identically zero and satisfies the identity\footnote{If characters are linearly extended to act on the hypergroup algebra $\mathcal{A(T)}$, condition (\ref{eq:Character DEFINITION}) becomes  $\mathcal{X}_\mu(ab) =
\mathcal{X}_\mu(a) \mathcal{X}_\mu(b)$, $\forall a,b \in \mathcal{A(T)}$; in other words, the characters of $\mathcal{T}$ are also the characters of  $\mathcal{A(T)}$.}
\begin{equation}\label{eq:Character DEFINITION}
\mathcal{X}_\mu(ab)=\mathcal{X}_\mu(a)\mathcal{X}_\mu(b)= \sum_{c}n_{ab}^c \mathcal
X(c)\quad\textnormal{and}\quad \mathcal{X}_\mu(\overline{a})=\overline{\mathcal{X}_\mu}(a)\quad\textnormal{for all $a,b\in\mathcal{T}$}.
\end{equation}
For any abelian hypergroup $\mathcal{T}$, its set  $\mathcal{T^*}$ of character functions defines an  abelian \emph{signed} hypergroup with  the point-wise functional product as hyperoperation, the trivial character $\Hchi_1(a)=1$ as identity and  the complex conjugate map $\mathcal{X}_\mu\rightarrow \overline{\mathcal{X}_\mu}$ as involution: here, ``signed''  means that $\mathcal{T^*}$ fulfills (i-ii-iii) but may have some negative structure constants $m_{\mu\nu}^\gamma$, which represent negative probabilities. If all constants $m_{\mu\nu}^\gamma$ are non-negative, $\mathcal{T}^*$ is a hypergroup called the \emph{character hypergroup of $\mathcal{T}$}, and $\mathcal{T}$ is said to be \emph{strong} \cite{BloomHeyer95_Harmonic_analysis}. Throughout the paper, we assume all hypergroups  to be strong (without notice) so that the associated character hypergroups $\mathcal{T}^*$ define  new ``dual theories'' of particle collisions\footnote{Many of the hypergroup concepts and properties presented in this section as well as our results in sections \ref{sect: Circuit Model}, \ref{sect:Simulation} can be effortlessly extended to the setting where $\mathcal{T}$ is an {abelian signed hypergroup}, in which case $\mathcal{T}^*$ is also an abelian signed hypergroup  \cite{McMullen79_Algebraic_Theory_Hypergroups,McMullen_Duality_abelian_Hypergroups,BloomHeyer95_Harmonic_analysis,Ichihara_thesis_Hypergroup_Extensions,Yamanaka2013_Thesis}). Though it seems plausible, we have not investigated whether our results  in section \ref{sect:Stabilizer Formailsm} can be extended to signed hypergroups. In the remaining sections, we focus on class and character hypergroups that arise from finite groups, which are always strong.}

\myparagraph{Weight-order duality.} The hypergroups $\mathcal{T}$ and $\mathcal{T^*}$ have the same cardinalities and weights:
\begin{equation}\label{eq:Weight Duality}
\varw{\mathcal{T}}=\sum_{a\in\mathcal{T}}\w{a} = \sum_{\mathcal{X}_\mu\in\mathcal{T}^*}\w{\mathcal{X}_\mu}=\varw{\mathcal{T}^*}.
\end{equation}

\myparagraph{Abelian hypergroup duality.}\label{sect:Duality} The  hypergroup  $\mathcal{T}^{**}$ of characters of $\mathcal{T}^*$ is \emph{isomorphic} to the original hypergroup $\mathcal{T}$. This isomorphism is constructed canonically by sending $a \in \mathcal{T}$ to a character 
\begin{equation}\label{eq:Hypergroup Duality}
\widetilde{\mathcal{X}_a}(\mathcal{X}_\mu)=\overline{\Hchi_{\mu}}(a).
\end{equation} In particular, this shows that  the hypergroups $\mathcal{T}$, $\mathcal{T}^*$  have the same number of elements.

\begin{remark}[\textbf{Notation}] Throughout the text,  we identify dual characters $\widetilde{\mathcal{X}}_{a}\in\mathcal{T}^{**}$ with elements $a\in\mathcal{T}$ via the isomorphism (\ref{eq:Hypergroup Duality}). We write the hyperoperation of $\mathcal{T}^*$ compactly as $\Hchi_\mu \Hchi_\nu = \sum_{\gamma} m_{\mu\nu}^{\gamma} \Hchi_\gamma$ and, occasionally, use the expression $\Hchi_{\overline{\mu}}$ as a shorthand for $\overline{\Hchi_\mu}$.
\end{remark}
The notions of character and duality lead to a family of related concepts that are extremely valuable in hypergroup theory and in the present work:

\myparagraph{Annihilators.} The \emph{annihilator $\mathcal{N^\perp}$} of a subhypergroup $\mathcal{N\subset T}$  is a  subhypergroup  of $\mathcal{T^*}$
\begin{equation}\label{eq:annhilator}
\mathcal{N}^\perp := \{\mathcal{X}_\mu \in \mathcal{T}^*: \mathcal{X}_\mu(a)=1 \textnormal{ for all } a\in \mathcal{N}\}.
\end{equation} 
\myparagraph{Subhypergroup duality} 
A stronger form of hypergroup duality relates the notions of annihilator, subhypergroup and quotient: the annihilator $\mathcal{N}^\perp$ is \emph{isomorphic} to the characters $(\mathcal{T/N})^*$ of  $\mathcal{T/N}$; moreover, the character hypergroup $\mathcal{N^*}$ is isomorphic to $\mathcal{T^*/\mathcal{N}^\perp}$. 

\myparagraph{Character orthogonality.} Character functions are orthogonal with the inner product
\begin{equation}\label{eq:Character Orthogonality}
\langle \Hchi_{\mu},\Hchi_\nu\rangle =
\sum_{a\in\mathcal{T}} \frac{\ws{\mathcal{X}_\nu}\w{a}}{\varpi_{\mathcal{T}}} \overline{\Hchi_{\mu}}(a)\Hchi_\nu(a)=\delta_{\mu,\nu}.
\end{equation}
Moreover, due to hypergroup and subhypergroup duality, for any subhypergroup $\mathcal{N}\subset\mathcal{T}$, any two cosets $a\mathcal{N},b\mathcal{N}\in\mathcal{T/N}$ and any $\mathcal{X_\mu,X_\nu\in \mathcal{N}^\perp}$, the following generalized orthogonality relationships are always fulfilled
\begin{equation}\label{eq:Character Orthogonality Subhypergroups}
\sum_{a\in \mathcal{N}} \frac{\ws{\Hchi_\nu\mathcal{N}^\perp}\w{a} }{\varpi_\mathcal{N}} \overline{\Hchi_{\mu}}(a) \mathcal{X}_\nu(a) =\delta_{\mathcal{X}_\mu \mathcal{N}^\perp ,\mathcal{X}_\nu \mathcal{N}^\perp} ,\quad   \sum_{\mathcal{X}_\mu \in \mathcal{N}^\perp} \frac{\ws{b\mathcal{N}} \w{\mathcal{X}_\mu}}{{\varpi_{\mathcal{N}^\perp}}} \overline{\Hchi_{\mu}}(a) \mathcal{X}_\mu(b) ={\delta_{a\mathcal{N},b\mathcal{N} }},
\end{equation}

\subsection{Examples from group theory}
\label{sect:Examples Hypergroups}

We now introduce two  examples of hypergroups that play a central role in our work (namely, in sections \ref{sect:HNSP} and \ref{sect:Quantum Algorithms}).  For an arbitrary finite  group, these are the hypergroups of conjugacy classes and of characters, which are dual to each other in the sense of (\ref{eq:Hypergroup Duality}). The existence of these hypergroups linked to arbitrary groups lets us apply our hypergroup normalizer circuit and stabilizer formalisms (sections \ref{sect: Circuit Model}--\ref{sect:Stabilizer Formailsm}) to nonabelian groups.

\subsubsection{The hypergroup of  conjugacy classes of $G$}\label{sect:ConjClassHypergroup}

Let $G$ be any finite group. For any $g \in G$, we let $C_g := \set{g^a}{a
\in G}$, where $g^a := a^{-1}ga$ denotes the conjugacy class of $g$. We let
$\Conj{G}$ be the set of distinct conjugacy classes of $G$.

Let $C = \{g_1, g_2, \dots\}$ and $D = \{h_1, h_2, \dots\}$ be two conjugacy
classes.  Then, for any product, $g_i h_j$, its conjugate $(g_i h_j)^a = g_i^a
h_j^a$ is a product of conjugates, so it can be written as $g_k h_\ell$ for
some $k$ and $\ell$. Furthermore, if there are $M$ distinct products $g_{i_1}
h_{j_1}, \dots, g_{i_M} h_{j_M}$ producing some element $x$, then the distinct
products $g_{i_1}^a h_{j_1}^a, \dots, g_{i_M}^a h_{j_M}^a$ all produce $x^a$.
Thus, for each conjugacy class $E$ arising in the product of elements of $C$
and $D$, we get a well defined number of ``how many times'' that class arises,
which we denote $M_{C,D}^E$.

We will denote by $\A(\Conj{G})$ the complex vector space with the distinct
conjugacy classes as a basis, which make into a $\Complex$-algebra by defining
the product $C D := \sum_{E \in \Conj{G}} M_{C,D}^E E$.

We take the map $C_g \mapsto C_{g^{-1}}$, extended to all of $\A(\Conj{G})$ by
linearity, as our involution.

It is easy to see that $C_e$ arises in a product $C_g C_h$ iff $C_h$ contains
$g^{-1}$, which occurs iff $C_h = C_{g^{-1}}$. Thus, we can see that the first
of the two required properties holds for the product with structure constants
$M_{C,D}^E$.

To get the normalization property to hold, though, we must make a minor change.
For each $C_g \in \Conj{G}$, define $c_g$ to be the vector
$(1/\abs{C_g})\,C_g$. Then we will take $\set{c_g}{C_g \in \Conj{G}}$
to be a new basis. The structure constants become $m_{C,D}^E :=
M_{C,D}^E \abs{E} / \abs{C} \abs{D}$. Since the total number of products of
elements formed multiplying $C$ by $D$ is $\abs{C} \abs{D}$, we can see that
$\sum_{E \in \Conj{G}} M_{C,D}^E \abs{E} = \abs{C} \abs{D}$, which means that
these new structure constants, $m_{C,D}^E$, are properly normalized. Thus,  conjugacy classes define a hypergroup, up to this normalization, which we call the \textbf{class hypergroup} $\overline{G}$. 

Finally, we note that this hypergroup is abelian, even if the underlying group
is not abelian. To see this, we calculate $gh = h h^{-1} g h = h g^h = (hg)^h$
(since $h^h = h$), which shows that $gh$ and $hg$ are in the same conjugacy
class. Hence, if we are multiplying conjugacy classes instead of elements, we
do not distinguish between $gh$ and $hg$, and we get an abelian structure.

\subsubsection{The hypergroup of characters}
\label{sect:CharacterHypergroup}

Let $\Repr{G}$ denote the set of irreducible characters of the finite group  of $G$, $\A(\Repr{G})$ the complex vector space with basis $\widehat{G}$, and $\chi_\mu$ the character of the irreducible representation $\mu$. As we explain next, $\Repr{G}$ has a natural hypergroup structure.

First, the  involution of $\widehat{G}$ will be the linear extension of the map  $\chi_\mu \mapsto
\overline{\chi}_\mu$, for $\chi_\mu \in \Repr{G}$, the image also being an
irreducible character.

Second, for any two characters, $\chi_\mu$ and $\chi_\sigma$, the pointwise product of
these functions $\chi_\mu \chi_\sigma$ is also a character, though it is not
necessarily irreducible. However, as is well known, any representation
can be written as a linear combination of irreducible characters: $\chi_\mu
\chi_\sigma = \sum_{\tau \in \Repr{G}} N_{\mu,\sigma}^\tau \chi_\tau$ for some
non-negative integers $N_{\mu,\sigma}^\tau$. Using this as our product,
$\A(\Repr{G})$ becomes a $C^*$-algebra, where the identity element is the
trivial representation, $\chi_1$, given by ${\chi_1}(g) \equiv 1$.

Third, the coefficient $N_{\mu,\sigma}^\tau$, as is also well known from
representation theory, is given by $\langle{\chi_\mu \chi_\sigma},{\chi_\tau}\rangle$,
where $\langle{\cdot},{\cdot}\rangle$ is the inner product 
 $\langle{\chi_\mu},{\chi_\sigma}\rangle = \abs{G}^{-1} \sum_{g \in G}
\chi_\mu(g) \overline{\chi_\sigma(g)}$.  From this, we can see that
$N_{\mu,\sigma}^{1} = \langle{\chi_\mu \chi_\sigma},{{\chi_1}}\rangle =
\langle{\chi_\mu},{\overline{\chi_\sigma}\,{\chi_1}}\rangle =
\langle{\chi_\mu},{\overline{\chi_\sigma}}\rangle$.  Since $\chi_\mu$ and
$\overline{\chi_\sigma}$ are both irreducible, this is 1 if $\chi_\sigma =
\overline{\chi_\mu}$ and 0 otherwise.  Hence, we can see that the structure
constants $N_{\mu,\sigma}^\tau$ have the first required property.

Finally, we will normalize the characters, as in section \ref{sect:ConjClassHypergroup}, in order to have (ii). For this, we define  $\widehat{G}$, the \textbf{character hypergroup of $G$}, to be the hypergroup with elements
\begin{equation}\label{eq:Normalized Character}
\mathcal{X}_\mu:=\frac{\chi_\mu}{d_\mu}
\end{equation}
where $d_\mu$ is the dimension of the irrep $\mu$. The  structure constants now become $n_{\mu,\sigma}^\tau
:= N_{\mu,\sigma}^\tau d_\tau / d_\mu d_\sigma$. Since $\chi_\mu
\chi_\sigma$ is actually the character of the representation $\mu \otimes
\sigma$, which splits into a direct sum of irreducible
representations (as described  above), we must have $\sum_{\tau
\in \Repr{G}} N_{\mu,\sigma}^\tau d_\tau = d_\mu d_\sigma$ as the latter is the
dimension of the tensor product. This implies that (ii) is fulfilled and that  $\widehat{G}$ (now suitably normalized) is indeed a hypergroup.

Finally, we note that, in this case,  our product is manifestly
abelian since $\chi_\mu \chi_\sigma$ denotes the element-wise product of these
functions, which takes place in the abelian group $\Complex$.

\subsubsection{The relationship between $\overline{G}$ and $\widehat{G}$}

Crucially, the characters of the hypergroup $\Conj{G}$ turn out to be the  normalized characters $\mathcal{X}_\mu=\chi_\mu/d_\mu$ of $G$ and, due to duality (\ref{eq:Hypergroups}),  conjugacy classes are the characters of $\widehat{G}$. Classes and characters have weights $\ws{C_g}=|C_g|$ and $\ws{\mathcal{X}_\mu}=d_\mu^2$, respectively, where $d_\mu$ is the dimension of the irrep $\mu$. This fantastic connection between groups and hypergroups lets one easily derive many well known results in nonabelian group character theory \cite{Humphrey96_Course_GroupTheory,Isaacs1994character} using the properties of section \ref{sect:DefinitionsHypergroup}, including the usual character orthogonality relationships and the famous $|G|=\sum_{\Hchi_\mu\in \widehat{G}} d_\mu^2$ identity: the latter can be derived from (\ref{eq:Weight Duality}), which leads to $\varw{\overline{G}}=\sum_{C_g\in\overline{G}} |C_g| =|\overline{G}|= \sum_{\Hchi_\mu\in \widehat{G}} d_\mu^2$ and also implies, $\varw{\overline{G}}=|\overline{G}|$ and $\varw{\widehat{G}}= |\widehat{G}|$.

\section{Understanding the Hidden Normal Subgroup Problem}
\label{sec:hnsp}
\label{sect:HNSP}

In this section, we demonstrate a formal connection between the hidden \emph{normal} subgroup problem (HNSP) and a problem on abelian hypergroups, defined below, which we call the CC-HSHP. Specifically, we show that, in many cases, we can efficiently reduce the HNSP to the CC-HSHP, classically. This reduction tells us that, even though the HNSP is defined in terms of nonabelian groups, it can be translated into a problem about an algebraic structure that is abelian, albeit one that is more complex than a group (a hypergroup).

In the remainder of the paper, we will see the effects of moving from nonabelian groups to abelian hypergroups. While the switch from groups to hypergroups creates some new difficulties, we also gain a great deal by working with an abelian structure. In particular, we will see that the mathematical structure of abelian hypergroups leads to a beautiful stabilizer formalism and to new quantum algorithms. Here, we explain how abelian hypergroups arise specifically when looking for hidden normal subgroups  before moving to the more general setting.

In the first subsection, we formally define the two problems mentioned above, the HNSP and the CC-HSHP. Afterwards, we show how to reduce the former to the latter.

\subsection{The HSNP and the CC-HSHP}

In the HNSP, we are given an oracle $f:G\rightarrow \{0,1\}^*$, assigning
labels to group elements, that is promised to \emph{hide} some normal subgroup
$N \lhd G$. The latter means that we have $f(x) = f(x')$ for $x, x' \in G$ if
and only if $x' = xn$ for some $n \in N$. An algorithm solves the HNSP if it
can use this oracle and other quantum computation in order to determine the
subgroup $N$ with high probability.

The algorithm of Hallgren et al. for the HNSP finds the hidden subgroup $N$
using exclusively information provided by characters of the group. They showed
that this works only for normal subgroups as it cannot distinguish a non-normal
subgroup $H \le G$ from a conjugate subgroup $H^a \not= H$.

If we are only examining the characters of the group $G$, then it stands to
reason that we can get the same information from the hypergroup of characters
$\Repr{G}$ or, equivalently, from the hypergroup of conjugacy classes
$\Conj{G}$ since these two hypergroups contain the same
information.\footnote{After all, each can be recovered from the other as its
dual hypergroup.} Hence, we may expect that the HNSP on $G$ is related to some
problem on the abelian hypergroup $\Conj{G}$.

A natural question for abelian hypergroups like these is the hidden
subhypergroup problem \cite{Amini_hiddensub-hypergroup}. For our
abelian hypergroup of conjugacy classes, we will refer to this problem as the
conjugacy class hidden subhypergroup problem or CC-HSHP.  Here, we are given an
oracle $f : \Conj{G} \rightarrow \{0,1\}^*$, assigning labels to conjugacy
classes, that hides some subhypergroup, and we are
asked to determine that subhypergroup via oracle queries and quantum
computation. We will see next how this is related to the HNSP.

\subsection{Reducing the HNSP to the CC-HSHP}

Since a normal subgroup $N$ is (the union of) a set of conjugacy classes that
is closed under multiplication and taking inverses, it also defines a
subhypergroup of $\Conj{G}$, which we denote by $\subhyp{N}{G}$.\footnote{We
distinguish this from $N$, which is a set of group elements, because
$\subhyp{N}{G}$ is a set of conjugacy classes.} Hence, any subgroup that can be
found as the solution of the HNSP can also be found as the solution of the
CC-HSHP. Indeed, as we will see next, in many cases, we can directly reduce
the HNSP to the CC-HSHP.

In order to perform this reduction, we need to provide a CC-HSHP oracle. Our proofs will show how to translate an oracle for the HNSP into an oracle for the CC-HSHP. These translations assume that we can perform certain computations with conjugacy classes, described in detail in appendix~\ref{sect:class basis operations}, which we refer to as ``computing efficiently with conjugacy classes''. (While formally an assumption, we know of no group for which these calculations cannot be performed efficiently.)

\begin{theorem}[HNSP $\le$ CC-HSHP, I]
\label{thm:reduction1}
Let $G$ be a group. Suppose that we are given a hiding function $f : G \rightarrow \{0,1\}^*$ that is also a class function\footnote{This means that $f$ is constant on conjugacy classes. This will occur iff $G / N$ is abelian, where $N$ is the hidden subgroup.}. If we can compute efficiently with conjugacy classes, then we can efficiently reduce this HNSP to the CC-HSHP.
\end{theorem}
\begin{proof}
The assumptions about computing efficiently with conjugacy classes imply that, given a conjugacy class $C_g$, we can efficiently find an element $x \in C_g$ and apply $f$ to get a label. (Since $f$ is a class function, the label is the same for any $x' \in C_g$.) Let $N$ be the hidden subgroup.  Since $f$ hides $N$ and $N$ is normal, we can see that $f(xn) = f((xn)^a) = f(x^a n^a) = f(x^a) = f(x^a n')$ for any $n, n' \in N$. This shows that $f$ is constant on $C_g\subhyp{N}{G}$, which corresponds to a coset of the subhypergroup $\subhyp{N}{G} \le \Conj{G}$. It follows immediately that $f$ has distinct values on distinct cosets of $\subhyp{N}{G}$, so we can see that $f$ is a hiding function for this subhypergroup corresponding (uniquely) to $N$.
\end{proof}

\begin{theorem}[HNSP $\le$ CC-HSHP, II]
\label{thm:reduction2}
Let $G$ be a group. Suppose that we are given a hiding function $f : G \rightarrow H$ that is also a homomorphism.  If we can efficiently compute with conjugacy classes of $G$ and $H$, then we can efficiently reduce this HNSP to the CC-HSHP.
\end{theorem}

\begin{proof}
Consider any element $x \in G$. For any conjugate $x^a$, for some $a \in G$, we
see that $f(x^a) = f(a^{-1} x a) = f(a^{-1}) f(x) f(a)$ since $f$ is a
homomorphism. Furthermore, since $a^{-1} a = e$, we see that $f(a^{-1}) f(a) =
f(e) = e$, which shows that $f(a^{-1}) = f(a)^{-1}$. Putting these together, we
have $f(x^a) = f(a)^{-1} f(x) f(a) = f(x)^{f(a)}$. This means that the function
$\tilde{f}$ taking $x$ to the conjugacy class label of $f(x)$ is a class
function, which we can compute efficiently by assumption.\footnote{Also note
that, since $e$ is the only element in its conjugacy class, $\tilde{f}$ hides
the same subgroup as $f$.} Thus, by the same proof as in previous theorem, we can reduce this to
the CC-HSHP.
\end{proof}

This latter theorem applies to many of the important examples of HSPs. This
includes the oracles used for factoring, discrete logarithm over cyclic groups
and elliptic curves, and abelian group decomposition
\cite{BermejoLinVdN13_BlackBox_Normalizers}.
While all of these examples are abelian groups, it is true in general that, for
any normal subgroup of any group, there is always some hiding function that is a group
homomorphism.\footnote{If $H \trianglelefteq G$ is the hidden subgroup, then one example is
the canonical oracle $G \rightarrow G /H$ given by $x \mapsto xH$.}

From these proofs, we can see that the essential difference between the HNSP
and the CC-HSHP is the slightly differing requirements for their oracles.  We
have seen that, whenever we can convert an oracle for the former into one for
the latter, we can reduce the HNSP to the CC-HSHP.\footnote{This also assumes
the relatively minor assumption that we can compute with conjugacy classes.}
Above, we showed this can be done in the case that the two sets of requirements
are actually the same (\textbf{theorem~\ref{thm:reduction1}}) and the case where the
labels produced by the oracle are not opaque but rather come with enough
information to compute with their conjugacy classes
(\textbf{theorem~\ref{thm:reduction2}}).

Apart from this, it is worth reflecting on which of the types of oracle is the
most sensible for the problem of finding hidden normal subgroups. With this in
mind, we note that the oracle in the HNSP is not specific to normal
subgroups: the same type of oracle can hide non-normal subgroups as well --- we
are simply promised that, in these cases, the hidden subgroup happens to be
normal. In contrast, the oracle in the CC-HSHP can \emph{only} hide normal
subgroups because it is required to be constant on conjugacy classes. Hence,
even though we came upon the oracle definition from the CC-HSHP by looking at
hypergroups, it is arguable that this is actually a \emph{better} definition of 
hiding function for  normal subgroups.  Our proofs above demonstrate that,
whenever we are given an oracle of this type, we can reduce finding the hidden normal subgroup to the CC-HSHP.

We will return to the HNSP in section~\ref{sect:Quantum Algorithms}. There, we will show that the CC-HSHP can be efficiently solved on a quantum computer under reasonable assumptions. This, together with the theorems above, show that the HNSP is easy because the CC-HSHP is easy, which gives an explanation for why the HNSP is easy in terms of the presence of an \emph{abelian} algebraic structure.

Before we can do that, however, we need to first develop some tools for analyzing quantum algorithms using abelian hypergroups. These tools will be of independent interest.

\section{Normalizer circuits over abelian hypergroups}\label{sect: Circuit Model}

In section \ref{sect:HNSP}, we described our motivating example (the hidden normal subgroup problem) for considering how abelian hypergroups can be used to understand quantum computation. There, the abelian hypergroups arose from nonabelian groups. However, there are a vast number of interesting hypergroups with applications in physics and mathematics \cite{Wildberger1994HypergroupsApplications}, including many of the ones used in topological quantum computation \cite{kitaev_anyons,kitaev_anyons}, that do \emph{not} arise from groups. So in the next three sections, we will work with a general abelian hypergroup $\mathcal{T}$, which could come from any of these settings.

Our plan in these next few sections is to follow the model that allowed abelian groups to be used so successfully to understand quantum computation \cite{VDNest_12_QFTs,BermejoVega_12_GKTheorem,BermejoLinVdN13_Infinite_Normalizers,BermejoLinVdN13_BlackBox_Normalizers}. We will see that much of the same machinery developed with abelian groups can also be developed with abelian hypergroups. While new difficulties appear, many of the most useful properties remain. Since abelian hypergroups generalize abelian groups, our constructions generalize those for abelian groups as well.

We start, in this section, by defining the class of circuits that we will analyze. We call this model, defined in section \ref{sect: Circuit Model Definitions}, normalizer circuits over hypergroups. In section \ref{sect:Examples Normalizer Circuits}, we go through a few examples of what these models consist of for different hypergroups. In later sections, we develop a stabilizer formalism and a Gottesman-Knill-type theorem that applies to these circuits.

\subsection{Circuit model}\label{sect: Circuit Model Definitions}

Fix $\mathcal{T}$  to be an arbitrary \emph{finite abelian hypergroup}. We now define a circuit model, which we call \emph{normalizer circuits} over $\mathcal{T}$. The gates of these circuits are called \emph{normalizer gates}.

\myparagraph{The Hilbert space:} Normalizer gates over  $\mathcal{T}$  act on a Hilbert space  $\Comp_\mathcal{T}$ with two orthonormal bases, $\mathsf{B}_{\mathcal{T}}=\{\ket{a}, a \in \mathcal{T}\}$ and  $\mathsf{B}_{\mathcal{T}^*}=\{\ket{\mathcal{X}_\mu}, \mathcal{X}_\mu \in \mathcal{T}^*\}$, labeled by elements and characters of $\mathcal{T}$,\footnote{Note that duality (\ref{eq:Hypergroup Duality}) implies $\dim \Comp_\mathcal{T} = \dim \Comp_\mathcal{T^*}$.} that  are related via the \emph{quantum Fourier transform} (QFT) of $\mathcal{T}$:
\begin{equation}\label{eq:Quantum Fourier Transform over Hypergroup T}
\Fourier{\mathcal{T}}\ket{a}=\sum_{\mathcal{X_\mu}\in\mathcal{T^*}}\sqrt{\frac{\w{\mathcal{X}_\mu} \w{a}}{\varpi_{\mathcal{T^*} }}} \mathcal{X}_\mu(a)\ket{\mathcal{X}_\mu},\qquad \Fourier{\mathcal{T}}^\dagger \ket{\mathcal{X_\mu}} = \sum_{a\in\mathcal{T}}\sqrt{\frac{\w{a}\w{\overline{\mathcal{X}_\mu}}}{\varpi_{\mathcal{T} }}} \overline{\mathcal{X}_\mu}(a)\ket{a}.
\end{equation}
Character orthogonality  (\ref{eq:Character Orthogonality}) implies that (\ref{eq:Quantum Fourier Transform over Hypergroup T}) is a unitary transformation.

\myparagraph{Registers:} Because in many settings it is important  to split a   quantum computation in multiple registers, we let   $\mathcal{T}$ and $\mathcal{H_T}$  have a general direct product and tensor product form
\begin{equation}\label{eq:Product Hypergroup and Hilbert Space}
\mathcal{T} = \mathcal{T}_1 \times \dots \times \mathcal{T}_m \qquad\longleftrightarrow \qquad
\Comp_{\mathcal{T}} \cong \Comp_{\mathcal{T}_1} \otimes \dots \otimes \Comp_{\mathcal{T}_m}.
\end{equation}
In this case, the QFT over $\mathcal{T}$ is the tensor product  of the QFTs over the $\mathcal{T}_i$'s:
\begin{equation}\label{eq:partial QFT}
\Fourier{\mathcal{T}}=\Fourier{\mathcal{T}_1}\otimes \cdots\otimes  \Fourier{\mathcal{T}_m}.
\end{equation}
\myparagraph{Input states:} Each register $\mathcal{H}_{\mathcal{T}_i}$ is initialized to be in either an \emph{element state} $\ket{{x_i}},{x_i}\in \mathcal{T}_i$ or in a \emph{character state} $\ket{\Hchi_\mu}, \Hchi_\mu\in\mathcal{T}_i^*$.

\myparagraph{Gates:}The allowed \emph{normalizer gates} at step $t$ of a normalizer circuit depend on a parameter $\mathcal{T}(t)$, which is a hypergroup, related to $\mathcal{T}$, of the form
\begin{equation}\label{eq:Hypergroups}
\mathcal{T}(t)=\mathcal{T}(t)_{1}\times\cdots \times \mathcal{T}(t)_{m} \quad \textnormal{with}\quad  \mathcal{T}(t)_i \in \{{\mathcal{T}}_i, {\mathcal{T}}_i^*\}.
\end{equation}
The role of $\mathcal{T}(t)$ is to indicate whether the operations carried out by circuit at time $t$ will be on the element or character basis.  At step $0$, $\mathcal{T}(0)$ is chosen so that  $\mathcal{T}_{i}(0)\in \{\mathcal{T}_i, \mathcal{T}_i^*\}$ indicates whether $\mathcal{H}_{\mathcal{T}_i}$ begins on an element or character state.  At any steps $t> 0$, $\mathcal{T}(t)$  depends on the gates that have been applied at earlier steps, following  rules given below.

Normalizer gates at time $t$  can be of four types:
\begin{enumerate}
\item \textbf{Pauli gates.}  Pauli gates of type X  implement the $\mathcal{T}(t)$ hyperoperation $\PX{\mathcal{T}(t)}(a)\ket{b}=\ket{ab}$ for invertible elements $a\in\mathcal{T}(t)$. Pauli gates of type Z multiply by phases $\PZ{\mathcal{T}(t)}(\mathcal{X}_\mu)\ket{b}=\mathcal{X}_\mu(b)\ket{b}$ which correspond to invertible characters in $\mathcal{T}(t)^*$. 
\item \textbf{Automorphism Gates.} Let $\alpha : \mathcal{T}(t) \rightarrow \mathcal{T}(t)$ be an
automorphism of the hypergroup $\mathcal{T}(t)$. Then the automorphism gate $U_\alpha$ taking $\ket{g}\mapsto \ket{\alpha(g)}$ is a valid normalizer gate. 

\item \textbf{Quadratic Phase Gates} A complex function  $\xi : \mathcal{T}(t) \rightarrow U(1)$ is
called ``quadratic'' if the map $B(g,h):\mathcal{T}(t)\times
\mathcal{T}(t)\rightarrow U(1)$ defined by
$\xi(gh)=\xi(g)\xi(h)B(g,h)$ is a bi-character, i.e., a character of
the hypergroup in either argument. A quadratic phase gate is a diagonal map
$D_\xi$ taking $\ket{g} \mapsto \xi(g) \ket{g}$ for some quadratic function
$\xi$.
\item \textbf{Quantum Fourier Transforms}. A \emph{global QFT} implements the  gate $\mathcal{F}_{\mathcal{T}(t)}$ over $\mathcal{T}(t)$  (\ref{eq:Quantum Fourier Transform over Hypergroup T}). \emph{Partial QFTs}  implement the gates $\mathcal{F}_{{\mathcal{T}(t)}_i}$ on single registers $\mathcal{H}_{\mathcal{T}(t)_i}$ (while the other registers remain unchanged).
\end{enumerate}
\indent \myparagraph{Update rule:} Because QFTs change the hypergroup that labels the standard basis (\ref{eq:Quantum Fourier Transform over Hypergroup T}), the rules above do not specify which  normalizer gates should be applied on the second step. For this reason, in our gate model, we  \emph{update} the value of $\mathcal{T}(t+1)$ at time $t+1$ so that $\mathcal{T}(t+1)_i=\mathcal{T}(t)_i^*$ if a QFT acts on $\mathcal{H}_{\mathcal{T}(t)_i}$ and $\mathcal{T}(t+1)_i=\mathcal{T}(t)_i$ otherwise.

\myparagraph{Measurements:} At the final step $T$, every register $\mathcal{H}_{\mathcal{T}_i}$ is measured in either the element or the character basis depending on the configuration of the QFTs in the circuit: specifically,  $\mathcal{H}_{\mathcal{T}_{i}}$ is measured in basis $\textsf{B}_{\mathcal{T}_i}$ labeled by elements of $\mathcal{T}_i$  when  $\mathcal{T}(T)_i=\mathcal{T}_i$, and in the character basis $\textsf{B}_{\mathcal{T}_i^*}$ when  $\mathcal{T}(T)_i=\mathcal{T}_i^*$. In the end, the final string of measurement outcomes identifies an element of the hypergroup $\mathcal{T}(T)$.

\subsection{Examples from group theory}\label{sect:Examples Normalizer Circuits}

We now give examples of normalizer gates over conjugacy class and character hypergroups with the aim to illustrate our definitions and, furthermore, show how our results can be applied to define models of \emph{normalizer circuits over nonabelian groups}.

\subsubsection*{Example 1: Clifford and abelian-group normalizer circuits}

For an abelian group $G$, all conjugacy classes contain a single group element. Consequently, the class hypergroup $\overline{G}$ is always a \emph{group} and it is equal to $G$. In this scenario, our gate model coincides with the finite abelian-group normalizer-circuit model studied in \cite{VDNest_12_QFTs,BermejoVega_12_GKTheorem}, which contain numerous examples of normalizer gates. Choosing  $G=\Integers_2^n$ or  $G=\Integers_d^n$, normalizer circuits become the standard Clifford circuits for $n$-qubits \cite{Gottesman99_HeisenbergRepresentation_of_Q_Computers} and $n$-qudits \cite{Gottesman98Fault_Tolerant_QC_HigherDimensions}. More exotic examples are given in \cite{VDNest_12_QFTs,BermejoVega_12_GKTheorem} for the case $G=\Integers_{d_1}\times\cdots\times \Integers_{d_m}$, where  normalizer circuits can contain quantum Fourier transforms, powerful quantum gates used  Shor's algorithm \cite{Shor}.

All the above examples of normalizer circuits are efficiently classical simulable \cite{VDNest_12_QFTs,BermejoVega_12_GKTheorem}. More exotic instances  of abelian-group normalizer circuits that cannot be be simulated \cite{BermejoLinVdN13_BlackBox_Normalizers} include groups of the form  $G=\Integers_N^\times$, solutions of elliptic curves and, in general, finite abelian groups that  cannot be efficiently written in the form  $\Integers_{d_1}\times\cdots\times \Integers_{d_m}$. In this case, it is shown in \cite{BermejoLinVdN13_BlackBox_Normalizers} that Shor's discrete-log quantum algorithm is a normalizer circuit over $\Integers_{p-1}^2 \times \Integers_{p}^{\times}$; a similar result is proven in \cite{BermejoLinVdN13_BlackBox_Normalizers} for Shor's factoring, which can be understood as a normalizer circuit over an infinite group (which we do not consider in this paper).\footnote{See also \cite{BermejoLinVdN13_Infinite_Normalizers} for more general types of normalizer circuits over infinite abelian groups.}

\subsubsection*{Example 2: normalizer circuits over nonabelian groups}
\label{sect:Circuits nonabelian group}

We now apply our circuit formalism to introduce (new) models of normalizer gates over any finite nonabelian group $G$. For this, we  associate a Hilbert space $\mathcal{H}_{G}$ to $G$ with basis $\{\ket{g},g \in G\}$ and restrict the computation to act on its  (nontrivial) subspace  $\mathcal{I}_{\overline{G}}$ of  \emph{conjugation invariant} wavefunctions.\footnote{That is, wave functions $\psi(x)$ such that $\psi(x^g) = \psi(x)$ for all $x,g \in G$.}

As is well-known from representation theory \cite{Isaacs1994character}, the Dirac delta measures $\delta_{C_g}$ over conjugacy classes $C_g\in\overline{G}$ and the  character functions $\chi_\mu$ of the irreducible representations $\mu\in\mathrm{Irr}(G)$ form two dual orthonormal bases of $\mathcal{I}_\Conj{G}$. In our circuit model, recalling the definitions of class hypergroup $\overline{G}$ and character hypergroup $\widehat{G}$ (see section \ref{sect:Examples Hypergroups}), this means that $\mathcal{I}_{\overline{G}}$ can be viewed as the Hilbert space of the conjugacy class hypergroup $\mathcal{H}_{\overline{G}}$ with a \emph{conjugacy-class basis} $\mathsf{B}_{\overline{G}}=\{\ket{C_{g}},C_{g}\in G\}$ and a \emph{character basis} $\mathsf{B}_{\widehat{G}}=\{\ket{\mathcal{X}_{\mu}},\mathcal{X}_{\mu}\in\widehat{G}\}$ if we define these bases within $\mathcal{H}_G$ as the vectors
\begin{equation}\label{eq:Bases Nonabelian Group}
\ket{C_{g}}=\frac{1}{\sqrt{|C_{g}|}}\sum_{aga^{-1} \in C_{g}} \ket{aga^{-1}}\ \text{ and }\ \ket{\mathcal{X}_{{\mu}}}={\sqrt{\frac{d_\mu^2}{|G|}}}\sum_{g\in G}  \overline{\Hchi_{\mu}} (g)\ket{g}.
\end{equation}

With these identifications, we can now define a \emph{normalizer circuit over $G$} to be a normalizer circuit over the hypergroup $\overline{G}$: the latter acts on the conjugation invariant subspace, admits conjugacy class and character state inputs, and applies QFTs, group automorphisms, and quadratic phase functions associated to $\overline{G}$ and $\widehat{G}$. Furthermore, if we have a direct product $G=G_1\times \cdots \times G_m$, then $\overline{G}=\overline{G}_1\times \cdots \times \overline{G}_m$, $\widehat{G}=\widehat{G}_1\times \cdots \times \widehat{G}_m$ and $\mathcal{H}_{\Conj{G}}=\mathcal{H}_{\Conj{G}_1}\otimes \cdots \otimes \mathcal{H}_{\Conj{G}_m}$. In this setting, normalizer gates such as partial QFTs and entangling gates over different registers are allowed.

It is straightforward to check, using the identities $\ws{C_g}=|C_g|$, $\ws{\Hchi_\mu}=d_\mu^2$ and $\varw{\Conj{G}}=\varw{\widehat{G}}=|G|$ (section \ref{sect:Examples Hypergroups}), that the QFT defined with the bases from (\ref{eq:Bases Nonabelian Group}) is actually the identity map. Even so, the QFT performs a useful purpose in these circuits as it changes the basis used for subsequent gates, $\mathcal{T}(t+1)$. In particular, the QFT can change the final basis to the character basis, which means that the final measurement is performed in the character basis rather than the element basis.

As we show in appendix \ref{app:CC implementation details}, we can perform a final measurement in the character basis provided that we have an efficient QFT circuit for the group $G$. The same techniques also allow us to prepare initial states and perform all the gate types (Pauli gates, automorphisms, and quadratic phases) in the character basis efficiently.

Performing the gate types in the conjugacy class basis is straightforward if we make some modest assumptions about our ability to compute with conjugacy classes of the group. For example, we need a way to map an element $g \in G$ to a label of its conjugacy class $C_g$. These details are discussed in appendix~\ref{sect:class basis operations}, where we explain why these assumptions are easily satisfied for typical classes of groups.

\subsubsection*{Example 3: quaternionic circuits}\label{sect:Quaternionic circuits}

Lastly, we give concrete examples of normalizer gates over nonabelian groups  for systems of the form $Q_8^n$ where  $Q_8$ is the quaternion group with  presentation \begin{equation}
Q_8=\langle -1,i,j,k | (-1)^2=1,i^2=j^2=k^2=ijk=-1\rangle
\end{equation} Note that $Q_8$ is nonabelian and that, although it has  eight elements $\pm 1, \pm i, \pm j, \pm k$, it has only five conjugacy classes $\{1\}$, $\{-1\}$, $\{\pm i\}$, $\{\pm j\}$, $\{\pm k\}$. Hence, although the Hilbert space  $\mathcal{H}_{Q_8}^{\otimes n}=\{\ket{g}: g\in Q_8\}$ is $8^n$-dimensional, our quantum computation based on normalizer gates will never leave the $5^n$-dimensional  conjugation invariant subspace  $\Comp_{\Conj{Q_8}}^{\otimes n}$, which can be viewed as a system of 5-dimensional qudits. Using the group character table \cite{Atiyah61_Characters_Cohomology}, it is easy to write down the conjugacy class and character basis states of $\Comp_{\Conj{Q_8}}$:
\begin{itemize}
\item \textbf{Conjugacy-class states:}
\begin{equation*}
\ket{C_1}=\ket{1},\quad \ket{C_{-1}}=\ket{-1},\quad\ket{C_{i}}=\tfrac{\ket{i}+\ket{-i}}{\sqrt{2}},\quad\ket{C_{j}}=\tfrac{\ket{j}+\ket{-j}}{\sqrt{2}},\quad\ket{C_{k}}=\tfrac{\ket{k}+\ket{-k}}{\sqrt{2}}\vspace{-5pt}
\end{equation*}

\item \textbf{Character states:} 
\begin{align*}
\ket{\mathcal{X}_1} &= \tfrac{1}{\sqrt{8}} \left(\ket{C_1} +\ket{C_{-1}}+\sqrt{2}\ket{C_{i}}+\sqrt{2}\ket{C_{j}}+\sqrt{2}\ket{C_{k}}\right),\\
\ket{\mathcal{X}_i} &= \tfrac{1}{\sqrt{8}} \left(\ket{C_1} +\ket{C_{-1}}+\sqrt{2}\ket{C_{i}}-\sqrt{2}\ket{C_{j}}-\sqrt{2}\ket{C_{k}}\right),\\
\ket{\mathcal{X}_j} &= \tfrac{1}{\sqrt{8}} \left(\ket{C_1} +\ket{C_{-1}}-\sqrt{2}\ket{C_{i}}+\sqrt{2}\ket{C_{j}}-\sqrt{2}\ket{C_{k}}\right),\\
\ket{\mathcal{X}_k} &= \tfrac{1}{\sqrt{8}} \left(\ket{C_1} +\ket{C_{-1}}-\sqrt{2}\ket{C_{i}}-\sqrt{2}\ket{C_{j}}+\sqrt{2}\ket{C_{k}}\right),\\
\ket{\mathcal{X}_2} &= \tfrac{2}{\sqrt{8}} \left(\ket{C_1} -\ket{C_{-1}}\right),
\end{align*}
\end{itemize}
We now give a list of nontrivial normalizer gates (not intended to be exhaustive), which we obtain directly from the definitions in section \ref{sect: Circuit Model} applying basic properties of the quaternion group \cite{Humphrey96_Course_GroupTheory,Atiyah61_Characters_Cohomology}. For the sake of conciseness, the elementary group-theoretic derivations are omitted.
\begin{itemize}
\item \textbf{Quantum Fourier transform.} For one qudit, the QFT implements the change of basis between the conjugacy-class and character bases written above. For $n$-qudits, the total QFT implements this change of bases on all qudits. Partial QFTs, instead, implement the QFT on single qudits.
\vspace{5pt}
\item \textbf{Pauli gates:} $\PX{\overline{Q}_8}(-1)\ket{C_x}=\ket{-C_x},\quad \PZ{\overline{Q}_8}(\mathcal{X}_\ell)\ket{C_x}=\mathcal{X}_\ell(C_x)\ket{C_x}\quad$for $\ell = i, j, k$.
\vspace{5pt}
\item \textbf{Automorphism gates:} All automorphisms of the class-hypergroup can be obtained by composing  functions $\alpha_{xy}$ that swap pairs of conjugacy classes $C_x$, $C_y$ with $x, y\in \{i,j,k\}$. The corresponding \emph{swap gates} $U_{\alpha_{xy}}\ket{C_z}=\ket{\alpha_{xy}(C_z)}$ are instances of one-qudit quaternionic automorphism gates.

\item \textbf{Quadratic phase gate.} Next, we give examples of non-linear quadratic phase gates. For one qudit, quadratic phase gates $D_{\xi_{i}}$, $D_{\xi_{j}}$, $D_{\xi_{k}}$ defined as 
\begin{equation*}
D_{\xi_{x}}\ket{C_y}=\ket{C_y}, \textnormal{ if } y\in\langle x\rangle=\{\pm 1, \pm x\}  \quad\textnormal{and} \quad
D_{\xi_{x}}\ket{C_z}=i\ket{C_z}   \textnormal{ otherwise},
\end{equation*} 
  provide quaternionic analogues of the one-qubit $P=\mathrm{diag}(1,i)$ Clifford gate. 
  
  For two qudits, there is also a ``\emph{quaternionic controlled-Z gate}'' $D_\xi$, which implements a quadratic function  $\xi(C_x,C_y)=f_{C_x}(C_y)$, with $f_{C_x}$ being a linear character specified by the following rules: $f_{C_{\pm1}}=\mathcal{X}_1$ and $f_{C_{x}}=\mathcal{X}_x$ for $x=i,j,k$. We refer the reader to  appendix \ref{app:Quadratic functions} for a proof that the above functions are quadratic.
  
\end{itemize}

Most of the above gates act on a single copy of $\mathcal{H}_{Q_8}$ and, thus, cannot generate entanglement. Entangling normalizer gates can be found by considering two copies of $\mathcal{H}_{Q_8}$. The allowed  normalizer gates are now those associated to the group $Q_8\times Q_8$.

We give next three examples of two-qudit  automorphism gates $U_{\alpha_i}$, $U_{\alpha_j}$, $U_{\alpha_k}$, that can generate \textbf{quantum entanglement} and provide quaternionic analogues of the qubit CNOT \cite{Gottesman99_HeisenbergRepresentation_of_Q_Computers} and the qudit CSUM gates \cite{Gottesman98Fault_Tolerant_QC_HigherDimensions}. The three are defined as
\begin{equation}
U_{\alpha_x}\ket{C_1,C_2}=\ket{\alpha_x(C_1,C_2)}=\ket{C_1,f_x(C_1)C_2}
\end{equation}   where  $f_x(C_y)=C_1$ if $C_y$ is contained in the subgroup $\langle x \rangle = \{\pm 1,\pm x\}$ generated by $x$ and $f_{x}(C_y)=C_{-1}$ otherwise\footnote{The function   $f_x$ defines a group homomorphism from $\overline{Q}_8$ into its center $Z(Q_8)$. Using this fact, it is easy to show that $\alpha_x$ is a group automorphism.}. 
The action of any of these gates on  the product state
  $\ket{\mathcal{X}_1}\ket{C_1}$ generates an entangled  state; we show this explicitly for  $U_{\alpha_i}$:
\begin{equation*}
U_{\alpha_i}\ket{\mathcal{X}_1}\ket{C_1}= \tfrac{1}{2} \left( \tfrac{\ket{C_1} +\ket{C_{-1}}}{\sqrt{2}}+\ket{C_{i}}\right)\ket{C_1}+ \tfrac{1}{2}\left(\ket{C_{j}}+\ket{C_{k}}\right)\ket{C_{-1}}.
\end{equation*}
  Quaternionic quadratic-phase gates can also generate \textbf{highly entangled states}.
  For instance, the action of $D_\xi$ on a product state $\ket{\mathcal{X}_1}\ket{\mathcal{X}_1}$ creates an entangled bi-partite state with Schmidt rank 4, which is close to the maximal value of 5 achievable for a state in $\Comp_{\Conj{Q_8}}\otimes\Comp_{\Conj{Q_8}}$:
\begin{equation}\label{eq:Entangled State}
D_\xi\ket{\mathcal{X}_1}\ket{\mathcal{X}_1}= \tfrac{1}{4} \left( \left(\tfrac{\ket{C_1}+\ket{C_{-1}}}{\sqrt{2}}\right)\ket{\mathcal{X}_1}+\ket{C_{i}}\ket{\mathcal{X}_i}+\ket{C_{j}}\ket{\mathcal{X}_j}+\ket{C_{k}}\ket{\mathcal{X}_k}\right).
\end{equation}
A  quaternionic analogue of the ($d=4$) qudit \emph{cluster state} \cite{ZhouZengXuSun03} displaying  multi-partite entanglement can prepared by repeatedly  applying $D_\xi$ to all pairs of neighboring qudits on a lattice, chosen to be initially in the state $\ket{\mathcal{X}_1}$.

As this example shows, while normalizer circuits have fairly simple algebraic
properties, they can produce states that are very complicated and often highly
entangled. Thus, as in the abelian case, it comes as a surprise that these
circuits can often be classically simulated efficiently, as we will see in
section 6.

\section{A Hypergroup Stabilizer Formalism}\label{sect:Stabilizer Formailsm}

In this section we develop a stabilizer formalism based on  abelian hypergroups that extends Gottesman's PSF \cite{Gottesman_PhD_Thesis,Gottesman99_HeisenbergRepresentation_of_Q_Computers,Gottesman98Fault_Tolerant_QC_HigherDimensions} and the abelian group extension of \cite{VDNest_12_QFTs,BermejoVega_12_GKTheorem,BermejoLinVdN13_Infinite_Normalizers,BermejoLinVdN13_BlackBox_Normalizers}. We  apply our formalism to the description of new types of  quantum many-body states,  including hypergroup coset states and  those that appear at intermediate steps of quantum computations by normalizer circuits over hypergroups.

This section is organized as follows. In section \ref{sect:Pauli Operators}, we introduce new  types of Pauli operators based on hypergroups that have richer properties than  those of \cite{Gottesman_PhD_Thesis,Gottesman99_HeisenbergRepresentation_of_Q_Computers,Gottesman98Fault_Tolerant_QC_HigherDimensions,VDNest_12_QFTs,BermejoVega_12_GKTheorem,BermejoLinVdN13_Infinite_Normalizers,BermejoLinVdN13_BlackBox_Normalizers}: most remarkably, they can be  non-monomial and non-unitary matrices. In section (section \ref{sect:Stabilizer States over Hypergroups}) we show that commuting \emph{stabilizer hypergroups} built of the latter Paulis can be used to describe interesting families of quantum states, which we call \emph{hypergroup stabilizer states}, as well as track the dynamical evolution of hypergroup normalizer circuits (\textbf{theorem \ref{thm:Evolution of Stabilizer States}}). In section \ref{sect:Normal Form CSS hypergroup Stablizier States}, we give a  powerful \emph{normal form} (\textbf{theorem \ref{thm:Normal Form CSS States}}) for hypergroup stabilizer states that are CSS-like \cite{CalderbankShor_good_QEC_exist,Steane1996_Multiple_Particle_Interference_QuantumErrorCorrection}. The latter will be an invaluable tool in our paper, which we use to describe hypergroup coset states (equation \ref{eq:Hypergroup Coset State}, corollary \ref{corollary:Coset State Preparations}) and analyze the quantum algorithms of  section \ref{sect:Quantum Algorithms}. The techniques in this section will also be the basis of the classical simulation methods developed in section \ref{sect:Simulation}.

The fact that our hypergroup stabilizer formalism is based on non-monomial, non-unitary stabilizers introduces nontrivial technical difficulties that are discussed in detail in sections \ref{sect:Pauli Operators}-\ref{sect:Stabilizer States over Hypergroups}. The techniques we develop to cope with the issues are unique in the stabilizer formalism literature since both the original PSF and all of its previously known extensions \cite{Gottesman_PhD_Thesis,Gottesman99_HeisenbergRepresentation_of_Q_Computers,Gottesman98Fault_Tolerant_QC_HigherDimensions,VDNest_12_QFTs,BermejoVega_12_GKTheorem,BermejoLinVdN13_Infinite_Normalizers,BermejoLinVdN13_BlackBox_Normalizers, AaronsonGottesman04_Improved_Simul_stabilizer,dehaene_demoor_coefficients,dehaene_demoor_hostens,VdNest10_Classical_Simulation_GKT_SlightlyBeyond,deBeaudrap12_linearised_stabiliser_formalism,nest_MMS,NiBuerschaperVdNest14_XS_Stabilizers} were tailored to handle  unitary monomial stabilizer matrices. For this reason, we regard them as a main contribution of our paper.

Like  in  previous section, we develop our stabilizer formalism over arbitrary abelian hypergroups. Throughout the section, we fix $\mathcal{T} = \mathcal{T}_1 \times \dots \times \mathcal{T}_m$ to be an arbitrary finite abelian hypergroup with Hilbert space $·\Comp_{\mathcal{T}} \cong \Comp_{\mathcal{T}_1} \otimes \dots \otimes \Comp_{\mathcal{T}_m}$.

\subsection{Hypergroup Pauli operators}\label{sect:Pauli Operators}

We introduce generalized Pauli operators over $\mathcal{T}$ acting on $\Comp_{\mathcal{T}}$ with analogous properties to the qubit and abelian-group Pauli matrices \cite{VDNest_12_QFTs,BermejoVega_12_GKTheorem,BermejoLinVdN13_Infinite_Normalizers}. In our formalism Pauli operators perform operations associated to the abelian hypergroups of section \ref{sect: Circuit Model}. First, \emph{Pauli operators over a hypergroup $\mathcal{T}$} (\ref{eq:Pauli operators DEFINITION})  implement multiplications by hypergroup characters as well as the hypergroup operation of $\mathcal{T}$. For the character group $\mathcal{T^*}$, \emph{Pauli operators over $\mathcal{T^*}$} are defined analogously (\ref{eq:Pauli operators DUAL BASIS}).  More precisely, for all $x,y \in\mathcal{T}$, $\mathcal{X}_\mu, \mathcal{X}_\nu\in\mathcal{T^*}$, we define
\begin{align}\label{eq:Pauli operators DEFINITION}
 \PZ{\mathcal{T}}(\mathcal{X}_\mu)\ket{x}:=\Hchi_\mu(x)\ket{x},\qquad\:  & \PX{\mathcal{T}}(x)\ket{y}:= \sum_{z\in\mathcal{T}} \sqrt{\tfrac{\w{y}}{\w{z}}}\,  n_{x,y}^{z} \ket{z},\\
\label{eq:Pauli operators DUAL BASIS}
 \PZ{\mathcal{T^*}}(x)\ket{\mathcal{X}_\mu}:= \Hchi_\mu(x) \ket{\mathcal{X}_\mu}, \qquad\: & \PX{\mathcal{T^*}}(\mathcal{X}_\mu)\ket{\mathcal{X}_\nu}:= \sum_{\Hchi_\gamma\in\mathcal{T^*}} \sqrt{\tfrac{\w{\nu}}{\w{\gamma}}}\,  m_{\mu,\nu}^{\gamma} \ket{\mathcal{X}_\gamma}.
\end{align}
With this definition, Pauli operators over a product $\mathcal{T}=\mathcal{T}_1\times\cdots \times \mathcal{T}_m$ inherit a tensor product form $\PX{\mathcal{T}}(a):=\PX{\mathcal{T}_1}(a_1)\otimes \cdots \otimes \PX{\mathcal{T}_m}(a_m)$ and $\PZ{\mathcal{T}}(\mathcal{X}_\mu)=\PZ{\mathcal{T}_1}(\mathcal{X}_{\mu_1})\otimes \cdots \otimes \PZ{\mathcal{T}_m}(\mathcal{X}_{\mu_m})$. Any operator  that can be written as a product of operators of type $\PX{\mathcal{T}}(a)$ and $\PZ{\mathcal{T}}(\chi_\mu)$ will be called a generalized \emph{hypergroup Pauli operator}.

We  state a few \emph{main properties} of hypergroup Pauli operators. First, it follows from (\ref{eq:Pauli operators DEFINITION})  that the  Pauli operators $\PZ{\mathcal{T}}(\mathcal{X}_\mu)$ commute and form a hypergroup isomorphic to $\mathcal{T}^*$: 
\begin{equation}\label{eq:Pauli Z Properties}
\PZ{\mathcal{T}}(\mathcal{X}_\mu) \PZ{\mathcal{T}}(\mathcal{X}_\nu)= \PZ{\mathcal{T}}(\mathcal{X}_\nu) \PZ{\mathcal{T}}(\mathcal{X}_\mu)=\sum_{\gamma} m_{\mu\nu}^{\gamma} \PZ{\mathcal{T}}(\mathcal{X}_\gamma),  \quad  \textnormal{for any $\mathcal{X}_\mu,\mathcal{X}_\nu\in \mathcal{T}^*$}.
\end{equation}
Although it is not obvious from the definitions, we show later (theorem \ref{thm:Evolution of Stabilizer States}, eq.\ \ref{eq:QFTs are Clifford}) that the X-Paulis $\PX{\mathcal{T}}(a)$ are also pair-wise commuting normal operators, which are diagonal in the the character basis  $\{ \ket{\mathcal{X}_\mu}, \mathcal{X}_\mu \in \mathcal{T}^* \}$, and form a faithful representation of the conjugacy-class hypergroup. Precisely, for any $a,b\in \mathcal{T}$, $\mathcal{X}_\mu\in\mathcal{T}^*$, we have
\begin{equation}\label{eq:Pauli X Properties}
\PX{\mathcal{T}}(a)\ket{ \mathcal{X}_\mu }=\mathcal{X}_{\mu}(a)\ket{\mathcal{X}_\mu}, \qquad \PX{\mathcal{T}}(a)\PX{\mathcal{T}}(b)=\PX{\mathcal{T}}(b)\PX{\mathcal{T}}(a)=\sum_{c} n_{ab}^c \PX{\mathcal{T}}(c).
\end{equation}

The following lemma characterizes when Pauli operators of different type commute.

\begin{lemma}[\textbf{Commutativity}]\label{lemma:Commutativity of Paulis}
The  operators $\PX{\mathcal{T}}(a)$, $\PZ{\mathcal{T}}(\mathcal{X}_\mu)$ commute iff $\mathcal{X}_\mu(a)=1$.
\end{lemma}
\begin{proof}
For any state $\ket{b}$ in the basis $\mathsf{B}_\mathcal{T}$, compare $\PZ{\mathcal{T}}(\mathcal{X}_\mu)\PX{\mathcal{T}}(a)\ket{b}=\sum_{c}\sqrt{\tfrac{\w{b}}{\w{c}}}n_{ab}^c \mathcal{X}_\mu(c)\ket{c}$ with $ \PX{\mathcal{T}}(a)\PZ{\mathcal{T}}(\mathcal{X}_\mu)\ket{b}=\sum_{c}\sqrt{\tfrac{\w{b}}{\w{c}}}n_{ab}^c \mathcal{X}_\mu(b)\ket{c}$. Then,  the ``if'' holds  because  $\mathcal{X}_\mu(a)=1$  implies  $\mathcal{X}_\mu(b)=\mathcal{X}_\mu(c)$ for any $b\in\mathcal{T}$, $c\in ab$ \cite[proposition 2.4.15]{BloomHeyer95_Harmonic_analysis}, and the two expressions coincide. Conversely, letting $b=e$ in these two expressions yields  the ``only if''.
\end{proof}
Note that the above properties are always fulfilled by qubit \cite{Gottesman99_HeisenbergRepresentation_of_Q_Computers}, qudit \cite{Gottesman98Fault_Tolerant_QC_HigherDimensions} and group Pauli operators \cite{VDNest_12_QFTs,BermejoVega_12_GKTheorem}. In this sense, hypergroup Pauli operators provide a generalization of these concepts. Yet, as we will see in the next section, there are some remarkable properties of group Pauli operators that are fully shared by their hypergroup counterparts.

\subsection{Hypergroup stabilizer states}\label{sect:Stabilizer States over Hypergroups}

We will now extend the notion of stabilizer state from groups to hypergroups.

\begin{definition}[\textbf{Stabilizer hypergroup and stabilizer state}]\label{def:Stabilizer Hypergroup} \label{def:Stabilizer State}

A \emph{stabilizer hypergroup} $\mathcal{S}^\lambda$ is a hypergroup of commuting  hypergroup Pauli operators over $\mathcal{T}$ (\ref{eq:Pauli operators DEFINITION}-\ref{eq:Pauli operators DUAL BASIS}) with an associated \emph{stabilizer function} $\lambda$ that selects an eigenvalue $\lambda(U)$ for every  $U$ in $\mathcal{S}^\lambda$.

Let $\{\mathcal{S}_i^{\lambda_i}\}_{i=1}^{r}$ be a \emph{collection} of $r$ mutually commuting stabilizer hypergroups. Then, a  quantum state $\ket{\psi}$ is called a \emph{stabilizer state} stabilized by   $\{\mathcal{S}_i^{\lambda_i}\}_{i=1}^{r}$ if, up to normalization and global phases, it is the unique non-zero solution to the system of spectral equations
\begin{equation}\label{eq:Stabilizer State Condition}
U\ket{\psi}= \lambda_i(U) \ket{\psi},\quad\textnormal{for all }U\in\mathcal{S}_i^{\lambda_i},\, i=1,\ldots, r.
\end{equation} 
By definition, every function $\lambda_i$ is further constrained to be a \emph{character} of $\mathcal{S}_i^{\lambda_i}$. This is necessary for the system  (\ref{eq:Stabilizer State Condition}) to admit nontrivial solutions.\footnote{Let $UV=\sum_{W}s_{UV}^{W} W$ for any $U,V\in\mathcal{S}_i^{\lambda_i}$ with   structure constants  $s_{U,V}^{W}$. Then (\ref{eq:Stabilizer State Condition}) implies  $UV\ket{\psi}=\lambda(U)\lambda(V)\ket{\psi}=\sum_{W}s_{UV}^{W} \lambda(W)\ket{\psi}$. Since $\ket{\psi}\neq 0$, every non-zero $\lambda$ must be a character (by definition).}
\end{definition}
In this work we focus on \emph{pure} stabilizer states and do not discuss mixed ones. 

Though, in the PSF and in the group GSF \cite{Gottesman_PhD_Thesis,Gottesman99_HeisenbergRepresentation_of_Q_Computers,Gottesman98Fault_Tolerant_QC_HigherDimensions,VDNest_12_QFTs,BermejoVega_12_GKTheorem,BermejoLinVdN13_Infinite_Normalizers,BermejoLinVdN13_BlackBox_Normalizers}, stabilizer groups can always be described efficiently in terms of generators or matrix representations of morphisms, the existence of efficient descriptions is hard to prove  in the hypergroup setting. The stabilizer hypergroups $\mathcal{S}^\lambda$ in our paper (see section \ref{sect:Normal Form CSS hypergroup Stablizier States}) have efficient  $\polylog{|\mathcal{T}|}$-size classical descriptions (where $|\mathcal{T}|$ is  the  dimension of the  Hilbert space $\mathcal{H_T}$) if poly-size descriptions for the subhypergroups of $\mathcal{T}$ and $\mathcal{T^*}$ are promised to exist. We highlight that this latter condition is fulfilled for many hypergroups of interest, including  conjugacy class and character hypergroups, and anyonic fusion rule theories\footnote{All examples mentioned belong to a class of so-called \emph{resonance} hypergroups $\mathcal{T}$ that  have integral weights, $\w{a}$, and fulfill a Lagrange theorem \cite{Wildberger_Lagrange}, which says that $\varpi_\mathcal{N}$ for any subhypergroup $\mathcal{N\subset T}$ always divides $\varpi_\mathcal{T}$. As in the finite group case \cite{VDNest_12_QFTs,BermejoVega_12_GKTheorem}, this implies that a random set $\{a_1,\ldots,a_m\}\subset\mathcal{T}$ generates $\mathcal{T}$ via hyperoperations with  high probability $\Omega(1-\tfrac{1}{2^m})$.}. The stabilizer hypergroups obtained from all those cases provide a powerful means to describe quantum many-body states that are uniquely defined via an equation of the form (\ref{eq:Stabilizer State Condition}).

The definition of stabilizer hypergroup and state generalizes the standard notions used in the PSF and the Group Stabilizer Formalism. At the same time, our hypergroup Paulis  also have novel interesting  properties that are  explained next.

\begin{comparison}[\textbf{Relationship to group stabilizer states.}]
All qubit, qudit, and abelian-group  stabilizer states \cite{Gottesman_PhD_Thesis,Gottesman99_HeisenbergRepresentation_of_Q_Computers,Gottesman98Fault_Tolerant_QC_HigherDimensions,VDNest_12_QFTs,BermejoVega_12_GKTheorem} are instances of hypergroup stabilizer states over an abelian hypergroup $\overline{G}$, where $G$ chosen to be a group of the form $\Integers_2^m$, $\Integers_{d}^m$, and $\Integers_{N_1}\times \cdots\times \Integers_{N_m}$ (respectively) with $\lambda(U)=+1$. Hypergroup Pauli operators and stabilizer hypergroups over $\overline{G}$ (\ref{eq:Pauli operators DEFINITION}) become standard Pauli operators and stabilizer groups over $G$ (in the notation of \cite{BermejoVega_12_GKTheorem}; cf.\ section 3 therein).
\end{comparison}

\begin{comparison}[\textbf{Subtleties of hypergroup Pauli operators.}] Interestingly, in spite of having some  Pauli-like mathematical properties (section \ref{sect:Pauli Operators}),  hypergroup Pauli operators are \emph{not as simple} as group Pauli operators \cite{Gottesman_PhD_Thesis,Gottesman99_HeisenbergRepresentation_of_Q_Computers,Gottesman98Fault_Tolerant_QC_HigherDimensions,VDNest_12_QFTs,BermejoVega_12_GKTheorem}: namely, they are not necessarily \emph{unitary} and no longer \emph{monomial} (1-sparse) matrices because the hyperoperations in (\ref{eq:Pauli operators DEFINITION})  are non-invertible and return multiple outcomes. The absence of these two properties is reflected in definition \ref{def:Stabilizer Hypergroup}. In our formalism,  stabilizer \emph{hypergroups} $\mathcal{S}^\lambda$ of commuting Paulis do not  necessarily form \emph{groups}. Stabilizer states are no longer restricted to be $+1$-eigenstates of hypergroup Pauli operators $U\in\mathcal{S}^\lambda$, as in \cite{Gottesman_PhD_Thesis,Gottesman99_HeisenbergRepresentation_of_Q_Computers,Gottesman98Fault_Tolerant_QC_HigherDimensions,VDNest_12_QFTs,BermejoVega_12_GKTheorem}, for Pauli operators can  now have zero eigenvalues\footnote{This follows from (\ref{eq:Pauli operators DEFINITION} because nonabelian group  and hypergroup characters can take zero values \cite{Wildberger97_Duality_Hypergroups,Amini2011fourier}.}. This allows us to include more states in the formalism.
\end{comparison}
\begin{comparison}[\textbf{Non-commutativity up to a phase.}]
Hypergroup Paulis do not satisfy an identity of the form $\PZ{\mathcal{T}}(\mathcal{X}_\mu)\PX{\mathcal{T}}(a)= \mathcal{X}_\mu(a) \PX{\mathcal{T}}(a)\PZ{\mathcal{T}}(\mathcal{X}_\mu)$ in general, although this is the case when $\mathcal{T}$ is a group \cite{VDNest_12_QFTs,BermejoVega_12_GKTheorem}. In our setting, this is fulfilled only in some special cases, e.g., when either $a$ or $\mathcal{X}_\mu$ is an invertible hypergroup element  (theorem \ref{thm:Evolution of Stabilizer States}, eq.\ \ref{eq:Pauli gates are Clifford}) or when $\mathcal{X}_\mu(a){=}1$ (where $\PZ{\mathcal{T}}(\mathcal{X}_\mu)$ and $\PX{\mathcal{T}}(a)$ \emph{commute} due to lemma \ref{lemma:Commutativity of Paulis}). 
\end{comparison}
\begin{comparison}[\textbf{Multiple stabilizer hypergroups}] The reason why we use multiple stabilizer hypergroups $\{\mathcal{S}_i^{\lambda_i}\}$ instead of merging them into a single commutative algebra is that finding a basis with hypergroup structure for the latter object is not a simple  task\footnote{We have not investigated this problem nor whether a hypergroup basis can always be found.}. On the other hand, we can easily keep track and exploit the available  hypergroup structures by simply storing a poly-size list of pairwise commuting stabilizer hypergroups.\footnote{Throughout the paper $r$ will always be poly-sized and all  examples  we give (section \ref{sect:Normal Form CSS hypergroup Stablizier States}) have $r\leq 2$.}
\end{comparison}
The next two results imply that any intermediate quantum state of a hypergroup normalizer circuit (section \ref{sect: Circuit Model}) computation is a hypergroup stabilizer state and, hence, can be characterized concisely as a joint-eigenstate of some commuting hypergroup Pauli operators. This observation motivates our further development of these concepts.

\begin{claim}[\textbf{Standard basis states}]\label{lemma:Standard Basis States are Stabilizer States}  Conjugacy-class and character states (\ref{eq:Quantum Fourier Transform over Hypergroup T}) are instances of hypergroup stabilizer states stabilized by \emph{single} stabilizer hypergroups.
\end{claim} 

\begin{theorem}[\textbf{Evolution of stabilizer states.}]\label{thm:Evolution of Stabilizer States}
Normalizer gates map hypergroup Pauli operators to new hypergroup Pauli operators under conjugation and, therefore, transform hypergroup stabilizer states into stabilizer states. It follows from this and the previous claim that the quantum state of a normalizer circuit is always a stabilizer state. 
\end{theorem}
Theorem \ref{thm:Evolution of Stabilizer States} extends a theorem of Van den Nest \cite{VDNest_12_QFTs} for abelian group stabilizer states.

In order to prove claim \ref{lemma:Standard Basis States are Stabilizer States}, theorem \ref{thm:Evolution of Stabilizer States}, and  many of the main results  in the next sections, we will develop a new kind of hypergroup  stabilizer formalism techniques that can cope with the \emph{\textbf{non-monomiality}} and the \emph{\textbf{non-unitarity}} of hypergroup Pauli operators. A central part of the paper will be dedicated exclusively to this end. We stress the necessity to develop such techniques since currently available stabilizer-formalism methods---including the PSF \cite{Gottesman_PhD_Thesis,Gottesman99_HeisenbergRepresentation_of_Q_Computers,Gottesman98Fault_Tolerant_QC_HigherDimensions,AaronsonGottesman04_Improved_Simul_stabilizer,dehaene_demoor_coefficients,dehaene_demoor_hostens,VdNest10_Classical_Simulation_GKT_SlightlyBeyond,deBeaudrap12_linearised_stabiliser_formalism}, the Group Stabilizer Formalism \cite{VDNest_12_QFTs,BermejoVega_12_GKTheorem,BermejoLinVdN13_Infinite_Normalizers,BermejoLinVdN13_BlackBox_Normalizers},  the general  Monomial Stabilizer Formalism of Van den Nest \cite{nest_MMS}, and the recent  XS stabilizer formalism \cite{NiBuerschaperVdNest14_XS_Stabilizers}---can not be applied in our setting as they \emph{critically exploit  the monomiality/unitarity of stabilizer operators} for central tasks such as simulating Clifford/Normalizer operations, analyzing stabilizer state and code properties (e.g., code dimension,  code support), and giving normal forms for stabilizer states\footnote{The role of monomiality and unitarity in the PSF has been extensively discussed in \cite{nest_MMS}.}. The lack of these beneficial properties requires a change of paradigm in our setting. 

To prove our claim \ref{lemma:Standard Basis States are Stabilizer States}, we will show a stronger result (\textbf{theorem \ref{thm:Normal Form CSS States}} below), which gives a \emph{normal form} for  hypergroup stabilizer states; we outline the proof of theorem \ref{thm:Evolution of Stabilizer States} below, with details referred to   appendix \ref{app:A}.

\begin{proof}[Proof of theorem \ref{thm:Evolution of Stabilizer States}, part I] We show that normalizer gates  transform X- and Z-type Pauli operators  over $\mathcal{T}$  into new  Pauli operators (which may involve products of X, Z Paulis) under conjugation. This result extends to arbitrary products of these operators. (Compare to the abelian-group case \cite{VDNest_12_QFTs}.)

Specifically, for any invertible element $s\in\mathcal{T}$, invertible character $\mathcal{X}_\varsigma\in\mathcal{T^*}$, automorphism $\alpha$, quadratic function $\xi$, we can calculate this action for the normalizer gates $Z_\mathcal{T}(\mathcal{X}_\mu)X_\mathcal{T}(a)$, $U_\alpha$, $D_\xi$, and the hypergroup QFT:
\begin{align}
\PX{\mathcal{T}}(a)& \xrightarrow{\:Z_\mathcal{T}(\mathcal{X}_\varsigma)X_\mathcal{T}(s)\:} \mathcal{X}_{{\varsigma}}(a) \PX{\mathcal{T}}(a),& \PZ{\mathcal{T}}(\mathcal{X_\mu}) &\xrightarrow{\:Z_\mathcal{T}(\mathcal{X}_\varsigma)X_\mathcal{T}(s)\:} \mathcal{X}_{{\mu}}(\overline{s})\PZ{\mathcal{T}}(\mathcal{X_\mu}),
\label{eq:Pauli gates are Clifford}\\
\PX{\mathcal{T}}(a)&\xrightarrow{\:U_\alpha\:}  \PX{\mathcal{T}}(\alpha(a)), &  \PZ{\mathcal{T}}(\mathcal{X_\mu})&\xrightarrow{\:U_\alpha\:}  \PZ{\mathcal{T}}(\mathcal{\PX{\alpha^{-*}(\mu)}} ), \label{eq:Automorphism gates are Clifford} \\
\PX{\mathcal{T}}(a)&\xrightarrow{\:D_\xi\:}  \xi(a) \PX{\mathcal{T}}(a) \PZ{\mathcal{T}}(\beta{(a)}), &   \PZ{\mathcal{T}}(\mathcal{\PX{\mu}})&\xrightarrow{\:D_\xi\:}   \PZ{\mathcal{T}}(\mathcal{\PX{\mu}}), \label{eq:Quadratic Phase gates are Clifford}\\
\PX{\mathcal{T}}(a)&\xrightarrow{\:\mathrm{QFT}\:}  \PZ{\mathcal{T}^*}(a), &  \PZ{\mathcal{T}}(\mathcal{X_\mu})&\xrightarrow{\:\mathrm{QFT}\:} \PX{\mathcal{T}^*}(\overline{\Hchi_\mu}).\label{eq:QFTs are Clifford}
\end{align}
When $\Comp_{\mathcal{T}}=\Comp_{\mathcal{T}_1}\otimes\cdots \otimes \Comp_{\mathcal{T}_m}$, the partial QFT over $\mathcal{T}_i$ simply transforms the $i$th tensor factor of the  Pauli operators according to  (\ref{eq:QFTs are Clifford}). In (\ref{eq:Quadratic Phase gates are Clifford}),  $\beta$ is a homomorphism from $\mathcal{T}$ to the subhypergroup of invertible characters $\mathcal{T}_\mathrm{inv}^*$ that depends on $\xi$; in (\ref{eq:Automorphism gates are Clifford}), $\alpha^{-*}$ is the inverse of the \emph{dual automorphism $\alpha^*$} \cite{McMullen79_Algebraic_Theory_Hypergroups}:

\begin{definition}[\textbf{Dual automorphism \cite{McMullen79_Algebraic_Theory_Hypergroups}}]  \label{eq:Dual Automorphism DEFINITION}
The  \emph{dual automorphism} of $\alpha$, denoted  $\alpha^*$, is the automorphism of $\mathcal{T}^*$  that takes $\mathcal{X}_\mu$ to the  character  $\mathcal{X}_{\alpha^{*}(\mu)} := \mathcal{X}_{\mu} \circ \alpha$  for fixed $\mathcal{X}_{\mu}$.\footnote{This is a morphism because $\mathcal{X}_{\alpha^*(\mu)}\mathcal{X}_{\alpha^*(\nu)}{=}\left(\mathcal{X}_{\mu}\circ\alpha\right)\left(\mathcal{X}_{\nu}\circ\alpha\right){=}\sum_{\gamma}m_{\mu\nu}^\gamma \left(\mathcal{X}_{\gamma}\circ \alpha\right){=} \sum_{\gamma}m_{\mu\nu}^\gamma\,\mathcal{X}_{\alpha^*(\gamma)}$.}
\end{definition}
Proving (\ref{eq:Pauli gates are Clifford}-\ref{eq:QFTs are Clifford}) involves  bulky yet beautifully structured hypergroup calculations that are carried out in appendix \ref{app:A}.
\end{proof}
It is worth noting  that  normalizer gates transform Pauli operators over $\mathcal{T}$ into Pauli operators over  $\mathcal{T}$ if they are not QFTs and into Pauli operators over  $\mathcal{T}^*$ otherwise\footnote{Recall from section \ref{sect: Circuit Model} that the hypergroup label $\mathcal{T}$ keeps track of the basis in which basis (\ref{eq:Quantum Fourier Transform over Hypergroup T}) measurements are performed and indicates which Pauli operators are diagonal in each basis.}.

\subsection{A normal form for stabilizer states and examples}
\label{sect:Normal Form CSS hypergroup Stablizier States}
In this final subsection we give examples and a normal form for a class of hypergroup states that generalize the well-known notion of Calderbank-Shor-Steane (CSS) stabilizer states \cite{CalderbankShor_good_QEC_exist,Steane1996_Multiple_Particle_Interference_QuantumErrorCorrection,Calderbank97_QEC_Orthogonal_Geometry}:

\begin{definition}[\textbf{CSS stabilizer state}]\label{def:CSS state}
A hypergroup stabilizer state $\ket{\psi}$ over $\mathcal{T}$  is said of CSS type if it is uniquely stabilized by two mutually commuting stabilizer hypergroups $\mathcal{S}_Z^{\lambda_z}$, $\mathcal{S}_X^{\lambda_x}$  consisting only of Z and X Pauli operators respectively.
\end{definition}
The standard definition of CSS state is recovered by setting $\mathcal{T}=\Integers_2^{n}$ (cf.\ the example sections in \cite{VDNest_12_QFTs,BermejoVega_12_GKTheorem}). For the sake of brevity, we assume  $\mathcal{S}_Z^{\lambda_z}$ and  $\mathcal{S}_X^{\lambda_x}$ to be maximal mutually commuting hypergroups in this section\footnote{This maximality assumption is not necessary in our derivation but it shortens the proofs.}. With these requirements, lemma \ref{lemma:Commutativity of Paulis}  imposes that $\mathcal{S}_Z^{\lambda_z}$,  $\mathcal{S}_X^{\lambda_x}$ must be of the following form: 
\begin{align}
\mathcal{S}_X^{\lambda_x}&=\{\PX{\mathcal{T}}(a), a\in \mathcal{N}\}, & \lambda_x(\PX{\mathcal{T}}(a))&=\mathcal{X}_\varsigma({a}) \notag \\ \mathcal{S}_Z^{\lambda_z}&=\{\PZ{\mathcal{T}}(\mathcal{X}_\mu), \mathcal{X}_\mu\in \mathcal{N}^\perp\}, & \lambda_z(\PZ{\mathcal{T}}(\mathcal{X}_\mu))&= \mathcal{X}_\mu(s),\label{eq:Stabilizer Hypergroup CSS State}
\end{align}
where $s\in\mathcal{T}$, $\mathcal{X}_\varsigma\in\mathcal{T}^*$, $\mathcal{N} \le \mathcal{T}$ is a subhypergroup, and $\mathcal{N}^\perp$ is the \emph{annihilator} of $\mathcal{N}$ (\ref{eq:annhilator}).

We are now ready to prove the main technical result of this section, theorem \ref{thm:Normal Form CSS States}, which  characterizes the set of CSS stabilizer states and leads to specific examples.

\begin{theorem}[\textbf{Normal forms for CSS states}]\label{thm:Normal Form CSS States}
\textbf{(a)} The quantum states stabilized by  $\{\mathcal{S}_X^{\lambda_x},\mathcal{S}_Z^{\lambda_z}\}$ from (\ref{eq:Stabilizer Hypergroup CSS State}) are those in the subspace
\[ \mathcal{V_S}:=\mathrm{span}\{\psi_y,\,y{\,\in\,} s\mathcal{N}\}=\mathrm{span}\{\widehat{\psi}_\nu,\,\mathcal{X}_\nu{\,\in\,}\mathcal{X_\varsigma N^\perp}\}, \]
where $\psi_y$ and $\widehat{\psi}_\nu$ are functions supported on $s\mathcal{N}$ and $\mathcal{X_\varsigma N^\perp}$, respectively, and defined by
\begin{equation}\label{eq:Consistency CSS}
\psi_y(x):= \sqrt{\w{x}} \left(\sum_{\substack{b}\in\mathcal{N} }n_{ x, {\overline{y}}}^{ b}\mathcal{X}_{\overline{\varsigma}}(b)\right)\ \text{ and }\ \ \widehat{\psi}_\nu(\mu):=\sqrt{\w{\mathcal{X}_\mu}} \left(\sum_{\substack{\beta}\in\mathcal{N}^\perp }m_{ \mu, {\overline{\nu}}}^{ \beta}\mathcal{X}_{{\beta}}(s)\right).
\end{equation}
A state $\ket{\psi}$ is \emph{uniquely} stabilized by $\{\mathcal{S}_X^{\lambda_x},\mathcal{S}_Z^{\lambda_z}\}$ iff $\mathrm{dim}(\mathcal{V_S})=1$. \textbf{(b)} Furthermore, if either $s$ or $\mathcal{X}_\varsigma$ is invertible, then $\{\mathcal{S}_X^{\lambda_x},\mathcal{S}_Z^{\lambda_z}\}$  stabilizes a \emph{unique} state  of form\footnote{Cf.\ section \ref{sect:Subhypergroups, Quotients} for definitions of  $\w{\Hchi_{\overline{\varsigma}} \mathcal{N}^\perp}, \ws{s \mathcal{N}},\varpi_{s\mathcal{N}},\varpi_{\varsigma\mathcal{N}^\perp}$.} \begin{align}\label{eq:Normal Form CSS state}
\ket{\psi}&=\sum_{x\in s \mathcal{N}} \sqrt{\frac{\w{x} \w{\Hchi_{\overline{\varsigma}} \mathcal{N}^\perp} }{\varpi_{s\mathcal{N}}}}\, \mathcal{X}_{\overline{\varsigma}}(x) \ket{x}, \qquad \Fourier{\mathcal{T}}\ket{\psi} = \sum_{\mathcal{X}_\mu \in \mathcal{X}_\varsigma \mathcal{N}^\perp} \sqrt{\frac{\w{\mathcal{X}_\mu} \ws{s\mathcal{N}}}{\varpi_{\varsigma \mathcal{N}^\perp}}}\, \mathcal{X}_\mu(s) \ket{\mathcal{X}_\mu}.
\end{align}
\end{theorem}
Theorem \ref{thm:Normal Form CSS States} is proven at the end of the section after mentioning a few main applications.

We highlight that, despite our focus on stabilizer \emph{states}, theorem \ref{thm:Normal Form CSS States} can be easily extended to study hypergroup stabilizer \emph{codes}. For instance, in case \textbf{(b)}, we could choose a smaller stabilizer hypergroup $\mathcal{S}_X^{\lambda_x}=\{\PX{\mathcal{T}}(a), a\in \mathcal{K}\}$ with $\mathcal{K\subsetneq N}$ to obtain an  stabilizer code $\mathcal{V_S}$, whose dimension is easy to compute with our techniques\footnote{With minor modifications of our proof of theorem \ref{thm:Normal Form CSS States}, we get that the dimension is the number of cosets of $\mathcal{K}$ inside $s\mathcal{N}$, if $\mathcal{X_\varsigma}$ is invertible, and the number of cosets of $\mathcal{N}^\perp$ inside $\mathcal{X_\varsigma K^\perp}$, if $s$ is invertible.}. Similarly, one could shrink $\mathcal{S}_Z^{\lambda_z}$ or both stabilizer hypergroups at the same time. 

\paragraph{Open question}  The hypergroup CSS code construction we outlined clearly mimics the standard qubit one \cite{CalderbankShor_good_QEC_exist,Steane1996_Multiple_Particle_Interference_QuantumErrorCorrection,Calderbank97_QEC_Orthogonal_Geometry,nielsen_chuang}. Interestingly, there could also be  hypergroup CSS codes with no qubit/qudit analogue if there exist groups (or even hypergroups $\mathcal{T}$ that do not arise from groups) for which $\mathcal{V_S}$ in theorem \ref{thm:Normal Form CSS States}\textbf{(a)} can be degenerate. Though this is never the case for abelian groups because the Pauli operators in (\ref{eq:Stabilizer Hypergroup CSS State}) generate a maximal lin.\ ind.\ set of commuting operators (with  rank one common eigenprojectors), it can happen in our setting.\footnote{Products of hypergroup Pauli operators  in (\ref{eq:Stabilizer Hypergroup CSS State})  can be linearly dependent (choose $\mathcal{N}= \{\pm1,C_i\}$ for the quaternion group, section \ref{sect: Circuit Model}) and their cardinality $|\mathcal{N}||\mathcal{N^\perp}|$ may not match the dimension of $\mathcal{H}_{\overline{G}}$ \cite{Ichihara_thesis_Hypergroup_Extensions}.} We leave open the question of whether such codes exist. 

\subsubsection*{Examples and applications of theorem \ref{thm:Normal Form CSS States}}\label{sect:Examples Normal Form CSS States}

Theorem \ref{thm:Normal Form CSS States} is an important technical contribution of this work that will be used three times within the scope of the paper: firstly, in the examples below,  to construct efficient classical descriptions for new kinds of complex many body states; secondly, in section \ref{sect:Simulation}, to devise classical  algorithms for simulating hypergroup normalizer circuits (theorems \ref{thm:simulation}); and finally, in section \ref{sect:Quantum Algorithms}, in the development of an efficient quantum algorithm for hidden subhypergroup problems (theorems \ref{thm:easy2}--\ref{thm:easy3}). 

We now give examples of CSS hypergroup stabilizer states  of the simpler form (\ref{eq:Normal Form CSS state}).

\paragraph{Example 1: standard basis states}

We show now that conjugacy-class states and character states (\ref{eq:Quantum Fourier Transform over Hypergroup T}) are instances of hypergroup CSS states (of type (b)), as anticipated above (claim 1). Consider, first, an arbitrary $\ket{a}$ with $a\in \mathcal{T}$. Eq.\ (\ref{eq:Pauli Z Properties}) implies that $\ket{\psi}$ is stabilized by    $\mathcal{S}_Z^{\lambda_z}=\{Z_\mathcal{T}(\mathcal{X}_\mu), {X}_\mu\in\mathcal{T^*}\}$ with maximal $\mathcal{N^\perp}=\mathcal{T^*}$ and $\lambda_z(Z_\mathcal{T}(\mathcal{X}_\mu))=\mathcal{X}_\mu(a)$. Letting $\mathcal{S}_X^{\lambda_x}$, $\mathcal{N}$, and $\lambda_x$  be trivial, the state can written in the form given in theorem \ref{thm:Normal Form CSS States}.{(b)} (note that $\lambda_x$ is an invertible character) and, hence, $\ket{a}$ is a uniquely stabilized  CSS state. An almost identical argument, using  (\ref{eq:Pauli X Properties}) instead, shows that any character state $\ket{\mathcal{X_\nu}}$ is uniquely stabilized by  $\mathcal{S}_X^{\lambda_x} = \{X_\mathcal{T}(a),a\in\mathcal{T}\}$ with $\lambda_x(X_\mathcal{T}(a))=\mathcal{X}_\nu(a)$. (Note that in both cases equation (\ref{eq:Normal Form CSS state}) reproduces (\ref{eq:Quantum Fourier Transform over Hypergroup T}) consistently. )

\paragraph{Example 2: hypergroup coset states}

The states in example 1 are always product states. Yet theorem \ref{thm:Normal Form CSS States} also implies that highly entangled states such as the abelian group coset states that appear in the abelian HSP quantum algorithms \cite{childs_vandam_10_qu_algorithms_algebraic_problems},
\begin{equation}\label{eq:Abelian Group Coset State}
\ket{x+H}=\sum_{h\in H} \frac{1}{\sqrt{|H|}}\ket{x+h}, \quad \textnormal{$H$ is a subgroup of a finite abelian group $G$},
\end{equation}
as well as the abelian hypergroup coset states used in the quantum algorithm \cite{Amini_hiddensub-hypergroup},
\begin{equation}\label{eq:Hypergroup Coset State}
\ket{s\mathcal{N}}=\sum_{x\in s \mathcal{N}} \sqrt{\frac{\w{x}}{\varpi_{s\mathcal{N}}}}\, \ket{x}, \quad \textnormal{$\mathcal{N}$ is a subhypergroup of a finite abelian hypergroup $\mathcal{T}$},
\end{equation}
are all instances of CSS hypergroup states of type (a) with trivial $\mathcal{X}_\varsigma$, uniquely stabilized by $\mathcal{S}_X^{\lambda_x}=\{X_\mathcal{T}(a),a\in\mathcal{N}\}$,  $\mathcal{S}_Z^{\lambda_z}=\{Z_\mathcal{T}(\mathcal{X}_\mu),\mathcal{X}_\mu\in\mathcal{N}^\perp\}$, $\lambda_z(Z_\mathcal{T}(\mathcal{X}_\mu))= \mathcal{X}_\mu(s)$, and trivial, invertible $\lambda_x$.

In the special case of  abelian group coset states  (\ref{eq:Abelian Group Coset State}), theorem \ref{thm:Normal Form CSS States} recovers a result by Van den Nest \cite{VDNest_12_QFTs} who identified the latter to generalized abelian group stabilizer states. Theorem \ref{thm:Normal Form CSS States} extends the latter result demonstrating the existence of complex many-body states---hypergroup coset states (\ref{eq:Hypergroup Coset State}) and more (\ref{eq:Consistency CSS}-\ref{eq:Normal Form CSS state})---that can be  described within the present Hypergroup Stabilizer Formalism but not within the standard PSF and its abelian-group extensions \cite{Gottesman_PhD_Thesis,Gottesman99_HeisenbergRepresentation_of_Q_Computers,Gottesman98Fault_Tolerant_QC_HigherDimensions,VDNest_12_QFTs,BermejoVega_12_GKTheorem}, or even (to the  best of our knowledge\footnote{The authors are not aware of any method to express hypergroup coset states in terms of monomial unitary stabilizers as in \cite{Gottesman_PhD_Thesis,Gottesman99_HeisenbergRepresentation_of_Q_Computers,Gottesman98Fault_Tolerant_QC_HigherDimensions,VDNest_12_QFTs,BermejoVega_12_GKTheorem,nest_MMS,NiBuerschaperVdNest14_XS_Stabilizers}. We doubt such a  description could exist and, at the same time, be easy to track under the action of hypergroup normalizer circuits as in theorem \ref{thm:Evolution of Stabilizer States}.}) or within other generalizations of the PSF such as the Monomial Stabilizer Formalism \cite{nest_MMS} and the X-S Stabilizer Formalism \cite{NiBuerschaperVdNest14_XS_Stabilizers}.

\subsubsection*{Proof of theorem \ref{thm:Normal Form CSS States}}

We finish this section by proving theorem \ref{thm:Normal Form CSS States} and giving a method for preparing coset states as a corollary  (corollary \ref{corollary:Coset State Preparations}). As announced in the previous section, the proof of this result relies on new technical ideas based on  hypergroup methods, which are needed to handle the \emph{non-unitary, non-monomial} stabilizers of theorem \ref{thm:Normal Form CSS States}.
 
\begin{proof}[Proof of theorem \ref{thm:Normal Form CSS States}] First note that the properties discussed in section \ref{sect:Pauli Operators} show that both stabilizer  hypergroups $\mathcal{S}_X^{\lambda_x}$, $\mathcal{S}_X^{\lambda_z}$ are well-defined. To prove the theorem, we will use some basic hypergroup theoretic results described in the following lemma.
\begin{lemma}[{\cite[2.4.15,2.4.16]{BloomHeyer95_Harmonic_analysis}}]
\label{lemma:Quotient-Annihilator Isomorphisms} For any subhypergroup $\mathcal{N}$, the hypergroup isomorphisms  $\mathcal{N}^*\cong\mathcal{T}^*/\mathcal{N}^\perp$ and $(\mathcal{T/N})^*\cong \mathcal{N}^\perp$ (cf.\ section \ref{sect:Glossary}), can be canonically realized as follows.
\begin{itemize}
\item[\textbf{(i)}] All characters of $\mathcal{N}$ are obtained via \emph{restriction} of characters of $\mathcal{T}$, and two characters $\mathcal{X}_\alpha,\mathcal{X}_\beta\in\mathcal{T}^*$ act \emph{equally} on $\mathcal{N}$ if and only if $\mathcal{X}_\alpha\in\mathcal{X}_\beta\mathcal{N}^\perp$. 
\item[\textbf{(ii)}]  All quotient characters are obtained by letting characters in $\mathcal{N}^\perp$ act on cosets $x\mathcal{N}$, and this map is well-defined because  the former act \emph{constantly} on the latter. 
\end{itemize}
\end{lemma}
Next, we identify  necessary and sufficient  conditions for the a state $\ket{\psi}$ to be uniquely stabilized by $\{\mathcal{S}_X^{\lambda_x},\mathcal{S}_Z^{\lambda_z}\}$. First, condition (\ref{eq:Stabilizer State Condition}) says  that  $\ket{\psi}$ is stabilized by $\mathcal{S}_Z^{\lambda_z}$ iff
\begin{align}
\PZ{\mathcal{T}}(\mathcal{X}_{\mu})\ket{\psi_{}} &= \mathcal{X}_{\mu}( s )\ket{\psi_{}} =\lambda_{z}(\PZ{\mathcal{T}}(\mathcal{X}_{\nu}))\ket{\psi} \quad\textnormal{for every } \mathcal{X}_\mu \in \mathcal{N}^\perp.\label{inproof:CSS state is stabilized 1}
\end{align}
Due to lemma \ref{lemma:Quotient-Annihilator Isomorphisms}(ii), this holds iff the wavefunction $\psi$ is supported on a subset of the coset $s\mathcal{N}$. On the other hand, we show that $\ket{\psi}$ is further stabilized by $\mathcal{S}_X^{\lambda_x}$ iff $\psi$ belongs to the   image of the following operator:
\begin{equation}
 P_X:=\varw{\mathcal{N}}^{-1}\sum_{ b \in \mathcal{N}}\w{\Hchi_{\overline{\varsigma}}\mathcal{N}^\perp} \w{b}\mathcal{X}_{\overline{\varsigma}}({b}) \PX{\mathcal{T}}( b).
\end{equation}
The ``only if'' follows from the fact that $X_{\mathcal{T}}(b)\ket{\psi} = \Hchi_\varsigma(b)$ and the orthogonality relationship (\ref{eq:Character Orthogonality Subhypergroups}). The ``if'' follows from the calculation
\begin{align}
\PX{\mathsmaller{\mathcal{T}}}(a)P_X &=\sum_{ b, c\in  \mathcal{N}} \frac{\w{\mathcal{X}_{\overline{\varsigma}}\mathcal{N}^\perp}\w{b}}{\varw{\mathcal{N}}} n_{ a b}^{ c}\mathcal{X}_{\overline{\varsigma}}(b)\PX{\mathsmaller{\mathcal{T}}}( c)=\sum_{ c\in  \mathcal{N}} \frac{\w{\mathcal{X}_{\overline{\varsigma}}\mathcal{N}^\perp}\w{c}}{\varw{\mathcal{N}}} \left(\sum_{b\in  \mathcal{N}}n_{ {\overline{a}} c}^{ b}\mathcal{X}_{\overline{\varsigma}}(b)\right)\PX{\mathsmaller{\mathcal{T}}}( c)\\&=\sum_{c\in  \mathcal{N}} \frac{\w{\mathcal{X}_{\overline{\varsigma}}\mathcal{N}^\perp}\w{c}}{\varw{\mathcal{N}}} \mathcal{X}_{\overline{\varsigma}}(\overline{a})\mathcal{X}_{\overline{\varsigma}}(c)\PX{\mathsmaller{\mathcal{T}}}(c)=\mathcal{X}_{{\varsigma}}({a}) P_X,\label{eq:inproof:Projector Pauli X}
\end{align}
which implies with (\ref{eq:Character Orthogonality Subhypergroups}) that $P_X$ is a \emph{projector}, and consequently,  we get $\PX{\mathcal{T}}( a)\ket{\psi} = \PX{\mathcal{T}}(a)P_X\ket{\psi}=\mathcal{X}_\varsigma(a) P_X\ket{\psi}=\mathcal{X}_\varsigma(a)\ket{\psi}$ as desired.

As a result, we obtain that the stabilized states of $\{\mathcal{S}_X^{\lambda_x},\mathcal{S}_Z^{\lambda_z}\}$ are the quantum states in the vector space $\mathcal{V_S} :=\mathrm{span}\{P_X\ket{y}:y\in s\mathcal{N}\}$, where
\begin{align}\label{eq:inproof_P y}
P_X \ket{y} &= \sum_{\substack{ b\in \mathcal{N}\\
  x\in   s  \mathcal{N} } } \frac{\w{\Hchi_{\overline{\varsigma}}\mathcal{N}^\perp} \w{b}}{\varw{\mathcal{N}}} \sqrt{\frac{\w{y}}{\w{x}}} n_{ b, y}^{ x} \mathcal{X}_{\overline{\varsigma}}(b)\ket{ x} \propto \sum_{\substack{ x\in   s  \mathcal{N}} } \sqrt{\w{x}}  \left(\sum_{\substack{b}\in\mathcal{N} }n_{ x, {\overline{y}}}^{ b}\mathcal{X}_{\overline{\varsigma}}(b)\right)\ket{ x},
\end{align}
and that $\ket{\psi}$ is uniquely stabilized  iff this space is one dimensional. The proof for the RHS of (\ref{eq:Consistency CSS}) is the same: due to duality, we can apply a QFT (\ref{eq:QFTs are Clifford}) and reach this equality by  repeating the whole  proof from the beginning with exchanged roles for $\mathcal{T}$ and $\mathcal{T^*}$. This proves  case \textbf{(a)}.

Finally, we prove  \textbf{(b)}. First, in the simplest case, $s = e$, we can see that $\psi_y(x) = \sqrt{w_x} \Hchi_{\overline{\varsigma}}(x \overline{y})= \sqrt{w_x} \Hchi_{\overline{\varsigma}}(x)\Hchi_{\overline{\varsigma}}(\overline{y})=\psi_1(x)\Hchi_{\overline{\varsigma}}(\overline{y})$ by (\ref{eq:Character DEFINITION}), since $x,y\in\mathcal{N}$ and $\Hchi_{\overline{\varsigma}}$ is a character of $\mathcal{N}$. When $x, y \in s\mathcal{N}$, for $s \not = e$, we can, in general, have $n_{x,\overline{y}}^z \not= 0$ for $z \not\in \mathcal{N}$, so we cannot apply (\ref{eq:Character DEFINITION}). However, when $s$ is invertible (so $s \overline{s} = 1$), we can get the same result as in the simplest case by a simple change of variables.

For any $x,y\in s\mathcal{N}$, we define  $x':=\overline{s}x$ and $y':=\overline{s}y$. Since $x' \overline{y}' = \overline{s} x \overline{y} s = s \overline{s} x \overline{y} = xy$, we have $n_{x'\overline{y}'}^b = n_{x\overline{y}}^b$ and, taking $b=1$ and $y=x$, we have $w_{x'} = w_x$ from the definition of $w_x$. As these are the only appearances of $x$ and $y$ in $\psi_y(x)$, this shows that $\psi_y(x)=\psi_{y'}({x'})$. This combined with the previous easy case (for $x', y' \in \mathcal{N}$), shows that all $\psi_y$'s are  proportional to the non-zero function $\psi_1(x)$, which shows that the space is one-dimensional and contains $\ket{\psi}$. Finally, it is easily checked, in the case that $\mathcal{X}_\varsigma$ is invertible, that the normalization constant in (\ref{eq:Normal Form CSS state}, LHS) is  $(\w{\mathcal{X}_\varsigma \mathcal{N}^\perp}/\varpi_{s\mathcal{N}})^{-1/2}$; otherwise, it follows from (\ref{eq:Character Orthogonality Subhypergroups}). As in case \textbf{(a)}, duality lets us repeat the argument to get (\ref{eq:Normal Form CSS state}, RHS).
\end{proof}

As a final remark, we highlight that theorem \ref{thm:Normal Form CSS States} introduces many new states that we are not aware to be  preparable  by standard or character basis inputs and normalizer gates (the ingredients of the computational model in section \ref{sect: Circuit Model}), in general. However, we point out that the CSS states of type theorem \ref{thm:Normal Form CSS States}.(b) can always be prepared by measuring Pauli operators.
\begin{corollary}[\textbf{Coset state preparations}]\label{corollary:Coset State Preparations} Let $\mathcal{C}$ be a circuit  takes  the standard basis state  $\ket{\mathcal{X}_1}$ as input, performs $\Fourier{\mathcal{T}}^\dagger$ (an inverse QFT) or $\Fourier{\mathcal{T}^*}$, and then performs a syndrome measurement of the Pauli operators in a stabilizer hypergroup  $\mathcal{S}_Z^{\lambda_z}$ of form (\ref{eq:Stabilizer Hypergroup CSS State})\footnote{This is the canonical measurement defined by the common eigenprojectors that may be implemented, e.g., by measuring a poly-size set  generating set of $\mathcal{S}_Z^{\lambda_z}$, which exists if there is one for $\mathcal{N}^\perp$ (section \ref{sect:Stabilizer Formailsm}).}. Then $\mathcal{C}$ prepares a coset state. Specifically, it prepares $\ket{s\mathcal{N}}$  with probability $\varpi_{s\mathcal{N}}/\varpi_{\mathcal{T}}$. Furthermore, if $\ket{x_0}$ is given, $\Fourier{\mathcal{T}}$ or $\Fourier{\mathcal{T}^*}^\dagger$ is applied,  and $\mathcal{S}_X^{\lambda_x}$  (\ref{eq:Stabilizer Hypergroup CSS State}) is measured, then the outcome is a  coset state  $\ket{\mathcal{X}_\varsigma \mathcal{N}^\perp}$ with probability $\varpi_{\varsigma\mathcal{N}^\perp}/\varpi_{\mathcal{T}^*}$.
\end{corollary}
All states of form (\ref{eq:Normal Form CSS state}) can further be prepared from a coset state by applying Pauli gates.
\begin{proof} 
If we prove the first case, the second  holds  due to hypergroup duality. Lemma \ref{lemma:Quotient-Annihilator Isomorphisms}(ii) implies that measuring $\mathcal{S}_Z^{\lambda_z}$ is equivalent to performing a projective measurement with projectors $\{P_{s\mathcal{N}}=\ket{s\mathcal{N}}\bra{s\mathcal{N}}\}$. The claim follows by rewriting $\Fourier{\mathcal{T}}^\dagger\ket{\mathcal{X}_1}=\Fourier{\mathcal{T^*}}\ket{\mathcal{X}_1}=\sum_{a\in\mathcal{T}}\sqrt{\tfrac{\w{a}}{\varpi_{\mathcal{T}}}}\ket{a}=\sum_{s\mathcal{N}\in\mathcal{T/N}}\sqrt{\tfrac{\varpi_{s\mathcal{N}}}{\varpi_{\mathcal{T}}}}\ket{s\mathcal{N}}$.
\end{proof}

\section{Classical simulation of hypergroup normalizer circuits}\label{sect:Simulation}

The Gottesman-Knill theorem  shows that Clifford circuits can be efficiently simulated on a
classical computer, as there is an efficient classical algorithm for sampling from
the distribution of measurement outcomes of these circuits \cite{Gottesman_PhD_Thesis,Gottesman99_HeisenbergRepresentation_of_Q_Computers}. This result was
generalized by Van den Nest \cite{VDNest_12_QFTs} and Bermejo-Vega \cite{BermejoVega_12_GKTheorem} to normalizer circuits over all finite abelian groups, recovering the result of \cite{Gottesman_PhD_Thesis,Gottesman99_HeisenbergRepresentation_of_Q_Computers} for $\mathcal{T}=\Integers_2^n$. 

The original GK theorem was surprising because the intermediate state of a Clifford circuit can be maximally entangled \cite{nest06Entanglement_in_Graph_States}, demonstrating  that entanglement by itself is not sufficient resource for exponential quantum speedup in the standard gate model. The generalization of \cite{VDNest_12_QFTs,BermejoVega_12_GKTheorem}  expanded the set of
quantum circuits that can be classically simulated and added new
insights about HSP quantum circuits, showing that highly sophisticated gates such as Quantum Fourier Transforms  sometimes \emph{fail} to  exploit the  entanglement present in the HSP coset states to provide a speed-up.

Our next theorem extends a variant of the original Gottesman-Knill to normalizer circuits over arbitrary abelian hypergroups.
\begin{theorem}[\textbf{Simulation}]
\label{thm:simulation}
Let $\mathcal{C}$ be a normalizer circuit over a finite abelian hypergroup $\mathcal{T}$ containing \emph{global} QFTs, automorphism gates, and Pauli gates (but no quadratic phase gates) followed by a final measurement in the standard basis (cf.\ section \ref{sect: Circuit Model}). Then, given certain computability assumptions about $\mathcal{T}$ and its characters (section \ref{sect:Assumptions on Hypergroups}), there exists an efficient classical algorithm for sampling the measurement outcomes of $\mathcal{C}$.
\end{theorem}
The proof of the theorem is given at the end of the section. Our simulation result greatly expands the families of quantum circuits of \cite{Gottesman_PhD_Thesis,Gottesman99_HeisenbergRepresentation_of_Q_Computers,Gottesman98Fault_Tolerant_QC_HigherDimensions,VDNest_12_QFTs,BermejoVega_12_GKTheorem}
can be classically simulated and it also adds yet more evidence to support the
idea that, for the HSP, quantum efficiency may go hand-in-hand with classical
simulability.

We highlight that our theorem generalizes the so-called CSS-preserving Gottesman-Knill theorem \cite{Delfosee14_Wigner_function_Rebits} without intermeadiate measurements, where the only normalizer gates allowed are those that send CSS states to CSS states (definition \ref{def:CSS state}). Our result also extends the CSS-preserving case of the theorems in \cite{VDNest_12_QFTs,BermejoVega_12_GKTheorem}.  Yet, theorem \ref{thm:simulation}  does not fully extend the ones in \cite{Gottesman_PhD_Thesis,Gottesman99_HeisenbergRepresentation_of_Q_Computers,VDNest_12_QFTs,BermejoVega_12_GKTheorem}, which altogether cover simulations of partial QFTs, quadratic phase gates, and intermediate Pauli-operators  measurements interspersed along the circuit. Simulating these extended cases is much harder in our nonabelian setting because of the  \emph{non-unitarity and non-monomiality} of hypergroup Pauli operators (cf.\ discussion in section \ref{sect:Stabilizer Formailsm}), which do not let us apply any  existing techniques for manipulating stabilizer codes \cite{Gottesman_PhD_Thesis,Gottesman99_HeisenbergRepresentation_of_Q_Computers,Gottesman98Fault_Tolerant_QC_HigherDimensions,VDNest_12_QFTs,BermejoVega_12_GKTheorem,BermejoLinVdN13_Infinite_Normalizers,BermejoLinVdN13_BlackBox_Normalizers,nest_MMS,AaronsonGottesman04_Improved_Simul_stabilizer,dehaene_demoor_coefficients,dehaene_demoor_hostens,VdNest10_Classical_Simulation_GKT_SlightlyBeyond,deBeaudrap12_linearised_stabiliser_formalism}; instead, the simulation method we give is based on the new hypergroup stabilizer techniques of section \ref{sect:Stabilizer Formailsm}. 

We stress that CSS normalizer operations can be highly nontrivial, as the the quantum algorithms for abelian HSP are normalizer circuits with CSS structure. Theorem \ref{thm:simulation} could be used to simulate such circuits gate-by-gate if the information about the hidden subgroups was not manifestly hidden and groups were presented in a factorized form (cf.\ \cite{BermejoLinVdN13_BlackBox_Normalizers} for an extended discussion). This means, for instance, that the entanglement present in a CSS-preserving circuit can be quite substantial.

Lastly, we  conjecture that our simulation result can be extended to all normalizer gates despite the non-monomiality / non-unitarity issues  we discuss.
\begin{conjecture}[\textbf{Conjecture}]
There exist nontrivial families of abelian hypergroups for which the normalizer circuits of theorem \ref{thm:simulation} can still be efficiently classically simulated if they are supplemented with partial QFTs and quadratic phase gates acting at arbitrary circuit locations.
\end{conjecture}

\myparagraph{Discussion: extensions of theorem \ref{thm:simulation}}\label{sect:Extensions theorem simulation}

In the light of our conjecture, we mention  a few simpler extensions of theorem \ref{thm:simulation} that we are aware of.

First, note that an efficient classical simulation is still possible if the main circuit is followed by another one  $\mathcal{C}'$ that contains any monomial normalizer gate (including quadratic-phase gates), which is then followed by a measurement in the standard basis\footnote{One can simply absorb those gates in the measurements \cite{nest_weak_simulations}.}: such circuits can prepare more types of entangled stabilizer states like (\ref{eq:Entangled State}) and the quaternionic cluster state (section \ref{sect:Quaternionic circuits}).

Second, the theorem can be easily extended to allow arbitrary CSS state / stabilizer state inputs with one further minimal assumption, namely, that  their corresponding wavefunctions can be sampled both in the  hypergroup element basis $\mathsf{B}_{\mathcal{T}}$ and in  the character basis $\mathsf{B}_{\mathcal{T}^*}$, which lets us, in particular, simulate QFTs acting on the state  (see section \ref{sect:Assumptions on Hypergroups}, condition (ii) and section \ref{sect:simulation proof}). Furthermore, if this holds for the simple CSS states of theorem \ref{thm:Normal Form CSS States}, then theorem \ref{thm:simulation} can be extended to circuits that can, e.g.,  prepare coset states as in  corollary \ref{corollary:Coset State Preparations} and/or accept coset states as inputs.

\subsection{Computability assumptions on hypergroups}\label{sect:Assumptions on Hypergroups}

In theorem \ref{thm:simulation}, we must restrict ourselves to  hypergroups with sufficient structure to let us efficiently compute within them and their character  hypergroups. Note that  assumptions of this kind are typically made in the HSP literature: for instance, in order for the HRT quantum algorithm for the HNSP \cite{Hallgren00NormalSubgroups:HSP} to be efficient one needs to be given the ability to intersect characters kernels.  The assumptions needed for theorem \ref{thm:simulation} are  listed next, followed by examples of hypergroups  that meet them.

First, in theorem \ref{thm:simulation} we assume that the hypergroup $\mathcal{T}$ as well as its dual $\mathcal{T}^*$ are \emph{efficiently computable}: we  say that a hypergroup $\mathcal{T}$ is {efficiently computable}\footnote{Efficiently computable hypergroups generalize  the black-box groups \cite{BabaiSzmeredi_Complexity_MatrixGroup_Problems_I} explored in \cite{BermejoLinVdN13_BlackBox_Normalizers}.} if its elements can be uniquely represented with $n=O(\mathrm{polylog}{|\mathcal{T}|})$ bits and there are $O(\mathrm{poly}(n))$-time classical subroutines to perform the hypergroup multiplication, i.e.,\ given two elements $x_i$, $x_j\in \mathcal{T}$ and an index $k$, we can efficiently compute the coefficient $n_{ij}^k$ for any $i,j,k$.

In theorem \ref{thm:simulation},  we  further need to assume that the involved hypergroups are what we call  \emph{doubly efficiently computable}: a hypergroup $\mathcal{T}$ is doubly efficiently computable if both $\mathcal{T}$ and $\mathcal{T}^*$ are efficiently computable and, furthermore, that the structure of their associated character tables is sufficiently well-known that we are able to efficiently perform the following tasks classically:
\begin{itemize}
\item[(i)] \textbf{Computable characters.} For any $a\in \mathcal{T}$, any character function $\Hchi_\mu(a)$ can be efficiently computed classically\footnote{For simplicity, we will assume that this can be done with perfect precision in this manuscript. It is straightforward to adapt our main results any degree of available accuracy.}.
\item[(ii)] \textbf{Simulable input states.} Quantum Fourier transforms of allowed input states can be efficiently \emph{sampled} classically, or equivalently,  the distributions $\{p_a\}$ and $\{q_\mu\}$, with $p_a \defeq \tfrac{\w{a}\w{\mathcal{X}_\mu}}{\varpi_{\mathcal{T}}}|\mathcal{X}_\mu(a)|$ for fixed $\Hchi_\mu \in \mathcal{T}^*$ and $q_\mu \defeq \tfrac{\w{a}\w{\mathcal{X}_\mu}}{\varpi_{\mathcal{T}}}|\mathcal{X}_\mu(a)|$ for fixed $a \in \mathcal{T}$, can be efficiently sampled.
\item[(iii)] \textbf{Computable dual morphisms.} For any efficiently computable hypergroup automorphism $\alpha:\mathcal{T}\rightarrow \mathcal{T}$,  its inverse $\alpha^{-1}$ and its dual automorphism (definition \ref{eq:Dual Automorphism DEFINITION}) $\alpha^*:\mathcal{T}\rightarrow \mathcal{T^*}:\chi\rightarrow f^*_{\chi}$,  can both be efficiently determined and computed. Duals of computable hypergroup homomorphisms $f:\mathcal{T}\rightarrow \mathcal{T}'$ can also be computed\footnote{Applying def.\ \ref{eq:Dual Automorphism DEFINITION} to a  homomorphism $f:\mathcal{T}\rightarrow \mathcal{T}'$ one gets a dual morphism  $f^*:\mathcal{T}'\rightarrow \mathcal{T}$  \cite[1.6.(ii)]{McMullen79_Algebraic_Theory_Hypergroups}.}.
\end{itemize}

\paragraph{Examples and remarks}

Both computability notions presented are preserved by  taking direct products $\mathcal{T}_1\times \mathcal{T}_2$. Furthermore, the notion of doubly efficiently computable hypergroup is preserved under  taking duals $\mathcal{T} \leftrightarrow \mathcal{T}^*$.

Any hypergroup of the form $\mathcal{T}_1\times \cdots \times \mathcal{T}_m$ is doubly efficiently computable if  homomorphisms are restricted to be of a product form $f_1 \times \cdots \times f_m$, where $m$ is constant, or if they act nontrivially only in a constant number of sites. As a result,  \emph{normalizer circuits} over hypergroups of the from $\mathcal{T}_1\times \cdots \times \mathcal{T}_m$ with constant $|\mathcal{T}_i|$ will always turn out to be efficiently simulable if they contain  at most $k$-local entangling gates, for any constant $k$ (theorem \ref{thm:simulation}). The examples given in section \ref{sect:Quaternionic circuits} over the quaternions were of this form.

As another example, for any finite abelian group $G$, all problems in (i-ii-iii)  can be solved  in  $O(\polylog |G|)$ time given that $G$ is explicitly given in the form $G=\Integers_{N_1}\times \cdots\times \Integers_{N_m}$. Condition (ii) holds for any abelian group stabilizer state with known stabilizer group. These results are invariant of the bit-size of any $N_i$ \cite{VDNest_12_QFTs,BermejoVega_12_GKTheorem}. 

For arbitrary finite abelian hypergroups finding simple bounds like those in \cite{VDNest_12_QFTs,BermejoVega_12_GKTheorem} is likely to be an ``impossible'' problem, since the question cannot even be addressed without classifying all conjugacy class and character hypergroups of all finite groups, whereas classifying the latter is regarded as a \emph{wild} \cite{MO_Classification_Finite_Groups,MO_Classication_Problem_Wild} problem. It is easier to prove polynomial-time  bounds for particular hypergroup/group families  that fulfill the minimal computability requirements (i-ii-ii), like in the two examples given above.

Finally, we highlight that there exist  efficiently computable hypergroups that are not \emph{doubly} efficiently computable unless there exist efficient classical algorithms for believed-to-be-hard problems like computing discrete-logarithms exist (see appendix \ref{app:Discret Log}). In the abelian group case, the associated normalizer circuits can realize  Shor's \cite{Shor} algorithms  \cite{BermejoLinVdN13_BlackBox_Normalizers} and lead \emph{exponential quantum speed-ups}. In section \ref{sect:Quantum Algorithms}, we will develop new quantum algorithms based on normalizer circuits over such ``black-box'' hypergroups.

\subsection{Proof of theorem \ref{thm:simulation}}
\label{sect:simulation proof}

We finish this section proving theorem \ref{thm:simulation} by giving an explicit classical algorithm for sampling the outcome distribution after measuring the final state  of the computation $\mathcal{C}\ket{\psi_0}$,
being $\ket{\psi_0}$ the input state. Our algorithm is efficient given that the hypergroup $\mathcal{T}$ is doubly efficiently computable (section \ref{sect:Assumptions on Hypergroups}). The key technique that we exploit in our simulation is a normal formal for CSS normalizer circuits.
\begin{lemma}[\textbf{Normal form.}]\label{thm:Normal form} Let $\mathcal{C}$ be a normalizer circuit over a $\mathcal{T}$ as in theorem \ref{thm:simulation}. Then, $\mathcal{C}$ can be put in a layered normal form $\mathcal{C}=M  F$, where $F$ is either trivial or a QFT and $M$ is a monomial circuit\footnote{That is, a circuit whose transformation in matrix form has one entry per row and column.} of automorphism gates and Pauli gates. Furthermore, if  $\mathcal{T}$ is  doubly efficiently computable, then $M$ and $F$ can be computed classically efficiently.
\end{lemma}
\begin{proof}[Proof of lemma \ref{thm:Normal form}]
First, note that the Pauli gates in the circuit (which are   of the form $\PX{\mathcal{T}}(C)$,  $\PZ{\mathcal{T}}(\mathcal{X})$ in the conjugacy-class basis and of the form  $\PZ{\mathcal{T}^*}(C)$,  $\PX{\mathcal{T}^*}(\mathcal{X})$ in the character basis)  can be conjugated with all the other normalizer gates using the update rules in theorem \ref{thm:Evolution of Stabilizer States}(\ref{eq:Automorphism gates are Clifford},\ref{eq:Quadratic Phase gates are Clifford}). As a result, if $U$ is a Pauli gate at an intermediate circuit position $\mathcal{C}=\mathcal{C}_2 U \mathcal{C}_1$, it can be removed from its location by adding a new Pauli-correction term $U'=\mathcal{C}_2 U \mathcal{C}_2^\dagger$ at the start of the circuit. By doing this, we put $\mathcal{C}$ in an intermediate two-layered normal form $\mathcal{C}= P \mathcal{C'}$, where $P$ is a circuit of Pauli gates  and $\mathcal{C}'$ collects all  QFTs and all automorphism gates that were present in $\mathcal{C}$, in the same temporal order.

We finish the proof of the lemma by showing that  $\mathcal{C}'$ can be put in a normal form $A S$ where $S$ is either the identity gate or a QFT, and $A$ is a circuit of automorphism gates. Once we have that, we can combine $P$ and $A$ into a single layer $M$ and obtain $\mathcal{C}=M S$.

To this end, we use that group automorphism gates can be conjugated with  Fourier transforms in an elegant way, by replacing them with dual automorphism gates. Specifically, we have $\mathcal{F_T} U_\alpha = U_{\alpha^*}^{-1} \mathcal{F_T} = U_{\alpha^{-*}} \mathcal{F_T}$ and $\mathcal{F_{T^*}} U_{\alpha^{-*}}=U_\alpha \mathcal{F_{T^*}}$, where $\alpha^{-*}$ is the dual of $\alpha^{-1}$ (definition \ref{eq:Dual Automorphism DEFINITION}). This follows by simple calculation, using (\ref{eq:Quantum Fourier Transform over Hypergroup T}) and the fact that automorphisms cannot change the weights of elements. Furthermore, we in fact have
\begin{equation}\notag
U_\alpha\ket{\mathcal{X}_\mu}=\sum_{a\in\mathcal{T}}\sqrt{\tfrac{\w{a} \ws{\overline{\Hchi_\mu}}}{\varw{\mathcal{T}}}} \, \overline{\Hchi_{\mu}}(a)\ket{\alpha(a)}\stackrel{b:=\alpha(a)}{=}\sum_{b\in\mathcal{T}}\sqrt{\tfrac{\w{b} \ws{\overline{\Hchi_{\alpha^{-*}(\mu)}}}}{\varw{\mathcal{T}} }}  \, \overline{\Hchi_{\alpha^{-*}(\mu)}}(b)\ket{b}=U_{\alpha^{-*}}\ket{\Hchi_\mu}.
\end{equation}
Hence, it follows that $U_\alpha = U_{\alpha^{-*}}$ as gates, which means that can implement the latter since we can implement the former by assumption.

Applying these rules, we can move the QFTs before the automorphisms. Furthermore, since the effect of each QFT is just to change the designated of the circuit, any product of QFTs is equivalent to a single QFT gate, $F$, that changes the basis once (or not at all). By this proces, we get $\mathcal{C'}=AF$, where $A$ is a product of automorphism gates.

Finally, the assumption that both $\mathcal{T}$ and $\mathcal{T^*}$ are efficiently computable (and, hence, so is any hypergroup of form in (\ref{eq:Hypergroups})) means that we can efficiently compute the conjugations to determine the gates in $P$ using (\ref{eq:Automorphism gates are Clifford}, \ref{eq:Quadratic Phase gates are Clifford}). The assumption that $\mathcal{T}$ is doubly efficiently computable also includes the assumption that we can find $\alpha^{-*}$ efficiently for any $\alpha$, so we can compute the automorphisms in $A$ efficiently as well. We can determine $F$ at the same time with no additional assumptions, so we can see that $M$ and $F$ from the statement of the lemma can be determined efficiently under our assumption on $\mathcal{T}$.
\end{proof}

To prove the theorem, we first apply lemma \ref{thm:Normal form} to put $\mathcal{C}$ in normal form, $M F$. Next, we note that $M$ acts as $M\ket{a}=\gamma_a\ket{\pi(a)}$, where $\gamma_a$ has unit modulus and $\pi$ is some permutation on the elements of the basis for the image of $F$. Thus, measuring in the final basis, after applying $M$, is equivalent to measuring in the basis after $F$ and then applying $\pi$, which consists of just the automorphisms and Pauli $X$-gates from $M$. 

The assumption that $\mathcal{T}$ is doubly efficiently computable (section \ref{sect:Assumptions on Hypergroups}) tells us that all of this can be done efficiently. We can compute $M$ and $F$ under these assumptions by lemma \ref{thm:Normal form}. We can simulate a measurement after applying $F$ by assumption (ii). Then we can compute the result of each automorphism and Pauli $X$-gates from $M$ by assumption (i) and the assumption that $\mathcal{T}$ and $\mathcal{T^*}$ are efficiently computable, respectively. Q.E.D.

\paragraph{Remark}

As mentioned in the discussion after theorem \ref{thm:simulation},  one can straightforwardly extend this simulation method to any input $\ket{\psi_0}$ if measurements on $\ket{\psi_0}$ on the hypergroup element and character bases are easy to simulate (because we only use that information about the state).

\section{Quantum algorithms for HNSP and abelian HSHP}
\label{sect:Quantum Algorithms}

In this last section, we apply the hypergroup methods of previous sections to the development new  quantum algorithms for the HNSP and for the CC-HSHP.

We give three quantum algorithms for the HNSP of increasing generality (and complexity). These algorithms are interesting because they are fundamentally different from the one of Hallgren et al. \cite{Hallgren00NormalSubgroups:HSP}, as they exploit the hypergroup structure of the HNSP:  to solve the problem, we turn it into it into a CC-HSHP (using theorems \ref{thm:reduction1}-\ref{thm:reduction2}, section \ref{sect:HNSP}) and solve the resulting abelian HSHP instead. Our results show that, in the cases considered here, the HNSP is easy because the CC-HSHP is easy, which gives an explanation for why the HNSP is easy in terms of the presence of an \emph{abelian} algebraic structure.

Our quantum algorithms for the CC-HSHP are interesting in their own right (beyond their use in solving the HNSP) as no provably efficient algorithms were previously known.

As we will see shortly, all of the quantum algorithms we consider here fit within our normalizer circuit model. This means that they can be analyzed using the stabilizer formalism of section \ref{sect:Stabilizer Formailsm}. The results of that section, especially theorem \ref{thm:Normal Form CSS States}, will be critical to all of our analysis below.

We begin in section \ref{sect:Comparison HRT vs AKR} by looking at a previously proposed algorithm for the HSHP \cite{Amini2011fourier}. We show that, in one subclass of cases, not only can we reduce the HNSP to the CC-HSHP (as we saw in section~\ref{sect:HNSP}), but in fact, the algorithm of \cite{Hallgren00NormalSubgroups:HSP} for the HNSP is identical, under a simple vector space isomorphism, to the proposed algorithm of \cite{Amini2011fourier} applied to the CC-HSHP. This demonstrates an even deeper connection between the HNSP and the CC-HSHP than that demonstrated by the reductions of section~\ref{sect:HNSP}.

In section~\ref{sect:AKR Analysis}, we use our stabilizer formalism to analyze the proposed algorithm of \cite{Amini2011fourier}. We describe instances where it does and does not work correctly. This analysis leads us to a new algorithm, described in the same section, which works correctly for all groups.

In section~\ref{sect:New Algorithms}, we develop new quantum algorithms, taking advantage of the unique structure of abelian hypergroups. The resulting algorithms work for all nilpotent (hyper)groups (along with some non-nilpotent groups) and requires fewer assumptions than those of section~\ref{sect:AKR Analysis}. As with the algorithms of section~\ref{sect:AKR Analysis}, the stabilizer formalism remains key to our analysis.

Finally, in section~\ref{sect:Open Problems}, we mention a few further results and some open problems.

\subsection{A comparison of two simple algorithms for the HNSP}
\label{sect:Comparison HRT vs AKR}

To illustrate ideas we use later and introduce the quantum algorithm proposed in \cite{Amini2011fourier} for the CC-HSHP, we begin by discussing it and comparing it to the standard algorithm for the HNSP \cite{Hallgren00NormalSubgroups:HSP} in the case when the oracle $f: G \rightarrow \{0,1\}^*$ is a class function. This happens if and only if $G/N$ is abelian, where $N$ is the hidden normal subgroup. We show that, for this case, the two algorithms actually \emph{coincide}: the HNSP becomes an instance of the CC-HSHP and the same algorithm solves both problems.

As we saw in section~\ref{sect:HNSP}, such an $f$ is easily transformed into an oracle $\overline{f} : \Conj{G} \rightarrow \{0,1\}^*$ for the CC-HSHP since each coset $xN$ is in a conjugacy class by itself. This allows us to turn algorithms for the CC-HSHP into algorithms for the HNSP (and vice versa). We now compare two algorithms designed for this common problem via two different perspectives.

The quantum algorithm of Hallgren, Russell, and Ta-Shma for HNSP \cite{Hallgren00NormalSubgroups:HSP}, henceforth referred to as ``HRT'', operates as follow:
\begin{enumerate}
\item Initialize a workspace register in a quantum state $\ket{\chi_1}$. 

\item Apply an inverse QFT in order to obtain a superposition $\sum_{g\in
G}\ket{g}$.

\item Evaluate the oracle on an ancillary register to obtain $\sum_g
\ket{g,f(g)}$. Measure the second register to project the state of the first
onto a coset state $\ket{xN}$, for some $x$ drawn uniformly at random.

\item Apply a QFT to $\ket{xN}$ and measure the label $\mu$ of an
irreducible representation.

\item Repeat the experiment $T$ times and record the outcomes $\mu_1, \ldots,
\mu_T$.

\item Determine the subgroup $\bigcap_i \ker \chi_{\mu_i}$. With exponentially
high probability $1-O(\tfrac{1}{2^T})$, the subgroup $\bigcap_i \ker
\chi_{\mu_i}$ is the hidden subgroup $N$.
\end{enumerate}
The quantum part of this algorithm (steps 1--5) can be implemented efficiently
if we have an efficient implementation of the QFT. However, the complexity of
the classical post-processing (step 6) is unknown, in general.

The quantum algorithm of Amini, Kalantar, and Roozbehani applies to the hidden
subhypergroup problem (HSHP). When applied to the conjugacy class hypergroup,
$\Conj{G}$, we refer to this algorithm as ``AKR''. It takes as input an oracle
$\overline{f} : \overline{G} \rightarrow \{0,1\}^*$ and operates as follows:
\begin{enumerate}
\item Initialize a workspace register in a quantum state $\ket{\Hchi_1}$. 

\item Apply an inverse QFT in order to obtain a superposition $\abs{G}^{-1/2}
\sum_{C_x\in \Conj{G}} \sqrt{\w{C_x}} \ket{C_x}$.

\item Evaluate the oracle on an ancillary register to obtain $\abs{G}^{-1/2}
\sum_{C_x} \sqrt{\w{C_x}} \ket{C_x,\overline{f}(C_x)}$. Measure the second register to
project the state of the first onto a hypergroup coset state $\ket{C_x
\subhyp{N}{G}}$, for some $x$ drawn uniformly at random.\footnote{Note that we
are working in the Hilbert space with basis $\set{\ket{C_x}}{C_x \in \Conj{G}}$ and
the state $\ket{C_x \subhyp{N}{G}}$ is a superposition of those conjugacy
classes that make up $C_x \subhyp{N}{G}$.}

\item Apply a QFT to the state $\ket{C_x \subhyp{N}{G}}$ and measure the label
$\Hchi_\mu$ of a character.

\item Repeat the experiment $T$ times and record the outcomes $\Hchi_{\mu_1},
\ldots, \Hchi_{\mu_T}$.

\item Determine the subhypergroup of classes in the kernels of all the
$\Hchi_{\mu_i}$'s.
\end{enumerate}
Like HRT, steps 1--5 can be implemented efficiently while the complexity of
step 6 is unknown, in general.

As the reader can see, the two algorithms perform the same steps. First,
they apply an inverse Fourier transform to prepare a superposition over the
entire basis on which they operate. Next, they apply their respective oracles
to an additional register, measure, and throw away its value. Finally, they
apply a Fourier transform and measure in the new basis. For AKR, this is
measuring a character label, while for HRT, this is measuring the label of an
irrep. It is critical to note that HRT does not use the value of the matrix
index register, which is part of the output of the QFT for the group.

When $f$ is a class function (so that $G/N$ is abelian), we have $C_x\subhyp{N}{G} = xN$
since each coset $xN$ is in its own conjugacy class. Hence, we can see that the
two algorithms are in matching states after steps 1--3. In particular, the
state of HRT after step 3 is conjugation invariant, and since the QFT preserves
conjugation invariance, so is the state of HRT after step 4. In fact, it is
easy to check that the QFT of $G$ applied to the conjugation invariant subspace
is exactly the QFT of $\Conj{G}$. More precisely, we have the following
result, which is also easy to check.
\begin{proposition}[\textbf{HRT = AKR}]
If $f : G \rightarrow \{0,1\}^*$ is a class function, then HRT operates
entirely within the conjugation invariant subspace $\Comp_{\overline{G}}$. Furthermore,
within this subspace, HRT is identical to AKR.
\end{proposition}
The first part of the following lemma simplify our analysis of the algorithm. (And the second part will be useful to us later on.)
\begin{lemma}[\cite{Roth75_Character_Conjugacy_Hypergroups}]\label{lemma:Roth}
For any normal subgroups $N,K$ such that $N\trianglelefteq K \trianglelefteq G$, the  hypergroup $\overline{G}/\subhyp{N}{G}$ is isomorphic to $\overline{G/N}$ (the class hypergroup of $G/N$) and $\subhyp{K}{G}/\subhyp{N}{G}$ is a subhypergroup of $\overline{G}/\subhyp{N}{G}$ isomorphic to  $\subhyp{K/N}{G/N}$ (which is a subhypergroup of $\overline{G/N}$).
\end{lemma}
Since $G/N$ is abelian, the quotient hypergroup $\Conj{G} / \Conj{N}_G \cong \Conj{G/N}$ is actually a group. Hence, it follows immediately from the standard results on Fourier sampling of abelian groups \cite{kitaev_phase_estimation} that both algorithms are correct in this case as they are actually just performing Fourier sampling of an abelian group.

\begin{theorem}[\textbf{CC-HSHP is easy, I}]
\label{thm:easy1}
Let $G$ be a group. Suppose that we are given a function $\overline{f} :
\Conj{G} \rightarrow \{0,1\}^*$ hiding the subhypergroup corresponding to a
normal subgroup $N \lhd G$ such that $G/N$ is abelian, and suppose that we can
efficiently compute the QFT for $\Conj{G}$. Then there is an efficient
quantum algorithm for the CC-HSHP.
\end{theorem}
 
As a corollary of theorem \ref{thm:easy1}, we can see that there is an efficient hypergroup-based
algorithm for solving HNSP in the case when the oracle $f : G \rightarrow
\{0,1\}^*$ is a class function.

\begin{corollary}[\textbf{HNSP is easy, I}]
Let $G$ be a group. Suppose that we are given a hiding function $f : G
\rightarrow \{0,1\}^*$ that is a also class function.  If we can efficiently
compute the QFT for $G$ and compute efficiently with conjugacy classes of $G$,
then we can efficiently solve this HNSP.
\end{corollary}

\begin{proof}
This follows since we can efficiently reduce the HNSP to a CC-HSHP by theorem \ref{thm:reduction1}, we can use the QFT of $G$ to implement the QFT of $\Conj{G}$ as described in appendix~\ref{app:CC implementation details}, and we can solve this CC-HSHP efficiently by theorem \ref{thm:easy1}.
\end{proof}

\subsection{Analysis of the algorithm of Amini et al.}
\label{sect:AKR Analysis}

In this subsection, we analyze the AKR algorithm in more detail. Our analysis will benefit from our hypergroup stabilizer formalism tools (section \ref{sect:Stabilizer Formailsm}), namely, our normal forms in theorem \ref{thm:Normal Form CSS States}, which will let us characterize the outcome probability distribution of the quantum algorithm.\footnote{This is a new result. No formula for this probability was given in \cite{Amini_hiddensub-hypergroup,Amini2011fourier}.} This is possible due to the following connection with our normalizer circuit framework.
\begin{theorem}[\textbf{AKR is normalizer}]\label{thm:AKR is normalizer}
For any finite group $G$, the AKR quantum algorithm for the CC-HSHP over $\overline{G}$ is a normalizer circuit  with intermediate Pauli measurements. Furthermore, all of its intermediate quantum states are  CSS stabilizer states of form (\ref{eq:Normal Form CSS state}).\footnote{Although we do not discuss the full AKR algorithm, this theorem also holds for all abelian hypergroups.}
\end{theorem}
\begin{proof}
Steps 1-3 implement a coset state preparation scheme as in corollary \ref{corollary:Coset State Preparations}, using  the oracle as a black box to perform a syndrome measurement of type $\mathcal{S}_Z^\lambda$ (cf.\  the proofs of corollary \ref{corollary:Coset State Preparations}, theorem \ref{thm:Normal Form CSS States} for details). Step 4  takes the QFT of a coset state of form (\ref{eq:Normal Form CSS state}).
\end{proof}
Using this connection, we can apply the tools developed in the previous
sections to compute the probability of measuring a given character label
$\Hchi_\mu$ at step 5.

The mixed state of AKR after the oracle is called and its value discarded is
given by
\[ \rho = \sum_{C_x\subhyp{N}{G} \in
      \overline{G}/\subhyp{N}{G}} \frac{\w{C_x\subhyp{N}{G}}}{\varw{\overline{G}/\subhyp{N}{G}}}
          \ket{C_x\subhyp{N}{G}}\bra{C_x\subhyp{N}{G}} \]
since the probability of measuring each coset $C_x \subhyp{N}{G}$ is proportional to its weight.\footnote{The initial superposition has probability of measuring $C_x$ proportional to $w_{C_x}$, so the probability of measuring the coset $C_x \subhyp{N}{G}$, which contains all the elements with the same value from the oracle as $C_x$, is proportional sum of all of their weights, $\varpi_{C_x \subhyp{N}{G}}$. These sums are proportional to the weights $w_{C_x \subhyp{N}{G}}$, so we get the form in the equation above after normalizing the probabilities to sum to 1.} 
Each coset state $\ket{C_x\subhyp{N}{G}}$ is a stabilized by $X_{\overline{G}}(C_y)$
for any $C_y \in \subhyp{N}{G}$ and any $Z_{\widehat{G}}(\Hchi_\mu)$ for any $\Hchi_\mu
\in \subhyp{N}{G}^\perp$, which means that it is a CSS stabilizer state of the form shown
in equation (\ref{eq:Stabilizer Hypergroup CSS State}) with $s=C_x$ and
$\Hchi_\varsigma = \Hchi_1$, the trivial character. Thus, we can read off the Fourier
transform of this state directly from theorem \ref{thm:Normal Form CSS States}.

Our application of  theorem \ref{thm:Normal Form CSS States} is greatly
simplified by the fact that $\Hchi_{\overline{\varsigma}}$ is the trivial character $\Hchi_{1}$. This means
 that
$\w{\Hchi_{\overline{\varsigma}}\subhyp{N}{G}^\perp} = 1$ since the weight of a character is the
dimension of the underlying representation $\mu$.\footnote{Note that
$\Hchi_{\overline{\varsigma}}\subhyp{N}{G}^\perp$ is an element of $\widehat{G}/\subhyp{N}{G}^\perp \cong
\subhyp{N}{G}^*$ (section \ref{sect:Glossary}), so this is the weight of $\Hchi_{\overline{\varsigma}}$ viewed as a
character of $\subhyp{N}{G}$, where it remains trivial.} This means that the state
$\ket{\psi}$ for $s = C_x$ and $\mathcal{N} = \subhyp{N}{G}$ is proportional to
$\sum_{C_y \in C_x\subhyp{N}{G}} \sqrt{\w{C_y}} \ket{C_y}$, which is proportional to
$\ket{C_x\subhyp{N}{G}}$, and thus, we have $\ket{\psi} = \ket{C_x\subhyp{N}{G}}$
since both states are normalized. Thus, theorem \ref{thm:Normal Form CSS
States} tells us
\[ \Fourier{\overline{G}} \ket{C_x\subhyp{N}{G}} = \sum_{\Hchi_\mu \in \subhyp{N}{G}^\perp}
   \sqrt{\frac{\w{\mu} \w{C_x\subhyp{N}{G}}}{\varpi_{\subhyp{N}{G}^\perp}}} \Hchi_\mu(C_x)
   \ket{\Hchi_\mu},  \]
and hence, the Fourier transform of the mixed state of AKR is
\begin{eqnarray*}
\Fourier{\overline{G}} \rho \Fourier{\overline{G}}^\dag
&=& \sum_{C_x\subhyp{N}{G}\in\overline{G}/\subhyp{N}{G}} \frac{ \w{C_x\subhyp{N}{G}} }{\varw{\overline{G}/\subhyp{N}{G}}}
 \left( \sqrt{\w{C_x\subhyp{N}{G}}} \sum_{\mu \in \subhyp{N}{G}^\perp}
    \sqrt{\frac{\w{\mu}}{\varpi_{\subhyp{N}{G}^\perp}}} \Hchi_\mu(C_x) \ket{\Hchi_\mu} \right)
  \times \\
&& \qquad \qquad \qquad \ 
   \left( \sqrt{\w{C_x\subhyp{N}{G}}} \sum_{\nu \in \subhyp{N}{G}^\perp}
     \sqrt{\frac{\w{\nu}}{\varpi_{\subhyp{N}{G}^\perp}}} \Hchi_\nu(C_x) \bra{\Hchi_\mu}
     \right) \\
&=& \sum_{\mu, \nu \in \subhyp{N}{G}^\perp} \sqrt{\w{\mu} \w{\nu}}
    \left( \sum_{C_x\subhyp{N}{G} \in \overline{G}/\subhyp{N}{G}}
      \frac{\w{C_x\subhyp{N}{G}}^2}{\varw{\overline{G}/\subhyp{N}{G}}^2}
      \Hchi_\mu(C_x) \overline{\Hchi_{\nu}}(C_x) \right)
    \ket{\Hchi_\mu} \bra{\Hchi_\nu}, \\
\end{eqnarray*}
where, in the last step, we have used the fact that $\subhyp{N}{G}^\perp \cong 
\left(\overline{G}/\subhyp{N}{G}\right)^*$. Finally, using lemma \ref{lemma:Roth}, we conclude that the probability of measuring the outcome $\Hchi_\mu$ is
\begin{equation}
\label{eq:AKR-prob1}
\Pr(\Hchi_\mu) =
\w{\mu} \sum_{C_x\subhyp{N}{G} \in \overline{G}/\subhyp{N}{G}} \frac{\w{C_x\subhyp{N}{G}}^2}{\varw{\overline{G}/\subhyp{N}{G}}^2}
      \Hchi_\mu(C_x) \overline{\Hchi_{\mu}}(C_x)=
      \w{\mu} \sum_{C_{xN} \in \overline{G/N}} \frac{\w{C_{xN}}^2}{\varw{\overline{G/N}}^2}
            \Hchi_\mu(C_x) \overline{\Hchi_{\mu}}(C_x).
\end{equation}
This formula is the key to our analysis of AKR in this section and the next.

\subsubsection{Non-convergence of AKR for simple instances}

We begin examining these probabilities by looking at an example.

\begin{example}[\textbf{AKR Counterexample}]
\label{ex:non-convergence}
The Heisenberg group over $\Integer_p$ (with $p$ prime) is the set
$\Integer_p^3$ with multiplication defined by $(x,y,z)\cdot(x',y',z') =
(x+x',y+y',z+z'+xy')$. For a nice review of it's representation theory, see
the article of Bacon \cite{Bacon_Clebsch-Gordon}.

The center of this group is the normal subgroup $Z(G) = \set{(0,0,z)}{z \in
\Integer_p}$. Hence, any element of $Z(G)$ is in a conjugacy class of its own.
For any $(x,y,z) \not\in Z(G)$ (i.e., with $(x,y) \not= (0,0)$), it is easy to
check that its conjugacy class is the coset $(x,y,z) Z(G)$. Hence, the weight
of the former classes are 1 and those of the latter are $\abs{Z(G)} = p$.

Let us consider the case when the hidden subgroup is trivial $N = \{e\}$. For
any $(a,b) \not= (0,0)$ there is a 1-dimensional irrep of the Heisenberg group,
which we denote $\Hchi_{a,b}$ given by $\Hchi_{a,b}(x,y,z) = \omega_p^{ax+by}$,
where $\omega_p$ is a $p$-th root of unity. We can apply
equation~(\ref{eq:AKR-prob1}) to compute the probability of measuring this
irrep in AKR. To do this, we first note that this probability would be the
inner product of $\Hchi_\mu$ with itself except that weights have been squared.
Since there are only two sizes of conjugacy classes, it is not difficult to
rewrite this expression in terms of that inner product as follows.
\begin{eqnarray*}
\Pr(\Hchi_{a,b})
 &=& \sum_z \frac{1}{p^6} \abs{\Hchi_{a,b}(0, 0, z)}^2 +
     \sum_{(x,y) \not= (0,0)} \frac{p^2}{p^6} \abs{\Hchi_{a,b}(x,y,0)}^2 \\
 &=& \frac{1}{p^6} \sum_z 1 +
     \frac{p^2}{p^6} \sum_{(x,y) \not= (0,0)} 1 \\
 &=& \frac{1}{p^6} p + \frac{p^2}{p^6} (p^2 - 1) \\
 &=& \frac{p^4 - p^2 + p}{p^6} \\
\end{eqnarray*}
where we have used the fact that $\abs{\Hchi_{a,b}(x,y,z)} = 1$ for all $x,y,z
\in \Integer_p$.

Finally, the probability of measuring any of the 1-dimensional irreps is 
\[ \sum_{(a,b)\not=(0,0)} \Pr(\Hchi_{a,b})
    = (p^2-1) \frac{p^4 - p^2 + p}{p^6}
    = \frac{p^6-2p^4+p^3+p^2-p}{p^6}
    = 1 - O\left(\frac{1}{p^2}\right). \]
This means that if $p$ is exponentially large, we are unlikely to ever see
an irrep other than the $\Hchi_{a,b}$'s, and since $Z(G)$ is in the kernel of
all such irreps, the intersection of the kernels of polynomially many irreps
will include $Z(G)$ with high probability.

Since the AKR algorithm returns the intersection of the kernels of the sampled
irreps as its guess of the hidden subgroup, this demonstrates that the AKR
algorithm will fail to find the hidden subgroup with high probability for the
Heisenberg group with $N=\{e\}$. Indeed, this shows that AKR will fail to
distinguish between $N = Z(G)$ and $N = \{e\}$, despite the fact that the
former is exponentially larger than the latter.
\end{example}
The probability distribution over irreps established by AKR in this example
(and many others) favors the small dimensional irreps, whereas the distribution
established by HRT favors the large dimensional irreps. In this example, that
fact prevents AKR from ever seeing the irreps needed to uniquely determine $N$.
(Paradoxically, when finding non-normal hidden subgroups, it is often the small
irreps that are most useful and HRT that struggles to find them.\footnote{While
the AKR distribution would be better, AKR does not apply to finding non-normal
hidden subgroups since the hypergroup structure is no longer present.})

\subsubsection{An application of AKR: a 2nd hypergroup algorithm for the HNSP}

While AKR may fail to uniquely determine $N$ from its samples, the following
lemma tells us that it will, with high probability, learn something about $N$.
\begin{lemma}
If $N \not= G$, then the intersection $K$ of the kernels of the irreps sampled by AKR is a strict subgroup of $G$ --- i.e., we will have $N \le K \lneq G$ --- with high probability.
\end{lemma}
\begin{proof}
The intersection of the kernels will be smaller than $G$ provided that at least
one of the samples has a kernel smaller than $G$. In other words, the
intersection will be a strict subgroup provided that at least one of the irreps
sampled is not the trivial irrep.

We can use eq.~(\ref{eq:AKR-prob1}) to calculate the probability of sampling
the trivial irrep. Since $\Hchi_1(C_x) = 1$ for any $C_x$, the probability for the
trivial irrep is just $\sum_{C_{xN} \in \overline{G/N}} \w{C_{xN}}^2 /
\varw{\overline{G/N}}^2$. If we let $G'$ denote the group $G/N$, then we can also
write this as $\sum_{C_{x'} \in \overline{G'}} \w{C_{x'}}^2 / \varw{\overline{G'}}^2$. Now,
our job is to determine how close this can be to 1.

Let $C_1, \dots, C_m$ be the conjugacy classes of $G'$, and define $s_i =
\w{C_i} / \varw{G'} = \abs{C_i} / \abs{G'}$. The $s_i$'s satisfy $\sum_i s_i = 1$
(since every element of $G'$ is in some conjugacy class). Since the size of a
conjugacy class divides the size of the group, we also know that $s_i \le 1/2$
for every $i$. Forgetting everything but these constraints, we can upper bound
the probability of sampling the trivial irrep by the solution of the
optimization problem
\begin{equation*}
\begin{aligned}
& \text{maximize}   & & \sum_i s_i^2         & \\
& \text{subject to} & & \sum_i s_i = 1       & \\
&                   & & s_i \le \tfrac{1}{2} & \text{for}\ i=1 \dots m
\end{aligned}
\end{equation*}
The latter problem can be solved by ordinary methods of calculus. In
particular, it is easy to check that the objective is increased if $s_i$ and
$s_j$ are replaced by $s_i+\epsilon$ and $s_j-\epsilon$ provided that $s_i >
s_j$. (I.e., the derivative in this direction is positive.) Hence, the
objective will be maximized when the two largest $s_i$'s have value $1/2$ and
all other $s_i$'s are 0. At that point, the objective is $(1/2)^2 + (1/2)^2 =
1/2$.

This tells us that the probability of sampling the $\Hchi_1$ is at most $1/2$ on
each trial. Hence, the probability that all the samples are the trivial irrep
is exponentially small.
\end{proof}

This lemma tells us that AKR will find, with high probability, a strict
subgroup $K \lneq G$ such that $N \le K$. This means that we can reduce the
problem of finding $N$ hidden in $G$ to the problem of finding $N$ hidden in
$K$, which is a strictly smaller group. Provided that we understand the
representation theory of $K$ as well and, in particular, have a QFT for it,
then we can recursively solve this problem. In more detail, we have:
\begin{theorem}[\textbf{CC-HSHP is Easy, II}]
\label{thm:easy2}
Let $G$ be a group and $N$ a normal subgroup. Suppose that, for each normal subgroup $K$ satisfying $N \le K \le G$, we have a function $\overline{f}_K : \overline{K} \rightarrow \overline{H_K}$ that hides $\subhyp{N}{K}$.\footnote{Alternatively, we may assume that, for any $K \lhd G$, we have a  function $\overline{f}_K$ that hides $\subhyp{(N \cap K)}{K}$.} If we can efficiently compute the QFT for each such $\Conj{K}$, then there is an efficient quantum algorithm for the CC-HSHP.
\end{theorem}
As a corollary, we can see that there is an efficient hypergroup-based
algorithm for solving HNSP in the case when the oracle $f : G \rightarrow H$ is
a group homomorphism.
\begin{corollary}[\textbf{HNSP is easy, II}]
Let $G$ be a group. Suppose that we are given a hiding function $f : G
\rightarrow H$ for $N \lhd G$ that is also a group homomorphism. If we can
efficiently compute the QFT for any normal subgroup $K$ satisfying $N \le K \le
G$ and compute efficiently with conjugacy classes of each such $K$ and $H_K$,
where $H_K$ is the image of $f|_K$, then we can efficiently solve this HNSP.
\end{corollary}
\begin{proof}
The restriction of $f$ to $K$, $f|_K : K \rightarrow H_K$, is itself a group
homomorphism, so by the assumptions of the corollary and
theorem \ref{thm:reduction2}, we can efficiently compute from it a CC-HSHP
oracle $\overline{f}_K : \overline{K} \rightarrow \overline{H_K}$. Hence, the
result follows from theorem \ref{thm:easy2} and the results of appendix~\ref{app:CC implementation details} on implementing the QFT of a hypergroup with the QFT of the group.
\end{proof}
In comparison to HRT, we have replaced the assumption about intersecting
kernels of irreps with assumptions about computing QFTs over normal subgroups
of $G$. Since the complexity of computing intersections of subgroups is unknown
(at least in the black-box setting), the latter assumptions may be more
reasonable in many cases. This same approach of recursion on subgroups can also
be used to give a variant of HRT with similar assumptions to those above.

In this section and the previous, we have shown that, in cases when we can convert an oracle for the HNSP into an oracle for the CC-HSHP --- either because the two types of oracles are actually the same or because we have sufficient information about the output of the former oracle to compute with their conjugacy classes --- we can solve the HNSP by reducing it to CC-HNSP. That the latter problem can be solved efficiently depends crucially on the fact that the hypergroup is abelian. Thus, the results of the last two sections give us an explanation for why the HNSP is easy in terms of the presence of an \emph{abelian} algebraic structure, the abelian hypergroup of conjugacy classes.

\subsection{Efficient quantum algorithm for the nilpotent group HNSP and CC-HSHP}
\label{sect:New Algorithms}

In this section, we describe our most unique hypergroup-based quantum algorithms that solves the HNSP and the CC-HSHP. Unlike the algorithm of the previous sections, this algorithm does not require any extra assumptions about subgroups such as hiding functions for them or the ability to compute efficiently with their conjugacy classes. More importantly, this algorithm is fundamentally different from the known algorithm (HRT), showing that this import problem can be solved efficiently using a different approach.

Our approach takes advantage of unique properties of hypergroups. In
particular, as we are looking to reduce our problem on $G$ to a subproblem, we
note that, for any normal subgroup $K \le G$, there are actually two smaller
hypergroups associated to it. The first is the hypergroup $\Conj{K}$ that we
get by looking at $K$ as a group separate from $G$. The second is the
subhypergroup $\Conj{K}_G$, where $\Conj{K}_G$ contains the conjugacy classes
$C_x \in \Conj{G}$ such that $x \in K$.  Above, when we recursively solved a
problem in $K$, we were using the hypergroup $\Conj{K}$. However, as we will
see in this section, it is also possible to solve the subproblem on the
subhypergroup $\Conj{K}_G \le \Conj{G}$.

The subhypergroup $\Conj{K}_G$ has two advantages over $\Conj{K}$. The first is
that a CC-HSHP oracle for $\Conj{G}$ is also a CC-HSHP oracle for $\Conj{K}_G$:
since the conjugacy classes of the two hypergroups are the same, the condition
that the oracle is constant on conjugacy classes of $\Conj{G}$ means that the
same is true for $\Conj{K}_G$. The second advantage of this subhypergroup is
given in the following lemma.
\begin{lemma}[\textbf{Subhypergroup QFT}]
\label{lma:subhypergroup-qft}
Let $K \lhd G$. If we can efficiently compute $\Fourier{\overline{G}}$, the QFT   over $\overline{G}$, then we can also efficiently compute  $\Fourier{\Conj{K}_G}$, i.e., the QFT over $\Conj{K}_G$.
\end{lemma}
\begin{proof}
In fact, the QFT for $\Conj{K}_G$ is implemented by the same QFT as for
$\Conj{G}$ provided that we choose the \emph{appropriate basis} for the dual of $\Conj{K}_G$.

To describe this basis, we first note that every character on $\Conj{G}$ is a character on $\Conj{K}_G$ simply by restricting its domain to $\Conj{K}_G$. This map of $\Conj{G} \rightarrow \Conj{K}_G$ is in fact surjective with kernel $\subhyp{K}{G}^\perp$ \cite{BloomHeyer95_Harmonic_analysis}. This means that the characters of $\Conj{K}_G$ are in 1-to-1 correspondence with the cosets of $\subhyp{K}{G}^\perp$ in $\Repr{G}$ (i.e., $\Conj{K}_G^* \simeq \Repr{G} / \subhyp{K}{G}^\perp$), so our basis for characters of $\Conj{K}_G$ should be a basis of cosets of $\subhyp{K}{G}^\perp$ in $\Repr{G}$. As described in the example of section~\ref{sect:Examples Normal Form CSS States}, the cosets of $\subhyp{K}{G}^\perp$ are of the form $\ket{\Hchi_\nu \subhyp{K}{G}^\perp} = \sum_{\Hchi_\mu \in \Hchi_\nu \subhyp{K}{G}^\perp} \sqrt{w_\mu / \varpi_{\Hchi_\nu \subhyp{K}{G}^\perp}} \ket{\Hchi_\mu}$. Note that $\varw{\Hchi_{\mu}\subhyp{K}{G}^\perp}=w_{\Hchi_{\mu}\mathcal{K}^\perp}\varw{\subhyp{K}{G}^\perp}=\varw{\Hchi_{\mu}\subhyp{K}{G}^\perp}\varw{\left(\overline{G}/K\right)^*}=\varw{\Hchi_{\mu}\subhyp{K}{G}^\perp}\varw{\overline{G}/K}=\varw{\Hchi_{\mu}\subhyp{K}{G}^\perp}\varw{\overline{G}}/\varw{\subhyp{K}{G}}$, using the definitions in section~\ref{sect:Hypergroups}. This means that we can write the state instead as \[ \ket{\Hchi_\nu \subhyp{K}{G}^\perp} = \sum_{\Hchi_\mu \in \Hchi_\nu \subhyp{K}{G}^\perp} \sqrt{\frac{w_\mu\, \varw{\subhyp{K}{G}}}{\w{\Hchi_{\mu}\subhyp{K}{G}^\perp} \varw{G}}} \ket{\Hchi_\mu}, \] which is the definition we will use below.

With that basis chosen, we can now calculate the Fourier transform of a conjugacy class state $\ket{C_x}$ with $C_x \in \Conj{K}_G$. The key fact we will use below is that $\Hchi_\mu(C_x) = \Hchi_\nu(C_x)$ whenever $\Hchi_\mu \in \Hchi_\nu \subhyp{K}{G}^\perp$ since they differ only by multiplication with a character that is identity on $C_x$ (lemma \ref{lemma:Quotient-Annihilator Isomorphisms}.(i)). Hence, we can define $\Hchi_{\nu}\subhyp{K}{G}^\perp(C_x)$ to be this common value.
\begin{eqnarray*}
\Fourier{\overline{G}} \ket{C_x}
 &=& \sqrt{\frac{\w{C_x}}{\varw{G}}} \sum_{\Hchi_\mu \in \Repr{G}} \sqrt{\w{\mu}} \Hchi_\mu(C_x) \ket{\Hchi_\mu} \\
 &=& \sqrt{\frac{\w{C_x}}{\varw{G}}} \sum_{\Hchi_\nu \subhyp{K}{G}^\perp \in \Repr{G}/\subhyp{K}{G}^\perp}  \Hchi_\nu \subhyp{K}{G}^\perp(C_x) \sum_{\mu \in \Hchi_\nu \subhyp{K}{G}^\perp} \sqrt{\w{\mu}} \ket{\Hchi_\mu} \\
 &=& \sqrt{\frac{\w{C_x}}{\varw{G}}} \sum_{\Hchi_\nu \subhyp{K}{G}^\perp \in \Repr{G}/\subhyp{K}{G}^\perp } \Hchi_\nu \subhyp{K}{G}^\perp(C_x) \sqrt{\frac{\w{\Hchi_{\nu}\subhyp{K}{G}^\perp} \varw{G}}{\varw{\subhyp{K}{G}}}} \ket{\Hchi_{\nu}\subhyp{K}{G}^\perp} \\
 &=& \sqrt{\frac{\w{C_x}}{\varw{\subhyp{K}{G}}}} \sum_{\Hchi_\nu \subhyp{K}{G}^\perp \in \Repr{G}/\subhyp{K}{G}^\perp} \sqrt{\w{\Hchi_{\nu}\subhyp{K}{G}^\perp}} \Hchi_\nu \subhyp{K}{G}^\perp(C_x) \ket{\Hchi_{\nu}\subhyp{K}{G}^\perp}
\end{eqnarray*}
Because $\subhyp{K}{G}^*\cong \Repr{G}/\subhyp{K}{G}^\perp$ and  $\varw{\subhyp{K}{G}}=\varw{\subhyp{K}{G}^*}$, this last line is, by definition, the QFT for $\Conj{K}_G$, so we have seen
that the QFT for $\Conj{G}$ implements this QFT as well.
\end{proof}
The lemma shows that the assumption that we have an efficient QFT for the whole hypergroup $\Conj{G}$ is sufficient to allow us to recurse on a subproblem on $\Conj{K}_G$ without having to assume the existence of another efficient QFT specifically for the subproblem, as occurred in our last algorithm.

Our next task is to analyze the AKR algorithm applied to this hypergroup and,
in particular, determine the probability distribution that we will see on
character cosets when we measure. Recall that the algorithm starts by preparing
 a weighted superposition over $\subhyp{K}{G}$ (which is a uniform distribution over $K$) and  invoking the oracle.
The result is
\[ \rho = \sum_{C_x\subhyp{N}{G} \in (\subhyp{N}{G}/\subhyp{K}{G})}\frac{ \w{C_x\Conj{N}_G} }{\varw{(\subhyp{K}{G}/\subhyp{K}{G})}}
\ket{C_x\Conj{N}_G} \bra{C_x\Conj{N}_G}. \]

As with AKR, we can find the Fourier transform of this state directly from
theorem~\ref{thm:Normal Form CSS States}. Following the same argument as
before, we see that the $\ket{\psi}$ from part (b) with $\Hchi_\varsigma = \Hchi_1$,
the trivial character, and $s=C_x$ is precisely the state
$\ket{C_x\overline{N}_G}$. Note that, since we are no longer working in $\overline{G}$ but the
subhypergroup $\overline{K}_G$, we do a QFT over $\overline{K}_G$ (lemma \ref{lma:subhypergroup-qft}) and the subhypergroup $\mathcal{N}^\perp=\subhyp{N}{G}^\perp$ in theorem~\ref{thm:Normal Form CSS States} belongs to
$\overline{K}_G^*$. 
\[ \Fourier{\overline{K}_G} \ket{C_x\overline{N}_G} =
    \sum_{\Hchi_\mu \in \subhyp{N}{G}^\perp \le \subhyp{K}{G}^*}
   \sqrt{\frac{\w{\Hchi_\mu} \w{C_x\overline{N}_G} }{\varw{\subhyp{N}{G}^\perp}}}
    \Hchi_\mu(C_x) \ket{\Hchi_\mu}, \]
By a similar calculation to before, we have
\begin{eqnarray*}
\Fourier{\overline{K}_G} \rho \Fourier{\overline{K}_G}^\dag &=&
   \sum_{\Hchi_\mu, \Hchi_\nu \in \subhyp{N}{G}^\perp}
   \sqrt{\w{\Hchi_\mu} \w{\Hchi_\nu}} 
 \sum_{C_x\overline{N}_G \in (\subhyp{K}{G}/\subhyp{N}{G})}
   \frac{\w{C_x\overline{N}_G}^2}{\varw{\subhyp{N}{G}^\perp}^2} \Hchi_\mu(C_x) \Hchi_{\bar{\nu}}(C_x)
   \ket{\Hchi_\mu}\bra{\Hchi_\nu}.
\end{eqnarray*}
Finally, using $\subhyp{N}{G}^\perp\cong(\subhyp{K}{G}/\subhyp{N}{G})^*\cong (\subhyp{K/N}{G/N})^*$ (lemma \ref{lemma:Roth}),  we conclude that the probability of measuring $\Hchi_\mu\in\subhyp{N}{G}^\perp$ is
\begin{align}
\Pr(\Hchi_\mu) &= \w{\Hchi_\mu} \sum_{C_x\subhyp{N}{G} \in (\subhyp{K}{G}/\subhyp{N}{G})}
    \frac{\w{C_x\overline{N}_G}^2}{\varw{ (\subhyp{K}{G}/\subhyp{N}{G})}^2}
    \Hchi_\mu(C_x) \overline{\Hchi_{\mu}}(C_x)\notag\\
    &= \w{\Hchi_\mu} \sum_{C_{xN} \in ({\subhyp{K/N}{G/N}})}
    \frac{\w{C_{xN}}^2}{\varw{\subhyp{K/N}{G/N}}^2}
    \Hchi_\mu(C_{xN}) \overline{\Hchi_{\mu}}(C_{xN}),\label{eq:AKR-prob2}
\end{align}
which is analogous to what we saw in equation~(\ref{eq:AKR-prob1}).

With this in hand, we can now prove the following result.
\begin{lemma}[\textbf{HSHP over \emph{p}-groups}]\label{lma:p-group}Let $G$ be a $p$-group, for some prime $p$. Suppose that we are given a hiding function $\overline{f} : \Conj{G} \rightarrow \{0,1\}^*$. If we can efficiently compute the QFT for $\overline{G}$ and we can efficiently compute the kernels of irreps when restricted to subgroups, then there is an efficient quantum algorithm for the CC-HSHP.
\end{lemma}
\begin{proof}
We follow a similar approach to before, applying AKR to subhypergroups
$\Conj{K}_G$ (starting with $K=G$) until we measure an irrep $\nu$ with
kernel smaller than $K$ and then recursing on $\subhyp{J}{G} \le \subhyp{K}{G}$, where $J =
\Ker(\nu|_K)$. By assumption, we can compute $\Ker(\nu|_K)$ efficiently. If we
fail to find such a $\nu$ in polynomially many samples, then we can conclude
that $K$ is the hidden subgroup with high probability.

By lemma \ref{lma:subhypergroup-qft} and the notes above it, our assumptions
imply that we have an oracle and an efficient QFT for each subproblem. To
implement AKR, we also need the ability to prepare a uniform superposition over
$K$. This was implemented earlier using the inverse Fourier transform. In this
case, that would require us to prepare a complicated coset state in
$\Repr{G}$. However, we can instead just prepare the superposition directly
by the result of Watrous \cite{Watrous_solvable_groups}.\footnote{This result
holds for black-box groups, so we only need to assume that we know the kernel
of each irrep not that we can intersect the kernels of arbitrary irreps.}

Finally, it remains to prove that we have a good probability (e.g., at least
$1/2$) of measuring a nontrivial irrep or, equivalently, that the probability
of measuring the trivial irrep is not too large (e.g., at most $1/2$). As
before, this amounts to putting a bound on $\sum_{C_{xN} \in(\subhyp{K/N}{G/N})}
(\w{C_{xN}} / \varw{\subhyp{K/N}{G/N}})^2$, this time by equation~(\ref{eq:AKR-prob2}). By
the same argument as before, this will hold if we can show that $\w{C_x\subhyp{N}{G}} /
\varw{\Conj{K/N}}$ is bounded by a constant less than 1.  Earlier, we showed this
by using the fact that the size of a conjugacy class divides the size of the
group.  Unfortunately, that does not help us here because we are comparing
$\w{C_x\subhyp{N}{G}}=|C_{xN}|$ not to the size of $G/N$ but to the size of the subgroup $K/N$. These
two need not be related by even a constant factor. In extreme cases, $K/N$ may
contain only one group element that is not in $C_{xN}$
\cite{math-exchange-one-non-identity-element}.

Instead, we will use properties of $p$ groups. Since $G$ is a $p$-group, so is
$K/N$, and the size of $K/N$ must be $p^k$ for some $k$. Since the size of a
conjugacy class divides the size of a group (this time $G/N$), the size of
$C_{xN}$ in $K/N$ must be $p^j$ for some $j$. Now, we must have $j \le k$ since $C_{xN} \subset \subhyp{K/N}{G/N}$. However, we cannot have $j=k$ unless $N = K$, which we have assumed
is not the case, since $K/N$ must contain at least two classes (one identity
and one non-identity). Hence, we can conclude that 
$\w{C_{xN}}=\abs{C_{xN}}$ is smaller than
$\varw{\subhyp{K/N}{G/N}}=|K/N|$ by at least a factor of $p \ge 2$. We conclude as before that the
probability of measuring the trivial irrep is at most $1/2$.
\end{proof}
\begin{theorem}[\textbf{CC-HSHP Is Easy, III}]
\label{thm:easy3}
Let $G$ be a nilpotent group. Suppose that we are given a hiding function
$\overline{f} : \Conj{G} \rightarrow \{0,1\}^*$. If we can efficiently compute
the QFT for $G$ and we can efficiently compute the kernels of irreps when
restricted to subgroups, then there is an efficient quantum algorithm for the
CC-HSHP.
\end{theorem}
\begin{proof}
A nilpotent group is a direct product of $p$-groups for different primes
\cite{Dummit_abstract_algebra}. This means that any subgroup must be a direct
product of subgroups, one in each of the $p$-groups.\footnote{This is, for
example, an immediate consequence of Goursat's Lemma \cite{Lang_algebra} (since
a quotient of a $p$-group and a quotient of a $q$-group, with $q \not= p$, can
only be isomorphic if they are both trivial groups).} Hence, it suffices to
solve the CC-HNSP in each of these $p$-groups.
\end{proof}
Finally, we can apply this result to the HSNP if we are given an oracle that
can be converted into a CC-HSHP oracle.
\begin{corollary}[\textbf{HNSP is easy, III}] Let $G$ be a nilpotent group. Suppose that we are given a hiding function $f :
G \rightarrow H$ that is a homomorphism. If we can efficiently implement the QFT
for $G$, compute with conjugacy classes of $G$ and $H$, and compute  kernels
of irreps when restricted to subgroups, then there is an efficient quantum
algorithm for the HNSP.
\end{corollary}
We note that the class of nilpotent groups includes the Heisenberg group,
which, as we saw in Example~\ref{ex:non-convergence}, is a case where the
original AKR algorithm does not find the hidden normal subgroup with
polynomially many samples. Our last algorithm, however, solves the problem
efficiently in this case.

Our assumption in the results above that we can efficiently compute the kernel
of any irrep of $G$ when restricted to a normal subgroup may seem weaker than
the assumption of HRT that we can intersect the kernels of irreps (which HRT
does in order to find $N$). However, it is not hard to see that the two
assumptions are in fact equivalent. Certainly, if we can intersect kernels of
irreps, then we can use that to find $\Ker(\nu|_K)$ since the latter equals
$\Ker \nu \cap K$ and $K$ itself is the intersection of kernels of irreps (by
induction). Likewise, if we are given a list of irreps, one way to compute the
intersection of their kernels is to order them and then repeatedly compute the
kernel of the next irrep restricted to the intersection of the kernels of those
previous.

Hence, both HRT and our last algorithm require the same assumptions in order
for the classical parts of the algorithms to be efficient. Nonetheless, the
quantum parts of the two algorithms are fundamentally different, so our
algorithm demonstrates a new way in which this important problem can be solved
efficiently by quantum computers.

\subsection{Further results and open problems}
\label{sect:Open Problems}

We finish this section with some discussion on whether this last result can be extended further. The most natural next step beyond nilpotent groups would be to show that the algorithm works for super-solvable groups. We start with a positive example in that direction.

\begin{example}[\textbf{Dihedral Groups}]
As we saw above, the algorithm will work correctly provided that the
probability of measuring the trivial irrep is not too close to 1. By
equation~(\ref{eq:AKR-prob2}), this is given by $\sum_{C_{xN} \in(\subhyp{K/N}{G/N})}
(\w{C_{xN}} / \varw{\subhyp{K/N}{G/N}})^2$.

Without loss of generality, we may assume $N=\{e\}$ by instead looking at the
group $G/N$. For a dihedral group, such a quotient is either dihedral or
abelian. Since abelian groups are (trivially) nilpotent, we know the algorithm
works in that case already.

By our earlier arguments, the probability $\sum_{C \in (\Conj{K})_G} (\w{C} /
\varw{\Conj{K}})^2$ is bounded by a constant below one provided that the
fractional weights $\w{C} / \varw{\Conj{K}}$ are bounded by a constant below one. In
other words, our only worry is that there is a normal subgroup $K$ containing a
conjugacy class that is nearly as large as $K$.

Let us consider the dihedral group of order $2n$, generated by a rotation $a$
of order $n$ and a reflection $r$ of over 2. Most of the normal subgroups are
contained in the cyclic subgroup $\langle a \rangle$. These are subgroups of
the form $\langle a^d \rangle$ with $d$ dividing $n$. Every conjugacy class in
this subgroup contains either 1 or 2 elements (since $r^{-1} a^j r = a^{-j}$
and hence $r^{-1} a^{-j} r = a^j$). Since a nontrivial normal subgroup cannot
consist of one conjugacy class, the worst case would be when $K$ has 3 elements
and contains a conjugacy class with 2 elements. In that case, the probability
of measuring the trivial irrep could only be as large as $2/3$, which still a
constant (independent of $n$) less than one\footnote{Note that if $\w{C_{x}} / \varw{\subhyp{K}{G}}\leq c$, then character orthogonality  (\ref{eq:Character Orthogonality Subhypergroups}) lets us bound the probability of measuring  a character $\Hchi_\mu\in\mathcal{T}^*$  (\ref{eq:AKR-prob2}) since $\mathrm{Pr}(\Hchi_\mu)/\ws{\Hchi_\mu}= \sum_{C_{x} \in\subhyp{K}{G}}
(\w{C_{x}} / \varw{\subhyp{K}{G}})^2 |\mathcal{\Hchi}_\mu(C_x)|^2\leq c (\sum_{C_{x} \in\subhyp{K}{G}}
 \w{C_{x}} / \varw{\subhyp{K}{G}}|\mathcal{\Hchi}_\mu(C_x)|^2)=c/\ws{\Hchi_\mu\subhyp{K}{G}^\perp}$. Specifically, for any invertible character we get $\mathrm{Pr}(\Hchi_\mu)\leq c$.}.

If $n$ is odd, then any normal subgroup $K$ containing $r$ is the whole group,
and the largest conjugacy class contains every $a^j r$ for $j \in \Integer_n$,
which is half the elements, so we get a bound of $1/2$ in that case. If $n$ is
even, then there are two more normal subgroups, one containing $a^{2j} r$ for
each $j$ and one containing $a^{2j+1} r$, but both also contain all rotations
of the form $a^{2j}$, so at least half of the elements in these subgroups
are contained in 1--2 element conjugacy classes, and once again we get a bound
of $1/2$.

All together, this shows that the probability of measuring the trivial irrep is
at most $2/3$ for the dihedral groups, so the algorithm will succeed with high
probability.
\end{example}
On the other hand, we also have a negative example.

\begin{example}[\textbf{Super-solvable group} \cite{math-exchange-one-non-identity-element}]
We will consider the group of simple affine transformations over $\Integer_p$.
These are transformations of the form $x \mapsto ax + b$ for some $a \in
\Integer_p^\times$ and $b \in \Integer_p$, which we denote by $(a,b)$. These
form a group under composition.  In particular, applying $(a,b)$ and then
$(c,d)$ gives $acx + bc + d$, which shows that $(c,d) \cdot (a,b) = (ac, bc +
d)$. A simple calculation shows the formula for the commutator \[[(a,b),(c,d)] = (1,
c^{-1}(1-a^{-1})d - a^{-1}(1 - c^{-1})b).\] This implies that the commutator
subgroup $[G,G]$ is contained in the set $\set{(1,b)}{b \in \Integer_p}$. On
the other hand, taking $a = 1$, $c = 2^{-1}$, and $d = 0$ in this formula gives
the result $(1,b)$, so $[G,G]$ must contain all the elements of this set. If
we mod out $[G,G]$, then we are left with the  abelian group
$\Integer_p^\times$. We have proven that the group is super-solvable.

On the other hand, for any element $(1,d)$, taking $a=2^{-1}, c=1$ and $b=0$ in the
formula above gives the result $(1,-d)$, which is not the identity $(1,0)$.
This means that the group has a trivial center, and thus, it cannot be
nilpotent.

Another simple calculation shows that conjugating $(1,b)$ by $(c,0)$ gives us
$(1,c^{-1}b)$. Hence, the conjugacy class of $(1,b)$ with $b\neq 0$ contains every $(1,b')$ with $b' \not= 0$. This is all of the subgroup $[G,G]$ except for the identity
element $(1,0)$. Hence, once we have $K=[G,G]$, we can see by
equation~(\ref{eq:AKR-prob2}) with $N=\{e\}$ that the algorithm will get the
trivial irrep with high probability, so we can see that the algorithm will fail
to find $N$ in this case.
\end{example}

Put together, these results show that our last algorithm works for some
non-nilpotent, super-solvable groups (like the dihedral groups\footnote{It is
super-solvable since it is a semi-direct product of abelian   groups, and it is
easy to check that it is not nilpotent unless $n$ is a power of 2.}), but not
all super-solvable groups since it fails on the affine linear group.
Determining exactly which super-solvable groups the algorithm does succeed on
is an open problem.

\section*{Acknowledgements}

JBV thanks Maarten Van den Nest for motivating them into investigating the classical simulability of nonabelian quantum Fourier transforms \cite{bermejo2011classical}. We thank Aram Harrow,  Oliver Buerschaper, Nicolas Delfosse,  and Robert Raussendorf, Martin Schwarz and Norman J.\ Wildberger for  comments on the manuscript; R.\ Jozsa, for providing references; and the MIT Center for Theoretical Physics and the Max Planck Institute of Quantum Optics for hosting us in 2013 and 2014. We acknowledge funding from  the ENB QCCC program, SIQS, and the NSF grant CCF-1111382.

\addcontentsline{toc}{section}{References}
\bibliographystyle{utphys}
\bibliography{database}

\appendix

\section{Proof of theorem \ref{thm:Evolution of Stabilizer States}, part II}\label{app:A}

In this appendix, we derive equations (\ref{eq:Pauli gates are Clifford}-\ref{eq:QFTs are Clifford}) finishing the proof of theorem \ref{thm:Evolution of Stabilizer States}.  We treat the different types of normalizer gates separately below.
\begin{enumerate}
\item \textbf{Automorphism Gates.}
It follows from the definition in section \ref{sect:Hypergroups} that any hypergroup automorphism $\alpha$ fulfills  $n_{\alpha(a),\alpha(b)}^{\alpha(c)}=n_{a,b}^{c}$  and $\w{\alpha(a)}=\w{a}$ for all $a,b,c\in \mathcal{T}$. Combining these properties with  (\ref{eq:Pauli operators DEFINITION}) we derive (\ref{eq:Automorphism gates are Clifford}).
\item \textbf{Quadratic phase gates.} The RHS\ of  (\ref{eq:Quadratic Phase gates are Clifford}) follows because $D_\xi$ is diagonal, hence, commutes with $Z_\mathcal{T}(\mathcal{X}_\mu)$.

The LHS can be derived by explicitly evaluating the action of $D_\xi \PX{\mathcal{T}}(a)D_\xi^\dagger$  on basis states in $\mathsf{B}_\mathcal{T}=\{\ket{b},b\in\mathcal{T}\}$  using that, for any $c, c'\in ab$ and any quadratic function $\xi$ with associated bicharacter $B$, the following identities holds:
 \begin{align*}
  (\mathrm{i})\quad  \xi(c)=\xi(c')=\xi(ab),  \quad  (\mathrm{ii}) \quad \xi(c)= \xi(a)\xi(b)  B(a,b), \quad (\mathrm{iii}) \quad  B(a,b)= \Hchi_{\beta(a)}(b)
 \end{align*}
  for some  homomorphism $\beta$ from $\mathcal{T}$ onto $\mathcal{T}_\mathrm{inv}^*$. Above, (i) follows from the triangle inequality and given properties, namely, $\xi(ab)=\sum_{c\in ab} n_{ab}^{c}\xi(c)$,  $|\xi(c)|=1$, and $\sum_{c} n_{ab}^c{=}1$; (ii) follows from (i) and the definition of quadratic function; and last, the normal form for bicharacters (iii) can be obtained by extrapolating the group-setting argument given in \cite{VDNest_12_QFTs}, lemma 5. 
\item \textbf{Quantum Fourier transforms.} We derive (\ref{eq:QFTs are Clifford}) by explicitly computing the action of Pauli operators on the states  $\Fourier{\mathcal{T}}^\dagger\ket{\mathcal{X}_\mu}$ (which form a basis)  using  (\ref{eq:Quantum Fourier Transform over Hypergroup T}): 
  \begin{align}
  \PX{\mathcal{T}}(a)\Fourier{\mathcal{T}}^\dagger\ket{\mathcal{X}_\mu}&=\sum_{b\in \mathcal{T}} \sqrt{\tfrac{\w{b}\ws{\overline{\Hchi_{\mu}}}}{\varpi_{\mathcal{T}}}} \overline{\Hchi_{\mu}}(b)\PX{\mathcal{T}}(a)\ket{b} \\&\stackrel{(\ref{eq:Pauli operators DEFINITION})}{=}\sum_{b\in \mathcal{T}}
  \sqrt{\tfrac{\w{b}\ws{\overline{\Hchi_{\mu}}}}{\varpi_{\mathcal{T}}}} \overline{\Hchi_{\mu}}(b)\left(\sum_{c} \sqrt{\tfrac{\w{b}}{\w{c}}}\,  n_{a,b}^{c} \ket{c}\right)\notag \\
  &\stackrel{(\ref{eq:Reversibility Property})}{=}\sum_{c\in \mathcal{T}} \sqrt{\tfrac{\w{c}\ws{\overline{\Hchi_{\mu}}}}{\varpi_{\mathcal{T}}}} \left(\sum_{b}   n_{{\overline{a}},c}^{b} \overline{\Hchi_{\mu}}(b) \right)\ket{c}=\sum_{c\in \mathcal{T}} \sqrt{\tfrac{\w{c}\ws{\overline{\Hchi_{\mu}}}}{\varpi_{\mathcal{T}}}} \overline{\Hchi_{\mu}}({\overline{a} } ) \overline{\Hchi_{\mu}}(c) \ket{c}\notag\\  &= \mathcal{X}_{{\mu}}({{a} } )\Fourier{\mathcal{T}}^\dagger\ket{\mathcal{X}_\mu}= \Fourier{\mathcal{T}}^\dagger\PZ{{\mathcal{T}^*}}(a)\ket{\Hchi_\mu}
  \end{align} 
 \begin{align}
 \PZ{\mathcal{T}}(\mathcal{X}_\mu)\Fourier{\mathcal{T}}^\dagger\ket{\mathcal{X}_\nu}&\notag\stackrel{(\ref{eq:Pauli operators DEFINITION})}{=}\sum_{b\in \mathcal{T}} \sqrt{\frac{\w{b} \ws{\overline{\Hchi_{\nu}}}}{\varpi_{\mathcal{T}}}} \, \mathcal{X}_{\mu}(b)\overline{\Hchi_\nu}(b)\ket{b}
\\&= \sum_{b\in \mathcal{T}} \sqrt{\frac{\w{b} \ws{\overline{\Hchi_{\nu}}}}{\varpi_{\mathcal{T}}}} \,  \left(\sum_{\overline{\Hchi_\gamma}\in{\mathcal{T}^*}}m_{\mu\overline{\nu}}^{\overline{\gamma}} \mathcal{X}_{{\overline{\gamma}}}(b)\right)\ket{b}\\
&= \sum_{\overline{\Hchi_\gamma}\in{\mathcal{T}^*}} \sqrt{\frac{\ws{\overline{\Hchi_{\nu}}}}{\ws{\Hchi_{\overline{\gamma}}} }}  m_{\mu\overline{\nu}}^{\overline{\gamma}}\left(\sum_{b\in \mathcal{T}} \sqrt{\frac{\w{b} \ws{\Hchi_{\overline{\gamma}}}}{\varpi_{\mathcal{T}}}} \,  \mathcal{X}_{{\overline{\gamma}}}(b)\ket{b}\right)\notag\\
& =\Fourier{\mathcal{T}}^\dagger \sum_{\overline{\Hchi_\gamma}\in{\mathcal{T}^*}} \sqrt{\frac{\ws{\overline{\Hchi_{\nu}}}}{\ws{\Hchi_{\overline{\gamma}}}}}  m_{\mu\overline{\nu}}^{\overline{\gamma}}\ket{\mathcal{X}_{{\gamma}}}= \Fourier{\mathcal{T}}^\dagger \PX{{\mathcal{T}^*}}(\overline{\Hchi_{\mu}}) \ket{\Hchi_\nu},
 \end{align} 
where we used   $\ws{\Hchi_{\overline{\gamma}}}=\ws{\Hchi_{\gamma}}$ and $m_{\mu\overline{\nu}}^{\overline{\gamma}}=m_{\overline{\mu}\nu}^\gamma$ from section \ref{sect:Hypergroups}. The analogous statement for partial QFTs follows straightforwardly using that character hypergroup of $\mathcal{T}_1\times \cdots \times \mathcal{T}_m$ is $\mathcal{T}_1^*\times \cdots \times \mathcal{T}_m^*$ \cite{BloomHeyer95_Harmonic_analysis} and the tensor-product structure of Pauli operators (section \ref{sect:Pauli Operators}).
\item\textbf{Pauli gates.} We can use  (\ref{eq:QFTs are Clifford}) to get   $Z_\mathcal{T}(\mathcal{X}_\varsigma)\PX{\mathcal{T}}(a) Z_\mathcal{T}(\mathcal{X}_\varsigma)^\dagger = \mathcal{X}_{{\varsigma}}(a) \PX{\mathcal{T}}(a)$ and  $Z_\mathcal{T}(\mathcal{X}_\varsigma)\PZ{\mathcal{T}}(\mathcal{X_\mu})Z_\mathcal{T}(\mathcal{X}_\varsigma)^\dagger=\PZ{\mathcal{T}}(\mathcal{X_\mu})$  since invertible characters are quadratic functions with trivial $\mathcal{B}$ and $\beta$. Moreover, we can apply   (\ref{eq:QFTs are Clifford}) and repeat the argument in the character basis, obtaining $X_\mathcal{T}(s)\PX{\mathcal{T}}(a) X_\mathcal{T}(s)^\dagger=\PX{\mathcal{T}}(a)$,  $X_\mathcal{T}(s)\PZ{\mathcal{T}}(\mathcal{X_\mu}) X_\mathcal{T}(s)^\dagger = \mathcal{X_\mu}(\overline{s})\PZ{\mathcal{T}}(\mathcal{X_\mu})$. Equation (\ref{eq:Pauli gates are Clifford}) is derived combining these expressions.
\end{enumerate}

\section{Quadratic functions}\label{app:Quadratic functions}

We prove that the functions $\xi_{i}$, $\xi_{j}$, $\xi_{k}$ and $\xi$   defined in section \ref{sect:Quaternionic circuits} are quadratic. The quadraticity of $\xi_{x}$, with $x=i,j,k$, follows from the fact that the function  can be obtained by composing the quotient map $\overline{Q}_8\rightarrow \overline{Q}_8/\{\pm 1, \pm x\} \cong \Integers_2$, with the isomorphism $\overline{Q}_8/\langle x\rangle \rightarrow  \Integers_2$ and the map $\Integers_2\rightarrow \C:a \rightarrow i^a$; since the latter is a quadratic function of $\Integers_2$ \cite{VDNest_12_QFTs}, it follows easily that $\xi_x$ is a quadratic function of $\overline{Q}_8$. Note that in this derivation we implicitly use that  $\{\pm 1, \pm x\}$ is a subhypergroup of $\overline{Q}_8$  \cite{Roth75_Character_Conjugacy_Hypergroups}, that the quotient $\overline{Q}_8/S$ is an abelian hypergroup for any subhypergroup $S$, and that the quotient map $\overline{Q}_8\rightarrow \overline{Q}_8/S$ is a  hypergroup homomorphism \cite{Roth75_Character_Conjugacy_Hypergroups}.

To show that $\xi:\overline{Q}_8\times \overline{Q}_8\rightarrow\C$  is quadratic, we use the fact, prove below, that the function $B(C_x,C_y):=f_{C_x}(C_y)$ is a symmetric bi-character of $\overline{Q}_8$. Given that property as a  promise and recalling that $\xi((C_x,C_y))=B(C_x,C_y)$ (by definition), we can see that
\begin{align}
\xi\left((C_a,C_b)\cdot (C_c,C_d)\right)&=B(C_{a}C_c,C_b C_d)=B(C_a,C_b C_d)B(C_{c},C_b C_d))\notag\\
&=B(C_a,C_b)B(C_a,C_d)B(C_c,C_b)B(C_c,C_d)\notag\\
&=\xi\left((C_a,C_b)\right)\xi\left( (C_c,C_d)\right)B(C_a,C_d)B(C_c,C_b)\notag\\
&=\xi\left((C_a,C_b)\right)\xi\left( (C_c,C_d)\right)B'\left((C_a,C_b),(C_c,C_d)\right),
\end{align}
where we define $B'\left((C_a,C_b),(C_c,C_d)\right)=B(C_{a},C_d)B(C_{c},C_b)$. The latter is easily seen to be a bi-character of $\overline{Q}_8\times\overline{Q}_8$, so that $\xi$ is indeed quadratic.

It remains to show that  $B(C_x,C_y)$ is a symmetric bi-character. To see this, note,  that both the quotient hypergroup  $\overline{Q}_8/\{\pm 1\}$ and the subhypergroup of linear characters $\widehat{Q_8}_\ell$ of $Q_8$ are isomorphic to the Klein four group $\Integers_2\times \Integers_2$. Observe next that the map $C_x\rightarrow f_{C_x}$ is a homomorphism $\overline{Q}_8\rightarrow \widehat{Q_8}_\ell$, as it can be obtained composing the quotient map $\overline{Q}_8\rightarrow \overline{Q}_8/\{\pm 1\}$ with a chain of isomorphisms $\overline{Q}_8/\{\pm 1\} \rightarrow \Integers_2\times \Integers_2 \rightarrow \widehat{Q_8}_\ell$.  This latter fact implies that $B(C_x,C_y)=f_{C_x}(C_y)$ is a character in both arguments, hence, a bi-character. Finally, it is routine to check that $B(C_x,C_y)=B(C_y,C_x)$   by explicit evaluation, which completes the proof.

\section{Efficient vs.\ doubly efficient computable hypergroups}\label{app:Discret Log}

We  give an example of an efficient computable abelian hypergroup that cannot be   doubly efficiently computable unless we are given the ability to compute discrete logarithms over $\Integers_p^\times$, a problem that is believed to be hard for classical computers and is the basis of the Diffie-Hellman public-key cryptosystem \cite{DiffieHellman}. (Quantum computers can solve this problem using Shor's discrete-log algorithm \cite{Shor}. This problem reduces to the so-called hidden subgroup problem over $\Integers_{p-1}^2$ \cite{Nielsen02UniversalSimulations} for a certain hiding function $f$, which defines a group homomorphism from $\Integers_{p-1}^2$ to $\Integers_p^\times$ \cite{BermejoLinVdN13_BlackBox_Normalizers}.)

Considering now the group $\mathcal{T}=\Integers_{p-1}^2\times\Integers_p^\times$, which is manifestly efficiently computable following our definition, we can define an efficiently computable group automorphism $\alpha:\mathcal{T}\rightarrow \mathcal{T}: (m,x)\rightarrow(m,f(m)x)$.  We show that  $\mathcal{T}$ cannot be doubly efficiently computable unless the initial hidden subgroup problem and, hence, the discrete logarithm problem, can be solved in probabilistic polynomial time (which, up to date, is not possible).

First, we show that, if we are able to compute\footnote{We are implicitly assuming that there are efficient unique classical encodings for representing the characters of $\Integers_p^\times$, which is a strong yet \emph{weaker} assumption that $\Integers_p^\times$ being doubly efficiently computable.} $\alpha^*$, we must also be able to compute $f^*:\widehat{\Z}_p^\times \rightarrow \widehat{ \Z}_{p-1}^2$ (the dual of $f$, which is defined analogously to $\alpha^*$), since for any $\Hchi_{\mu,\nu}:=\Hchi_\mu\otimes\Hchi_\nu$ we have
\begin{equation}\notag
\Hchi_{\alpha^*(\mu,\nu)}(m,x)=(\Hchi_\mu\otimes\Hchi_\nu)(m,f(m)x)=\Hchi_\mu(m)\Hchi_\nu(f(m))\Hchi_\nu(x) = \Hchi_\mu(m){\Hchi_{f^*(\nu)}}(m) \Hchi_\nu(x),
\end{equation}
consequently, $\Hchi_{\alpha^*(\mu,\nu)} = \left(\Hchi_{\mu}\cdot \Hchi_{f^*(\nu)}\right) \otimes \Hchi_{\nu}$. Hence, if we can evaluate $\alpha^*$ on any character  $\Hchi_1\otimes \Hchi_{\mu}$, then we can determine  $\Hchi_{f^*(\nu)}$, the value of $f^*$ on $\nu$, for any $\Hchi_{\nu}$. If we now evaluate $f^*(\mu_i)$ on all  elements of a $O(\log p)$-sized randomly-obtained generating set $\{\Hchi_{\mu_i}\}$ of $\widehat{\Z}_{p}^\times$  and use existing classical algorithms \cite{BermejoVega_12_GKTheorem} to solve the system of equations $\{[f^*(\mu_i)](x)=\Hchi_{\mu_i}(f(x))=1, x\in \Integers_{p-1}^2 \}$, whose solutions are those $x$ for which $f(x)=e$, we have found (in these $x$'s) generators of the hidden subgroup. This finishes the reduction.

\section{Implementing normalizer circuits over $\Conj{G}$}
\label{app:CC implementation details}

In this section, we present more details on how to efficiently implement normalizer circuits over the hypergroup $\Conj{G}$ when we are working in the Hilbert space $\mathcal{H}_G = \set{\ket{g}}{g \in G}$ labeled by elements of the group. As described in section \ref{sect:Circuits nonabelian group}, normalizer circuits over $\Conj{G}$ can be thought of as operating entirely within the subspace $\mathcal{I}_G \le \mathcal{H}_G$ of conjugation invariant wavefunctions; however, we will describe these operations in this section in terms of how they operate on the entire Hilbert space. In section~\ref{sect:character basis operations}, we discuss operations applied in the character basis, and in section~\ref{sect:class basis operations}, we discuss the same in the conjugacy class basis.

\subsection{Working in the character basis}
\label{sect:character basis operations}

Normalizer circuits allow of the following operations to be performed in the character class basis: preparation of initial states; Pauli, automorphism, and quadratic phase gates; and measurement of final states. It should be easy to understand how each of these could be implemented efficiently if we worked in a basis $\set{\ket{\mu}}{\mu \in \mathrm{Irr}(G)}$ of irrep labels. However, the Hilbert space $\mathcal{H}_G$ is only naturally labeled by group elements, which is why, in section~\ref{sect:Circuits nonabelian group}, we defined the character basis states $\set{\ket{\Hchi_\mu}}{\Hchi_\mu \in \Repr{G}}$ in the element basis.

Below, we will describe how to implement an isometry $\ket{\Hchi_\mu} \stackrel{\tau}{\mapsto} \ket{\mu}$ and its inverse, using a readily available choice for the basis $\set{\ket{\mu}}{\mu\in\mathrm{Irr}(G)}$. It should then be clear that we can implement each of the above gates by applying $\tau$, performing the operation in the irrep label basis, and then applying $\tau^{-1}$. To prepare an initial state $\ket{\Hchi_\mu}$, we prepare $\ket{\mu}$ in the irrep label basis and then apply $\tau^{-1}$. Finally, to measure in the character basis, we apply $\tau$ and then read the irrep label.

Our definition of the irrep label basis $\set{\ket{\mu}}{\mu\in\mathrm{Irr}(G)}$ comes from the definition of the QFT over the group $G$. Recall  that the QFT over any finite group $G$ \cite{childs_vandam_10_qu_algorithms_algebraic_problems}, denoted $\mathcal{F}_G$, is a unitary gate that sends an element state $\ket{g}$, for any $g\in G$, to a weighted superposition $\abs{G}^{-1/2} \sum_{\mu\in\mathrm{Irr}(G)} d_\mu \ket{\mu, \mu(g)}$, where  $\ket{\mu}$ is a state  that labels the irrep $\mu$ and $\ket{\mu(g)}$ is a $d_\mu^2$ dimensional state defined via
\begin{equation}\label{eq:QFT over NAB group}
\ket{\mu(g)}=\left(\mu(g)\otimes I_{d_\mu}\right) \sum_{i=1}^{d_\mu} \frac{\ket{i,i}}{\sqrt{d_\mu}}= \sum_{i,j=1}^{d_\mu} \frac{[\mu(g)]_{i,j}}{\sqrt{d_\mu}}\ket{i,j}.
\end{equation}
This transformation $\Fourier{G}$ has been extensively studied in the HSP literature and efficient quantum implementations over many groups are currently known (including  the symmetric group,  wreath products of  polynomial-sized groups and metabelian groups \cite{childs_vandam_10_qu_algorithms_algebraic_problems}). 

To see how we can use this, let's look at what $\Fourier{G}$ does to a character class state. In (\ref{eq:Bases Nonabelian Group}), we defined the state $\ket{C_x}$, when living inside the Hilbert space $\mathcal{H}_G$, to be a uniform superposition over the elements in the class $C_x$. If we apply $\mathcal{F}_G$ to this state, the result is
\begin{eqnarray*}
\mathcal{F}_G \ket{C_x}
 &=& \frac{1}{\sqrt{\abs{C_x}}} \sum_{g \in C_x} \mathcal{F}_G \ket{g} \\
 &=& \frac{1}{\sqrt{\abs{C_x}}} \sum_{g \in C_x} \frac{1}{\sqrt{\abs{G}}} \sum_{\mu \in \mathrm{Irr}(G)} d_\mu \ket{\mu} \otimes \sum_{i,j=1}^{d_\mu} \frac{[\mu(g)]_{i,j}}{\sqrt{d_\mu}} \ket{i,j} \\
 &=& \frac{1}{\sqrt{\abs{C_x} \abs{G}}} \sum_{\mu \in \mathrm{Irr}(G)} d_\mu \ket{\mu} \otimes \sum_{i,j=1}^{d_\mu} \frac{[\sum_{g \in C_x} \mu(g)]_{i,j}}{\sqrt{d_\mu}} \ket{i,j}.
\end{eqnarray*}
To simplify further, we need to better understand the sum in the numerator on the right.

The sum $\sum_{g \in C_x} \mu(g)$ is more easily analyzed if we write it as $(\abs{C_x}/\abs{G}) \sum_{h \in G} \mu(x^h)$: by standard results on orbits of group actions \cite{Lang_algebra}, each $\mu(g)$, for $g \in C_x$, arises the same number of times in the sum $\sum_{h \in G} \mu(x^h)$, which hence must be $\abs{G}/\abs{C_x}$ times for each, so we have $(\abs{C_x}/\abs{G}) \sum_{h \in G} \mu(x^h) = \sum_{g \in C_x} \mu(g)$. The sum $(1/\abs{G})\sum_{h \in G} \mu(x^h)$ may be familiar, as it is well known to be $\tfrac{1}{d_\mu} \chi_\mu(x) I$ \cite{Serre_representation_theory}.\footnote{This is a simple application of Schur's lemma. This sum is a $G$-invariant map $\mathcal{H}_G \rightarrow \mathcal{H}_G$, so it must be a constant times the identity. The constant is easily found by taking the trace of the sum.}

Putting these parts together, we can see that
\begin{eqnarray*}
\mathcal{F}_G \ket{C_x}
 &=& \frac{1}{\sqrt{\abs{C_x} \abs{G}}} \sum_{\mu \in \mathrm{Irr}(G)} d_\mu \ket{\mu} \otimes \sum_{i=1}^{d_\mu} \frac{\abs{C_x} \chi_\mu(x)}{d_\mu \sqrt{d_\mu}} \ket{i,i} \\
 &=& \sqrt{\frac{\abs{C_x}}{\abs{G}}} \sum_{\mu \in \mathrm{Irr}(G)} d_\mu \frac{\chi_\mu(x)}{d_\mu} \left(\frac{1}{\sqrt{d_\mu}} \sum_{i=1}^{d_\mu} \ket{\mu,i,i} \right),
\end{eqnarray*}
which is rewritten in our usual hypergroup notation as
\begin{equation}\label{eq:group Fourier on class}
\mathcal{F}_G \ket{C_x} = \sum_{\Hchi_\mu \in \Repr{G}} \sqrt{\frac{w_{C_x} w_{\Hchi_\mu}}{w_\Conj{G}}} \Hchi_\mu(C_x) \left(\frac{1}{\sqrt{d_\mu}} \sum_{i=1}^{d_\mu} \ket{\mu,i,i} \right).
\end{equation}
This precisely mirrors the definition of $\mathcal{F}_\Conj{G}$ with $\ket{\Hchi_\mu}$ replaced by $d_\mu^{-1/2} \sum_{i=1}^{d_\mu} \ket{\mu,i,i}$. We will denote the latter state below by $\ket{\mu_\mathrm{diag}}$. Thus, it follows by (\ref{eq:Character Orthogonality}) that $\Fourier{G} \ket{\Hchi_\mu} = \ket{\mu_\mathrm{diag}}$.

To implement the operation $\tau$, we apply $\Fourier{G}$ and then \emph{carefully} discard the matrix index registers.\footnote{In full detail, we do the following. First, apply the map $\ket{i,j} \mapsto \ket{i,j-i}$, which gives $\ket{i,0}$ when applied to $\ket{i,i}$. Next, write down $d_\mu$ in a new register and then invoke the Fourier transform over $\Integer_{d_\mu}$ on the first index register. The result of this will always be $\ket{0}$, so after uncomputing $d_\mu$, we are left with the state $\ket{0,0}$ in the index registers regardless of the value of $\mu$. At that point, they are unentangled and can be safely discarded.} By the above discussion, we can see that this maps $\ket{\Hchi_\mu}$ to the state $\ket{\mu}$, so this implements the operation $\tau$ correctly for any conjugation invariant state.

To implement the operation $\tau^{-1}$, we do the above in reverse. Starting with a state $\ket{\mu}$, we adjoin matrix index registers, prepare a uniform superposition over $\ket{1}, \dots, \ket{d_\mu}$ in the first index register using the inverse Fourier transform over the abelian group $\Integer_{d_\mu}$, and then copy the first index register to the second\footnote{Or rather, apply the map $\ket{i,j} \mapsto \ket{i, i+j}$, which gives $\ket{i,i}$ when applied to $\ket{i,0}$.} to get the state $\ket{\mu_\mathrm{diag}}$. Finally, we apply $\Fourier{G}^\dagger$ to get the state $\ket{\Hchi_\mu}$ per the calculations above.

As discussed earlier, the operations $\tau$ and $\tau^{-1}$ are all that we need in order to implement each of the required operations of normalizer circuits over $\Conj{G}$ in the character basis.

\subsection{Working in the character class basis}
\label{sect:class basis operations}

Most of the time, gates applied in the conjugacy class basis arise from operations on the whole group. For example, automorphisms of conjugacy classes often arise from automorphisms of the group. Likewise, Pauli Z operators in the conjugacy class basis are applications of characters, which are defined on the whole group, and Pauli X operators can also be implemented using multiplication in the group. Hence, it remains only discuss how to prepare initial states and measure in the conjugacy class basis.

For this, we need to assume that we can perform certain operations on conjugacy classes, as described in the following definition.

\begin{definition}\label{def:compute with CC}
Let $C_1, \dots, C_m$ be the conjugacy classes of $G$. Consider the following operations for working with conjugacy classes:
\begin{itemize}
\item Given a conjugacy class label $i$, produce the size of this class, $\abs{C_i}$.
\item Given an $x \in G$, produce the pair $(i,j)$, where $x = x_j$ in the class $C_i = \{x_1, \dots, x_t\}$.
\item Given a pair $(i,j)$, produce the element $x_j$ from $C_i$.
\end{itemize}
If each of these operations can be performed efficiently, then we say that we can \emph{compute efficiently with conjugacy classes} of $G$.
\end{definition}

We note that this assumption is trivial for abelian groups since each element is in its own conjugacy class. For some common examples of nonabelian groups, such as the dihedral and Heisenberg groups (and their higher nilpotent generalizations), elements are normally encoded in this manner already, so no additional assumption is actually required. For other common examples like the symmetric group, while elements are not always encoded directly in this manner, it is easy to see how the above calculations can be performed efficiently. In general, while we must formally make this assumption, we are not aware of any group for which these calculations cannot be performed efficiently.

If we can compute efficiently with conjugacy classes of $G$, then we can prepare initial states as follows. Starting with the conjugacy class label $i$ in a register, we first compute the size $\abs{C_i}$ into a new register. Next, we adjoin another new register and invoke the inverse Fourier transform over the abelian group $\Integer_{\abs{C_i}}$. After uncomputing the size $\abs{C_i}$, we are left with the superposition $M^{-1/2} \sum_{j=1}^M \ket{i,j}$, where $M = \abs{C_i}$. Finally, we apply the operation that turns pairs into group elements to get $M^{-1/2} \sum_{j=1}^M \ket{x_j}$, where $x_1, \dots, x_M$ are the elements of $C_i$, which is the desired initial state.

To perform a measurement in the conjugacy class basis, we can do the reverse of how we prepared the initial states in order to produce a conjugacy class label $\ket{i}$ in a register. Alternatively, we can simply measure in the group element basis and then, afterward, compute the conjugacy class of this element. These two approaches will give identical measurement probabilities.

Finally, we note that the two operations just described are the equivalent of the operations $\tau$ and $\tau^{-1}$ from section~\ref{sect:character basis operations} for the conjugacy class basis.\footnote{Indeed, the separation of a group element label into a conjugacy class label and an index label is analogous to how, in the space of irreducible representations, we separate each basis element into an irrep label and matrix index labels. It is frequently assumed that we can separate the latter into different registers whenever convenient, so our assumption that we can do the same for conjugacy classes is only affording the same convenience for the hypergroup $\Conj{G}$ that is often assumed for $\Repr{G}$.} As a result, if we do have gates that can be easily implemented on conjugacy classes but do not extend easily to the whole group, then we can implement these gates in the same manner as in the character basis: apply $\tau$ to convert into a basis of conjugacy class labels $\set{\ket{i}}{C_i \in \Conj{G}}$, apply the gate in this basis, and then apply $\tau^{-1}$ move back to the conjugacy class basis in $\mathcal{H}_G$.

Thus, we can see that the ability to compute efficiently with conjugacy classes of $G$ allows us to fully implement normalizer circuits operations applied in the conjugacy class basis. If we also have an efficient QFT for $G$, then as we saw in the previous section, we can implement normalizer circuits operations applied in the character basis as well. Together, these two assumptions allow us to fully implement normalizer circuits over $\Conj{G}$ when working in the Hilbert space $\mathcal{H}_G$.

\end{document}